\documentclass[10pt]{article}
\usepackage[margin=1.10in]{geometry}
\usepackage[utf8]{inputenc} 
\usepackage{fixltx2e} 
\usepackage{tikz-cd}
\usepackage{mathtools}
\usepackage{enumitem}
\usepackage{stackengine}
\usepackage{wrapfig}
\usepackage{complexity}
\usepackage{balance}
\usepackage{mathdots}
\RequirePackage{latexsym}
\RequirePackage{mathtools}
\RequirePackage{amsfonts}
\usepackage{amssymb}
\usepackage{wasysym}
\RequirePackage{amsthm}
\RequirePackage{array}
\RequirePackage{xspace}
\RequirePackage[noadjust]{cite}
\usepackage{url}
\usepackage{authblk}
\usepackage{graphicx}

\usepackage[framemethod=tikz]{mdframed}
\definecolor{mycolor}{rgb}{0.850, 0.850, 0.850}
\newmdenv[innerlinewidth=0.5pt, roundcorner=4pt,linecolor=mycolor,innerleftmargin=6pt,
innerrightmargin=4pt,innertopmargin=4pt,innerbottommargin=6pt]{mybox}

\newcommand{\suchthat}{\;\ifnum\currentgrouptype=16 \middle\fi|\;}

\newcommand{\ie}{i.e., }
\newcommand{\eg}{e.g., }

\newcommand{\wrt}{w.r.t.\ }
\newcommand{\cfr}{c.f.r.\ }
\newcommand{\wwlog}{w.l.o.g.\ }

\usepackage[textsize=small]{todonotes}

\newcommand{\N}{\mathbf{N}\xspace}

\def\C{{\cal C}}

\def\A{{\cal A}}
\def\J{{\cal J}}
\def\F{{\cal F}}

\newcommand{\figref}[1]{Fig.~\ref{#1}}

\newtheorem{Thm}{Theorem}
\newtheorem*{MainRes}{Main Result}
\newtheorem{Cor}[]{Corollary}
\newtheorem{Lem}[]{Lemma}
\newtheorem{Prop}[]{Proposition}

\theoremstyle{plain}
\newtheorem{Def}{Definition}

\newtheorem{Exa}{Example}

\makeatletter
\providecommand*{\cupdot}{%
  \mathbin{%
    \mathpalette\@cupdot{}%
  }%
}
\newcommand*{\@cupdot}[2]{%
  \ooalign{%
    $\m@th#1\cup$\cr
    \sbox0{$#1\cup$}%
    \dimen@=\ht0 %
    \sbox0{$\m@th#1\cdot$}%
    \advance\dimen@ by -\ht0 %
    \dimen@=.5\dimen@
    \hidewidth\raise\dimen@\box0\hidewidth
  }%
}

\newcounter{procedurealgo}
\usepackage[ruled,vlined,linesnumbered]{algorithm2e}
\usepackage{xcolor}
\newenvironment{procedurealgo}[1][htb]
  {\renewcommand{\algorithmcfname}{Procedure}
  \let\c@algocf\c@procedurealgo
   \begin{algorithm}[#1]%
  }{\end{algorithm}}
\usepackage{subfig}
\SetCommentSty{textsf}
\SetKwRepeat{DoWhile}{do}{while}
\SetAlFnt{\small}
\makeatletter
\newcommand{\removelatexerror}{\let\@latex@error\@gobble}
\makeatother

\let\oldnl\nl
\newcommand{\nonl}{\renewcommand{\nl}{\let\nl\oldnl}}

\usepackage{tikz}
\usepackage{pgfplots}
\usetikzlibrary{fit}
\usetikzlibrary{shapes,shapes.multipart,shapes.arrows}
\usetikzlibrary{decorations,decorations.pathmorphing,decorations.pathreplacing,decorations.markings,decorations.shapes}
\usetikzlibrary{arrows}
\usetikzlibrary{fit,backgrounds}
\usetikzlibrary{plotmarks}
\tikzstyle{squa}=[rectangle,draw,inner sep=4pt,transform shape, minimum size=1.7em]
\tikzstyle{circ}=[circle,draw,inner sep=4pt,transform shape]
\tikzstyle{node}=[circle,draw,inner sep=2pt,transform shape,minimum size=1.6em]

\tikzstyle{none}=[inner sep=0pt]

\tikzstyle{arrow}=[-,draw=Black,postaction={decorate},decoration={markings,mark=at position .5 with {\arrow{>}}},line width=2.000]

\title{Linear-Time Safe-Alternating DFS and SCCs}
\date{}
\author[C. Comin, R. Rizzi]{
	Carlo Comin\footnote{(e-mail: carlo.comin@scuola.istruzione.it)} \;\;\;\; Romeo Rizzi\footnote{(e-mail: romeo.rizzi@univr.it)}
}

\newcount\recurdepth
\begin{document}
\maketitle

\begin{abstract}
  An alternating graph is a directed graph whose vertex set is partitioned into two colour classes, existential and universal.

  This forms the basic arena for well-known models in formal verification, discrete optimal control, and many infinite duration two-player games where
  Player~$\square$ and~his opponent Player~$\ocircle$ alternate in a turn-based sliding of a pebble along the arcs they control.

  We study alternating strongly-connectedness on alternating graphs as a generalization of strongly-connectedness in directed graphs,
  aiming at providing a linear-time decomposition and a sound structural graph characterization.
  For this a novel notion of alternating reachability is introduced: Player~$\square$ attempts to reach vertices without leaving a prescribed subset of the vertices while Player~$\ocircle$ works against.
  This is named \emph{safe-alternating reachability}.
  It is shown that every alternating graph uniquely decomposes into \emph{safe-alternating strongly-connected components}, where Player~$\square$ can visit each vertex within a given component infinitely often without having to ever leave out the component itself.

  Our main result is a linear-time algorithm for computing this alternating graph decomposition.
  Both the underlying graph structures and the algorithm generalize the classical decomposition of a directed graph into strongly-connected components.
  Indeed, the proposed algorithm builds on a non-trivial generalization to alternating graphs of the depth-first search and the strongly-connected components 
  algorithm devised by Tarjan in 1972.

  Our theory has direct applications \eg in solving well-known infinite duration pebble games faster.
  Dinneen and Khoussainov showed in 1999 that deciding a given Update Game costs $O(mn)$ time,
	where $n$ is the number of vertices and $m$ is that of arcs. We solve that task in $\Theta(m+n)$ linear~time.
  In turn the complexity of Explicit McNaughton-M\"uller Games improves from cubic to quadratic.
\end{abstract}

\section{Introduction}\label{sect:intro}
The alternating model of computation originated in~\cite{ChandraS76,Kozen76,Chandra1981}
  as a generalization of nondeterminism in which existential and universal quantifiers alternate along the course of the computation.
Alternating Turing Machines were defined and the corresponding time and space complexity classes were characterized in terms of resource-bounded deterministic machines.
In the complexity landscape generalizing complete computational models to alternation leads more often than not to complexity blowups,
\eg alternating polynomial time equals deterministic polynomial space~\cite{Savitch70,Chandra1981}.

Still in~\cite{Chandra1981} alternation was also inquired by generalizing specific polynomial time computable problems.
One of the classical P-complete problems is the Alternating Graph Accessibility Problem (AGAP)~\cite{Chandra1981,Greenlaw1995}.
We are given a finite directed graph $(V, A)$ whose vertex set $V=V_\square\cup V_\ocircle$ is partitioned into two classes,
 existential $V_\square$ and universal $V_\ocircle$ (\ie an \emph{alternating graph}), plus a source vertex $s$ and target vertex $t$.
The task is to decide whether $t$ is \emph{alternating reachable} from $s$, that can be defined recursively as follows:
either \textit{i)} $s=t$, or \textit{ii)} $s\in V_\square$ and for some outgoing arc $(s,s')\in A$ the target $t$ is alternating reachable from $s'$, or
 \textit{iii)} $s\in V_\ocircle$ and for every outgoing arc $(s,s')\in A$ the target $t$ is alternating reachable from $s'$.
When restricted to only existential vertices, this is equivalent to the Directed Graph Accessibility Problem (GAP),
also known as $(s,t)$-Connectivity (st-Con), which is complete for nondeterministic log-space~\cite{Savitch70}.

Both GAP and AGAP admit linear-time algorithms. In GAP, a depth-first search starting from $s$ works out.
In AGAP a linear-time solution can be achieved by computing the $\square$-\emph{attractor} set $T_\square$ of target vertex $t$,
  \ie defined: \textit{i)} $t \in T_\square$; \textit{ii)} if $x\in V_\square$ has an outgoing arc $(x,y)$ such that $y\in T_\square$, then $x\in T_\square$;
\textit{iii)} if $x\in V_\ocircle$ has all outgoing arcs $(x,y)$ such that $y\in T_\square$, then $x\in T_\square$; \textit{iv)} nothing else in $T_\square$.

Algorithmic problems on alternating generalizations of graphs arise in the literature, \eg in the field of formal probabilistic verification~\cite{ChaKriHen2014}.
The input is a system that exhibits probabilistic behavior and a specification (set of desired behaviors),
where the algorithmic problem is to answer whether the system satisfies the specification~\cite{CourcoubetisY95}.
In probabilistic verification, systems are frequently modeled as a generalization of graphs called Markov decision processes (MDPs).
The generalization is needed to model two different kind of behaviors.
More specifically in MDPs there are two types of vertices, namely, the regular vertices $V_\square$
 where the algorithm chooses which outgoing arc to follow, and the probabilistic vertices $V_P$
 where the outgoing arc is chosen randomly according to some given probability distribution $\delta$;
still, putting aside probabilistic behavior, in MDPs the underlying static graph arena is actually an alternating graph (think of it as if $V_\ocircle=V_P$).
Also notice: (i) a directed graph is a special case of an alternating graph with $V_\ocircle=\emptyset$ (ii) similarly
a directed graph is a special case of an MDP with $V_P = \emptyset$ and (iii) a Markov chain is a special case of an MDP with $V_\square=\emptyset$.

In the literature, algorithmic problems on alternating graphs have already been tackled by relying on classical graph algorithmics employed as interleaved subprocedures, such as the depth-first search (DFS) and the strongly-connected components (SCCs) algorithm~\cite{Tar72}.
One notable instance of that being the \emph{maximal end-component (MEC)} decomposition that lies at the core of many algorithms in probabilistic verification,
generalizing to alternating graphs the problem of decomposing a directed graph into SCCs~\cite{ChaKriHen2014}.
The deterministic time complexity is nevertheless not known to be linear, the fastest deterministic algorithm which is currently known runs in time $O(\min(m^{3/2}, n^2))$~\cite{ChaKriHen2014}, whereas the fastest randomized algorithm achieves  $\tilde{O}(m)$ expected time (where the tilde hides polylogarithmic factors)~\cite{ChaHen2019}.

On the other hand recall that classical graph algorithmics shows that the DFS can be suitably adapted to decompose a finite directed graph into SCCs keeping the time complexity linear,
\eg the celebrated algorithm of Tarjan~\cite{Tar72} finds SCCs in linear-time, 
also see~\cite{Dijkstra1976,Sharir1981, Gabow2001, Mehlhorn1996, TARJAN2024}.

In this work we introduce a novel notion of \emph{alternating strongly-connectedness} on alternating graphs as a natural generalization of strongly-connectedness in directed graphs,
ultimately aiming at providing a \emph{linear-time} decomposition and a sound structural graph characterization.
For this a novel notion of alternating reachability is also introduced, where Player~$\square$ attempts
to reach vertices without leaving a prescribed subset of the vertices, while Player~$\ocircle$ works against.
This is named \emph{safe-alternating reachability}. It is shown that every alternating graph uniquely decomposes into safe-alternating strongly-connected components
  where Player~$\square$ can visit each vertex within a given component infinitely often, without having to ever leave out the component itself.
Our main result is a linear-time algorithm for computing the corresponding alternating graph decomposition.
A key technical ingredient will be to rely on the disjoint-sets union-find data structure~\cite{Tar75},
	and the linearity of the so-called \emph{incremental-tree set-union problem}~\cite{GT85} on RAM machines,
for the fast computation of \emph{lowest common ancestors} in certain search trees underlying our algorithm.
This possibly departs away from the MEC decomposition, for which the existence of a sub-quadratic procedure remains an open question.
Still our proposed algorithms and underlying graph structures do generalize the classical decomposition of a directed graph into SCCs.

Our theory has direct applications, \eg for faster solving some infinite two-player games that are well-known in the field of formal verification and automated synthesis.
Infinite duration graph games can be applied in the construction of finite state reactive systems, like communication protocols or discrete dynamic control systems,
where a central aim is to put the development of hardware and software on a mathematical basis which is both firm and practical.
Moreoover these games can provide game-theoretic foundations for studying infinite duration processes such as operating systems, networks, communication systems and concurrent computations, where a characteristic feature of such systems is their perpetual interaction with the environment as well as their non-terminating behaviour.
The theory of infinite duration games offers many appealing results under this prospect, see \eg~\cite{Gradel02}.

For a concrete instance consider the following communication network problem. Often one requirement is to share key information between all nodes of a network, \eg for monitoring purposes, suppose we have data stored on each node of a computer network and we want to continuously update all nodes with some consistent information.
Imagine a data packet of current information continuously going through all nodes.
Unfortunately not all routing choices are always under our control, some of them may be controlled by the network environment which could play against our benefits.
Essentially this  describes an infinite duration two-player game played on an alternating graph where Player~$\square$ wants to visit all vertices infinitely often, by keep moving the pebble everywhere around and forever,
while Player~$\ocircle$ works against by trying to rule out at least one vertex from a certain time moment onwards.
This model is named \emph{Update Game (UG)}~in~\cite{DK99,Din00,Bod01}.
Dinneen and Khoussainov~\cite{DK99} showed that deciding who's the winner in a given UG costs $O(mn)$ time, where $n$ is the number of vertices and $m$ is that of the arcs.
Solving UGs turns out to be a foundamental subproblem when solving Explicit McNaughton-M\"uller Games in polynomial time as in the algorithm of Horn~\cite{Horn08}.
Now we can solve UGs in linear~time. In turn the complexity of Explicit McNaughton-M\"uller Games improves from cubic to quadratic.

\subsection{Results and Organization}\label{sect:results}
To begin, Section~\ref{subsect:notation}~and~\ref{section:ascc} provide some background notions and formal notation.

In order to accomplish our tasks, a novel notion of alternating reachability is introduced in Section~\ref{subsect:safe-reachability},
  namely, \emph{safe-alternating reachability}, where Player~$\square$ attempts to reach vertices without leaving a prescribed subset of the vertices while Player~$\ocircle$~works~against.

In Section~\ref{sect:stcc}, it is shown that every alternating graph uniquely decomposes into safe-alternating strongly-connected components
  where Player~$\square$ can visit infinitely often each vertex within a given component without having to ever leave out the component itself.
Our main result goes as follows.

\begin{MainRes}
The DFS and SCCs algorithms of~\cite{Tar72, TARJAN2024} admit a non-trivial generalization to alternating graphs, this allows us to improve by a linear factor the time complexity of Update Games~\cite{DK99,Din00,Bod01} and Explicit McNaughton-M\"uller Games~\cite{Horn08}.
The resulting DFS/SCCs algorithms lean on safe-alternating reachability and run in linear time on RAMs~\cite{GT85} and at least Ackermann-linear\footnote{\ie $O(n+m\alpha(m,n))$, where $n$ is the number of vertices, $m$ that of the arcs, and $\alpha()$ is the inverse Ackermann's~\cite{Tar75}. 
We leave open the existence of a neat linear-time solution on pointer machines (this is discussed in Section~\ref{ref:relatedfutureworks}).} on pointer machines~\cite{Tar75}.
\end{MainRes}

Both the underlying graph structures (which are analyzed in Section~\ref{subsect:graphstructures}) and the algorithm (outlined and analyzed in Section~\ref{sect:algo-slip})
generalize the classical decomposition of a directed graph into SCCs.
The proposed linear-time decomposition is given in Section~\ref{sect:aDFS}.

As a first direct application of the proposed algorithmics we obtain the following neat result on Update Games~\cite{DK99,Din00,Bod01}.
The best previously known upper bound was $O(|V||A|)$, as shown by Dinneen and Khoussainov~in~\cite{DK99}.
Section~\ref{sect:UG} offers formal definitions and more details on UGs.
\begin{Cor}\label{thm:UG_linear}
  Deciding who wins a given Update Game $\A=(V, A, \langle V_\square, V_\ocircle \rangle)$ takes time $\Theta(|V|+|A|)$.
\end{Cor}
\begin{proof}[Proof Sketch]
  Player~$\square$ wins if and only if $\A$ has \emph{only one} safe-alternating strongly-connected component; when there are at least two components the winner is Player~$\ocircle$.
To decide this on input $\A$, it is sufficient one run of our proposed decomposition algorithm, $\textit{safe-$\alpha$SCC}(\A)$ (Algorithm~\ref{algo:STCC}) given in Section~\ref{sect:algo-slip}.

    Correctness and complexity will follow from that of Algorithm~\ref{algo:STCC}, see Section~\ref{sect:algo-slip} and Appendix~A.
\end{proof}

\emph{McNaughton-M\"uller Games (MMGs)}~\cite{Horn08} also provide a useful model for the synthesis of controllers in reactive systems, but their complexity depends on the representation of the winning conditions.
The most straightforward way to represent a M\"uller winning condition $\F\subseteq 2^V$ is to provide an explicit list $\F=\{F_i\}_{i=1}^\ell$ of subsets of vertices as in~\cite{Horn08}.
So-called \emph{Explicit MMGs} can be solved in polynomial time, \eg with Horn's algorithm~\cite{Horn08}, by repeatedly deciding UGs as a basic subproblem.
As aftermath of Corollary~\ref{thm:UG_linear}, the complexity of Explicit McNaughton-M\"uller Games also improves from cubic to quadratic.

This is established again in Section~\ref{sect:UG}, where the formal definition of Explicit MMGs is~recalled.
\begin{Cor}\label{thm:explicit-mg}
    Deciding who wins a given Explicit MMG $(\A,\F)$ takes time $O\big(|\F|\cdot (|\A|+|\F|)\big)$.
\end{Cor}

\subsection{Notation and Preliminaries}\label{subsect:notation}
An \emph{alternating graph ($\alpha$graph)} $\A = (V, A, \langle V_{\square}, V_{\ocircle}\rangle)$ is a
  finite directed simple graph  $G_\A\doteq (V, A)$ (\ie there are no loops nor parallel arcs)
 whose vertex set is split into the set of existential vertices $V_{\square}$ owned by Player~$\square$,
    and the set of universal vertices $V_{\ocircle}$ owned by Player~$\ocircle$.
Notice $G_\A$ is not required to be a bipartite graph on colour classes $V_{\square}$ and $V_{\ocircle}$.
Also let $[k]\doteq \{1, \ldots, k\}$ for any $k\in\N$.

The ingoing and outgoing neighbourhoods of any $u\in V$ are denoted by
  $N^{\text{in}}_\A(u)\doteq \{v\in V\mid (v,u)\in A\}$ and $N^{\text{out}}_\A(u)\doteq \{v\in V\mid (u,v)\in A\}$, respectively.

An $\alpha$graph can serve as an arena on which games can be played for infinitely many rounds by moving a pebble along the arcs from one vertex to an adjacent one. Initially the pebble is put on a starting position $s\in V$.
At each round, if the pebble is over $v\in V_i$, for some $i\in\{\square, \ocircle\}$,
Player~$i$ chooses an arc $(v,v')\in A$ and then the next round starts with the pebble on $v'$.

A finite (or infinite) \emph{path} in $G_\A$ is a sequence $v_0v_1\ldots v_n\ldots\in V^*$ (or $V^{\omega}$)
  such that $\forall{j\geq 0}\; (v_j, v_{j+1})\in A$, where the \emph{length} of $v_0v_1\ldots v_n$ is $n$.
A \emph{play path} (or simply, a \emph{play}) on $\mathcal{A}$ is any finite or infinite path in $G_\A$.
A \emph{strategy} for Player~$i$, where $i\in \{\square,\ocircle\}$,
is a map $\sigma_i:V^*\times V_i\rightarrow V$ such that for every finite path $p'v$ in $G_\A$,
where $p'\in V^*$ and $v\in V_i$, it holds that $(v, \sigma_i(p', v))\in A$.
The set of all strategies of Player~$i$ in $\mathcal{A}$ is denoted by $\Sigma^{\mathcal{A}}_i$.
A play $v_0v_1\ldots v_n\ldots $ is \emph{consistent} with some	$\sigma\in\Sigma^\A_i$ if $v_{j+1} = \sigma(v_0v_1\ldots v_{j-1}, v_j)$ whenever $v_j\in V_i$.
Given two strategies $\sigma_{\square} \in\Sigma^\A_\square$ and $\sigma_{\ocircle} \in\Sigma^{\mathcal{A}}_{\ocircle}$, and some $s\in V$,
the \emph{outcome} play $\rho_{\mathcal{A}}(s, \sigma_{\square}, \sigma_{\ocircle})$ is the (unique) play
that starts at $s$ and is consistent with both $\sigma_{\square}$ and $\sigma_{\ocircle}$.
For any $v\in V$, we denote by $\rho_{\mathcal{A}}(s, \sigma_{\square}, \sigma_{\ocircle})_{\leq v}$
the (unique) prefix of $\rho_{\mathcal{A}}(s, \sigma_{\square}, \sigma_{\ocircle})$ 
which ends at the first occurence of $v$ (\ie $v$ is a part of the prefix).
For any finite (or infinite) path $p\in V^*$ (or $p\in V^{\omega}$), 
the \emph{alphabet} $\Gamma(p)$ is the set of vertices appearing~in~$p$.
The $\square$-\emph{attractor} set $T_\square$ of any vertex $t\in V$
  is defined as follows: \textit{i)} $t \in T_\square$; \textit{ii)} if $x\in V_\square$ has an outgoing arc $(x,y)$ such that $y\in T_\square$, then $x\in T_\square$;
\textit{iii)} if $x\in V_\ocircle$ has all outgoing arcs $(x,y)$ such that $y\in T_\square$, 
then $x\in T_\square$; \textit{iv)} nothing else is in $T_\square$.

Let $T = (V_T, A_T)$ be an inward directed tree rooted at $r_T \in V_T$, 
\ie a tree in which each arc is oriented toward the root.
We simply write $u\in T$ for $u\in V_T$.
For each $u\in T$, there is only one \emph{path $p_u$} going from $u$ to $r_T$; the \emph{depth} $d(u)$ of $u$ is the length of $p_u$.
An \emph{ancestor} of $u\in T$ is any $v\in \Gamma(p_u)$; it is a \emph{proper ancestor} if $v\neq u$,
it is the \emph{parent} $\pi_T(u)$ of $u$ if $(u,v)\in A_T$.
The \emph{children} of $u\in T$ are all the $v\in T$ such that $\pi_T(v)=u$.
A \emph{descendant} of $u\in T$ is any $v\in T$ such that $u\in \Gamma(p_v)$; it is a \emph{proper descendant} if $v\neq u$.
A \emph{leaf} of $T$ is any $u\in T$ having no children.
The \emph{lowest common ancestor (LCA)} $\gamma_S$ of a subset of vertices $S\subseteq T$ is:
  \[ \gamma_S\doteq\arg\max \big\{ d(\gamma)\in \N \mid \gamma\in T
       \text{ is an ancestor of each vertex in } S\big\}.\]
The maximal subtree of $T$ that is rooted at any $u\in T$ is denoted by $T_u$.
  Given a LIFO stack $\textit{St}$ containing some element $v\in \textit{St}$,
    then $\textit{St}(v)$ denotes the set of all elements $u\in \textit{St}$ going from the top of $\textit{St}$ down until the first occurence of $v$ (extremes included).

    \section{Alternating Strongly-Connected Components}\label{section:ascc}
    This section deepens \emph{alternating strongly-connectedness} and its \emph{safe} form.
    We shall see that both concepts can be built bottom-up
     (\ie as a natural generalization of strongly-connectedness in directed graphs)
    and that they are sound and applicable (\ie they enjoy a clear characterization in terms of quotient sets
    	of reachability equivalence relations, and they can be directly applied for faster solving tasks concerning infinite games on graphs).

    Firstly, we consider alternating reachability and alternating strongly-connected components as
    the most natural notion in the neighborhood of possible definitions, already presenting some technical pitfalls compared to the traditional setting.
    Secondly, aiming at providing a linear-time decomposition algorithm and a sound structural graph characterization,
    we introduce \emph{safe-alternating reachability}, a novel notion of alternating reachability on $\alpha$graphs that
    	will form the backbone on which the forthcoming theory will sustain.
    Upon this, \emph{safe-alternating strongly-connectedness} is introduced (in turn, a novel notion of alternating strongly-connectedness).
    It is offered a sound definition of \emph{safe-alternating strongly-connected components} in
    	terms of safe-alternating reachability quotient sets (\ie equivalence classes).

    Safe-alternating reachability captures in a natural way the fundamental invariant property lying
    at the ground of both the forthcoming graph structures and linear-time decomposition algorithms --
    this is actually the reason why it seems necessary and not just interesting to introduce the safe form.

    To conclude the section, we shall observe that both alternating strongly-connectedness and its safe form
    	can be employed to solve Update Games, and thus Explicit McNaughton-M\"uller Games as shown in~\cite{Horn08}.
    The algorithm of Section~\ref{sect:algo-slip} will ultimately provide a faster solution to those two games.
    Let us start recalling alternating reachability.
    \begin{Def}[\cite{Chandra1981,Greenlaw1995}]
    Let $\A$ be an $\alpha$graph on vertex set $V$, and let $u,v\in V$ be any two vertices.

    We say that $v$ is \emph{alternating reachable ($\alpha$reachable)}	from $u$ in $\A$ if and only if
    there exists a strategy $\sigma_{\square}\in \Sigma^{\A}_{\square}$ such that for every $\sigma_{\ocircle}\in \Sigma^{\A}_{\ocircle}$ it holds that the target
    $v$ lies in the outcome play which starts at $u$ and proceeds consistently with the given strategies, \ie if and only if
    \[\exists{\sigma_{\square}\in \Sigma^{\A}_{\square}}\,
    \forall{\sigma_{\ocircle}\in \Sigma^{\A}_{\ocircle}}\;\;	v\in \Gamma\big(\rho_{\mathcal{A}}(u, \sigma_{\square}, \sigma_{\ocircle})\big).\]
    \end{Def}
    The $\alpha$reachability relation between $u$ and $v$ by strategy $\sigma_{\square}$ will be compactly denoted by $\sigma_{\square}:u \leadsto v$.

    Then let's consider a natural notion of alternating strongly-connectedness, also clarified in Example~\ref{ex:arena1}.
    \begin{Def}\label{def:strongly-trap-connected} Let $\A$ be an $\alpha$graph on vertex set $V$.
    We say $U\subseteq V$ is an \emph{alternating strongly-connected set ($\alpha$sc set)} if and only if
    $\forall{(u,v)\in U\times U}$ $\exists{\sigma_{\square}\in \Sigma^{\A}_{\square}}$ such that $\sigma_{\square}:u \leadsto v$.
    \end{Def}

    \begin{figure}[h]
    \begin{center}
    \begin{tikzpicture}[arrows=->,scale=.8,node distance=1.5 and 1.5]
      \node[node] (A) {$a$};
      \node[squa, xshift=6ex, above=of A, label={above right, yshift=-1ex: $\beta$}] (B) {$b$};
      \node[squa, right=of A, label={above right,  xshift=2ex, yshift=-1ex: $\gamma$}] (C) {$c$};
      \draw[] (A) to [bend left=30] node[below] {} (B);
      \draw[] (B) to [bend left=30] node[above] {} (C);
      \draw[] (A) to [bend left=30] node[above] {} (C);
      \draw[] (C) to [bend left=30] node[above] {} (A);

    	\draw[dashed, ultra thin, rounded corners=15pt] (0,2.75) rectangle (2,1.5);
    	\draw[dashed, ultra thin, rounded corners=15pt] (-1.5,.65) rectangle (3.5,-.65);
    \end{tikzpicture}
    \caption{An $\alpha$graph on vertex set $\{a,b,c\}$
      and its $\alpha$SCCs $\gamma=\{a,c\}$ and $\beta=\{b\}$, as in Example~{\ref{ex:arena1}}.}\label{fig:arena1}
    \end{center}
    \end{figure}

    \begin{Def}
    $\sim_{\alpha\text{sc}}$ is the binary relation on $V$ defined as follows:
    		$\sim_{\alpha\text{sc}}\doteq \big\{ (u,v)\in V\times V \mid	\{u,v\} \text{ is $\alpha$sc} \big\}$.
    \end{Def}

    It is easy to check that $\sim_{\alpha\text{sc}}$ is an equivalence relation on $V$, where every equivalence class is $\alpha$sc.
    \begin{Def}\label{def:aSCCs}
    Let $\A$ be an $\alpha$graph on vertex set $V$. Let $\C\subseteq V$ be a subset of the vertices
    and consider the relation $\sim_{\alpha\text{sc}}$.
    We say that $\C$ is an \emph{alternating strongly-connected component ($\alpha$SCC)}
    of $\A$ precisely when it is an equivalence class of $\sim_{\alpha\text{sc}}$.
    \end{Def}

    Otherwise stated, an $\alpha$SCC is any maximal (under set inclusion) $\alpha$sc set of vertices in the $\alpha$graph $\A$.

    \begin{Exa}\label{ex:arena1} Consider the $\alpha$graph $\mathcal{A} = (V, A, \langle V_\square, V_\ocircle \rangle )$, where 
	$V = V_\square \cup V_\ocircle$ with $V_\square = \{b, c\}$, $V_\ocircle = \{a\}$, and 
	$A = \{(a,c), (a,b), (b,c), (c,a)\}$. Figure~\ref{fig:arena1} illustrates the corresponding decomposition 
	into $\alpha$SCCs, $\{\beta, \gamma\}$. 
	Now, consider the component $\gamma = \{a, c\}$. If the play starts from $a$, the pebble will eventually reach $c$; 
	however, to do so, it might first escape from $\gamma$ by reaching $b$. That is, although $c$ is $\alpha$-reachable 
	from $a$, Player~$\square$ has no strategy to ensure that the play remains within $\gamma$. 
	Player~$\ocircle$ can always force the pebble out to $\beta$, preventing Player~$\square$ from maintaining safety inside $\gamma$.

    \end{Exa}

    Decomposing an $\alpha$graph into its $\alpha$SCCs can be achieved in $O(|V||A|)$ time. 
For each vertex $t \in V$, one can compute its $\square$-attractor set $T_\square$ in $O(|A|)$ time. 
Two vertices $u, v \in V$ belong to the same $\alpha$SCC if and only if each lies in the $\square$-attractor of the other. 
The existence of sub-quadratic time algorithms for computing $\alpha$SCCs is not addressed in this paper and remains an open problem.

These properties highlight key differences between $\alpha$SCCs and their classical counterparts in directed graphs. 
Nevertheless, $\alpha$SCCs are well-defined, as they correspond to the equivalence classes of a natural equivalence relation on the vertex set.

We explicitly analyze these fundamental properties because they will be extended to the \emph{safe} notion of alternating strong connectivity.
To formalize this, we now introduce safe-$\alpha$reachability.

\subsection{Safe-Alternating Reachability}\label{subsect:safe-reachability}
		
In the context of strategic interactions modeled by $\alpha$graphs, 
reachability is a fundamental property used to determine whether a target vertex can be reached under adversarial conditions.
		
\textit{Alternating reachability} is a well-established notion in this setting,  
where two players---one controlling existential choices ($\square$) and the other adversarial choices ($\ocircle$)---compete 
to either enforce or obstruct reaching a designated node.

However, in many applications, simply reaching the target is not enough; 
it is equally crucial that the path taken satisfies specific constraints. 
This leads to the refinement we introduce: \textit{safe-alternating reachability}. 

The main difference between these two notions is the inclusion of an additional safety condition:
\begin{itemize}
    \item \textbf{Alternating Reachability}: A vertex $v$ is \textit{alternatingly reachable} from $u$ if there exists a strategy for $\square$ that ensures $v$ is eventually reached, regardless of the adversarial strategy chosen by $\ocircle$.
    \item \textbf{Safe-Alternating Reachability}: In addition to ensuring reachability, 
	the traversal must remain entirely within a predefined subset of ``safe" vertices $U$. 
	That is, the path to $v$ must never leave $U$, irrespective of the adversary’s moves.
\end{itemize}

	The upcoming formal definition  of \emph{safe-$\alpha$reachability} rigorously establishes this concept, 
	ensuring a precise mathematical foundation for its applications.
    
    \begin{Def}\label{def:Treachability}
    Given an $\alpha$graph $\A$ on vertex set $V$, let $U\subseteq V$ and $u,v\in U$.
    A vertex $v$ is \emph{$U$-safe-$\alpha$reachable} from $u$ when
    \emph{there exists} a strategy $\sigma_{\square}\in\Sigma^{\mathcal{A}}_{\square}$ such that
    \emph{for every} adversarial strategy $\sigma_{\ocircle}\in\Sigma^{\mathcal{A}}_{\ocircle}$:
    \begin{itemize}
    	\item[] \emph{[$\alpha$reachability]} $v$ is eventually reached by playing $\sigma_\square$ starting from $u$,
    		\ie $v\in \Gamma\big[\rho_{\mathcal{A}}(u, \sigma_{\square}, \sigma_{\ocircle})\big]$;
    	\item[] \emph{[safety]} the pebble never leaves $U$ until it reaches $v$, \ie  
		$\Gamma\big[\rho_{\mathcal{A}}(u, \sigma_{\square}, \sigma_{\ocircle})_{\leq v}\big] \subseteq U$.
		(Notice that this implies that $v$ must be in $U$.)
    \end{itemize}
    \end{Def}
    In that case denote $\sigma_{\square}:u \overset{U}{\leadsto} v$, or $u\overset{U}{\leadsto} v$ when $\sigma_{\square}$ is implicit;
    if $U=V$, denote $\sigma_{\square}:u \leadsto v$ or $u\leadsto v$.

    \emph{Remark:} By convention, any $u\in U$ is $U$-safe-$\alpha$reachable from itself for every non-empty $U\subseteq V$. \qed

\subsubsection{Why Safe-Alternating Reachability?}

While alternating reachability is a fundamental concept, it does not ensure that the path to the target satisfies safety constraints. 
In many practical settings, ensuring reachability without violating safety constraints is critical. Consider the following motivating cases:

\begin{itemize}
    \item \textbf{Verification and Synthesis}: In verification and synthesis of reactive systems, a controller may need to ensure that a system reaches a safe operating state while avoiding hazardous conditions. Standard alternating reachability cannot guarantee that the system does not enter an unsafe region.
    \item \textbf{Games on Graphs}: In two-player games on graphs, a player may want to ensure that a winning condition is met without ever stepping into an opponent’s stronghold. Alternating reachability alone does not suffice, as it does not constrain the path taken.
    \item \textbf{AI and Robotics}: In AI planning and motion strategy problems, an agent may need to navigate towards a goal while ensuring that it never leaves a valid operational space. Traditional alternating reachability only guarantees eventual arrival but does not prevent excursions into restricted zones.
\end{itemize}

By explicitly incorporating a \textit{safety condition}, safe-alternating reachability refines the classical model to better capture scenarios where both goal achievement and constraint satisfaction are crucial. 

\paragraph{Strict Safety Constraints.} 
A natural alternative to the definition we adopt in this work is a strict notion of safe-alternating reachability, 
where the path is required to remain within $U$ for the entire play even beyond reaching $v$.

The strict notion can be formalized as follows:
\[
\text{\emph{[strict safety]} The pebble never leaves } U, \text{ \ie } \Gamma\big[\rho_{\mathcal{A}}(u, \sigma_{\square}, \sigma_{\ocircle})\big] \subseteq U.
\]
This strict condition can be relevant in applications where long-term safety is required, 
and once a system enters a safe region, it must never leave. Some examples include:
\begin{itemize}
    \item \textbf{Safety-Critical Systems}: In control theory and automated verification, where systems must always operate within predefined safety constraints, such as air traffic control or autonomous vehicle operation.
    \item \textbf{Security and Isolation}: In cybersecurity, where execution must remain within a sandboxed environment to prevent unauthorized access or data leakage.
    \item \textbf{Regulated Multi-Agent Systems}: In distributed AI and robotic systems, where agents are confined to bounded regions and must never leave their designated operational space.
\end{itemize}
In this work, we adopt the weaker Definition~\ref{def:Treachability}, 
which requires that the play stays within $U$ only until reaching $v$, but not necessarily afterward.
Remarkably, for the objective of studying safe-alternating strongly-connected components, 
both the weak and strict definitions lead to the same results. 
This is because if a component is weakly safe-alternating strongly-connected, 
then once a vertex within the component is reached, every other vertex in the component must still be reachable 
from it while staying entirely within $U$. 
Consequently, whether we impose the safety condition only up to $v$ (weak version) or throughout 
the entire play (strict version), the set of reachable nodes within the strongly-connected component remains unchanged. 

We are now ready to revisit alternating strong connectivity under safe-alternating constraints.

    \subsection{Safe-Alternating Strongly-Connected Components}\label{sect:stcc}

	\begin{Def}\label{def:safe-asc}
		Let $\mathcal{A}$ be an $\alpha$graph on vertex set $V$. 
		A subset $U \subseteq V$ is \emph{safe-alternating strongly connected} (safe-$\alpha$SC) if and only if, 
		for every pair of vertices $u, v \in U$, $v$ is $U$-safe-$\alpha$reachable from $u$.
		
		That is, there exists a strategy $\sigma_{\square} \in \Sigma^{\mathcal{A}}_{\square}$ such that:
		\[
			\sigma_{\square}: u \overset{U}{\leadsto} v.
		\]
	\end{Def}

    Notice that $\emptyset$ and $\{v\}$ are $\{v\}$-safe-$\alpha$sc for every $v\in V$. \qed

    \begin{figure}[h]
    \begin{center}
    \begin{tikzpicture}[arrows=->,scale=.8,node distance=1.5 and 1.5]
      \node[node, label={above left, xshift=.25ex, yshift=-.25ex: $\alpha$}] (A) {$a$};
      \node[squa, xshift=6ex, above=of A, label={above right, yshift=-1ex: $\beta$}] (B) {$b$};
      \node[squa, right=of A, label={above right,  xshift=.75ex, yshift=-1ex: $\gamma$}] (C) {$c$};
      \draw[] (A) to [bend left=30] node[below] {} (B);
      \draw[] (B) to [bend left=30] node[above] {} (C);
      \draw[] (A) to [bend left=30] node[above] {} (C);
      \draw[] (C) to [bend left=30] node[above] {} (A);
    	\draw[dashed, ultra thin, rounded corners=15pt] (0,2.75) rectangle (2,1.5);
    	\draw[dashed, ultra thin, rounded corners=15pt] (-1,.65) rectangle (1,-.65);
    	\draw[dashed, ultra thin, rounded corners=15pt] (1.35,.65) rectangle (3.25,-.65);
    \end{tikzpicture}
    \caption{An $\alpha$graph on vertex set $\{a,b,c\}$, and its safe-$\alpha$SCCs $\alpha=\{a\}, \beta=\{b\}, \gamma=\{c\}$.}\label{fig:safe-asc}
    \end{center}
    \end{figure}

\begin{Exa}\label{ex:safe-asc}
Figure~\ref{fig:safe-asc} depicts an $\alpha$graph with vertex set $\{a, b, c\}$, 
where directed arcs represent possible transitions between vertices. 
The $\alpha$graph is decomposed into its safe-alternating strongly connected components (safe-$\alpha$SCCs), 
which are enclosed by dashed boundaries. In this case, every vertex forms a singleton component: 
$\alpha = \{a\}$, $\beta = \{b\}$, and $\gamma = \{c\}$.

According to Definition~\ref{def:safe-asc}, a set $U \subseteq V$ is a safe-$\alpha$SCC if every vertex in $U$ is 
$U$-safe-$\alpha$reachable from every other vertex in $U$. In this example, no two distinct vertices are 
mutually safe-alternating reachable within a common set while satisfying the safety constraints. 
As a result, the decomposition consists solely of singleton components.
\end{Exa}

Example~\ref{ex:safe-asc} highlights how safe-alternating reachability can restrict connectivity,  
resulting in a finer decomposition than classical SCCs. 
It also illustrates the impact of safety constraints on the structure of the graph, 
where even direct connections between nodes do not necessarily imply safe-alternating 
reachability within a larger set.

    Next, let us observe the following \emph{composition} property concerning safe-$\alpha$sc sets.
    \begin{Lem}\label{lemma:safe-asc-union}
    Let $V_1,V_2\subseteq V$ be two safe-$\alpha$sc sets.
    If $V_1\cap V_2\neq \emptyset$, then $V_1\cup V_2$ is safe-$\alpha$sc.
    \end{Lem}
    \begin{proof} 
	Let \(u,v\in V_1\cup V_2\).
	If \(\{u,v\}\subseteq V_i\) for some \(i\in\{1,2\}\), we are done since \(V_i\) is safe-\(\alpha\)sc.
	Otherwise, \textit{w.l.o.g.} assume \(u\in V_1\setminus V_2\) and \(v\in V_2\setminus V_1\) (the other cross case is symmetric).
	Pick $z\in V_1\cap V_2$.
	Since $\{u,z\}\subseteq V_1$, and since $V_1$ is safe-$\alpha$sc,
    there exists some strategy $\sigma_{\square}(u,z)\in\Sigma^{\mathcal{A}}_{\square}$ such that:
    \[\sigma_{\square}(u,z):u \overset{V_1}{\leadsto} z,\] similarly,
    there is some other strategy $\sigma_{\square}(z,v)\in\Sigma^{\mathcal{A}}_{\square}$
    such that: \[\sigma_{\square}(z,v):z \overset{V_2}{\leadsto} v.\]
    Then, consider the strategy $\sigma_{\square}(u,v)\in \Sigma^{\mathcal{A}}_\square$ constructed by gluing $\sigma_{\square}(u,z), \sigma_{\square}(u,v)$ in sequence:
    \[ \sigma_{\square}(u,v)\doteq
    	\left\{
    	\begin{array}{l}
    	\text{\emph{(1)} Starting from $u$, play $\sigma_{\square}(u,z)$ until $z$ is firstly reached; then,}	\\
    	\text{\emph{(2)} once on $z$, play $\sigma_{\square}(z,v)$ until $v$ is finally reached.}
    	\end{array} \right.
    \]
    Clearly, $\sigma_{\square}(u,v):u\overset{V_1\cup V_2}{\leadsto} v$.
    Since $u$ and $v$ were chosen arbitrarily, then $V_1\cup V_2$ is safe-$\alpha$sc.
    \end{proof}

    Lemma~\ref{lemma:safe-asc-union} allows us to define and study the following binary relation on $V$.
    \begin{Def}
    	The binary relation $\sim_{\text{safe}}\subseteq V\times V$ is defined as:
    \[	\sim_{\text{safe}}\doteq \big\{ (u,v)\in V\times V \mid \exists{U\subseteq V} \text{ such that } U \text{ is safe-$\alpha$sc and } \{u,v\}\subseteq U \big\}. \]
    \end{Def}

    \begin{Lem}\label{lemma:safe-equiv}
    $\sim_{\text{safe}}$ is an equivalence relation on $V$.
    \end{Lem}
    \begin{proof}
    To begin, (i) $\sim_{\text{safe}}$ is clearly \emph{reflexive}: for any $u\in V$, let $U\doteq \{u\}$;
    then, $u\overset{U}{\leadsto} u$, so $U$ is safe-$\alpha$sc; this shows $u\sim_{\text{safe}} u$.
    (ii) $\sim_{\text{safe}}$ is \emph{symmetric}, (actually, by definition): for any $u,v\in V$, assume $u\sim_{\text{safe}} v$;
    then, there exists some $U\subseteq V$ which is safe-$\alpha$sc and $u,v\in U$; so, the same set $U$ certifies $v\sim_{\text{safe}} u$.
    Finally, (iii) $\sim_{\text{safe}}$ is \emph{transitive}: indeed, for any $a,b,c\in V$, assume  $a\sim_{\text{safe}} b$ and $b \sim_{\text{safe}} c$.
    Since $a\sim_{\text{safe}} b$, there exists $V_1$ which is safe-$\alpha$sc and such that $a,b\in V_1$;
    similarly, there exists $V_2$ which is safe-$\alpha$sc and such that $b,c\in V_2$.
    Consider $U\doteq V_1\cup V_2$. Since $b\in V_1\cap V_2$, and $V_1,V_2$ are both safe-$\alpha$sc,
    then $U$ is safe-$\alpha$sc by Lemma~\ref{lemma:safe-asc-union}. Moreover, $a,c\in U$. So, $a\sim_{\text{safe}} c$.

    Thus $\sim_{\text{safe}}$ is an equivalence relation.
    \end{proof}

    Let us point out some interesting properties of $\sim_{\text{safe}}$ equivalence classes.
    \begin{Lem}\label{lemma:STCC}
    Let $\{\mathcal{C}_i\}_{i=1}^k$ be all the distinct equivalence classes of $\sim_{\text{safe}}$ on $V$.
    Then, the following holds.
    \begin{enumerate}
    	\item If $U\subseteq V$ is safe-$\alpha$sc and $U\cap \mathcal{C}_i\neq \emptyset$ for some $i\in [k]$,
    					then $U\subseteq \mathcal{C}_i$;
    	\item $\mathcal{C}_i$ is safe-$\alpha$sc for each $i\in [k]$;
    	\item Let $U\subseteq V$ be safe-$\alpha$sc. Then, $\mathcal{C}_i\subsetneq U$ for \emph{no} $i\in [k]$.
    	\end{enumerate}
    \end{Lem}
	\begin{proof}
		\textit{Proof of (1).} Since $U\cap \mathcal{C}_i\neq \emptyset$, it's possible to pick $z\in U\cap \mathcal{C}_i$. 
		Pick $v\in U$, arbitrarily. Since $U$ is safe-$\alpha$sc and $z,v\in U$, then $v\sim_{\text{safe}} z$. 
		So, $v\in \mathcal{C}_i$ (because $z\in \mathcal{C}_i$, which is an equivalence class of $\sim_{\text{safe}}$).
		
		\textit{Proof of (2).} Let $u,v\in \mathcal{C}_i$, arbitrarily. Then, $u\sim_{\text{safe}} v$. 
		So, there exists some $U\subseteq V$ which is safe-$\alpha$sc and such that $u,v\in U$. 
		Thus, $u\overset{U}{\leadsto} v$. Notice, $u,v\in U\cap\mathcal{C}_i\neq \emptyset$. 
		Then, by item~1 of Lemma~\ref{lemma:STCC}, $U\subseteq \mathcal{C}_i$. 
		Since $u\overset{U}{\leadsto} v$ and $U\subseteq \mathcal{C}_i$, then $u\overset{\mathcal{C}_i}{\leadsto} v$. 
		So, $\mathcal{C}_i$ is safe-$\alpha$sc.
		
		\textit{Proof of (3).} Assume that $\mathcal{C}_i\subseteq U$, for some $i\in [k]$, 
		and some $U\subseteq V$ which is safe-$\alpha$sc. Then, since $U\cap \mathcal{C}_i = \mathcal{C}_i\neq\emptyset$, 
		by item~1 of Lemma~\ref{lemma:STCC} we have $U\subseteq \mathcal{C}_i$. So, $\mathcal{C}_i = U$.
	\end{proof}

    \begin{Def}\label{def:STCC}
    	Let $\A$ be an $\alpha$graph on vertex set $V$. Let $\C\subseteq V$ be a
    	subset of the vertices and consider the binary equivalence relation $\sim_{\text{safe}}$ on $V$.
    	We say that $\C$ is an \emph{alternating strongly-connected component ($\alpha$SCC)}
    	of $\A$ precisely when it is an equivalence class of $\sim_{\text{safe}}$.
    	\end{Def}

    	Otherwise stated, by Lemma~\ref{lemma:STCC}, an $\alpha$SCC is any maximal (under set inclusion) safe-$\alpha$sc vertex subset of the $\alpha$graph.

    Moreover, since safe-$\alpha$sc is a more constrained form of $\alpha$sc, the former implies the latter (as below).
    \begin{Prop}\label{prop:refinement}
    The $\sim_{\text{safe}}$ equivalence relation is \emph{finer} than $\sim_{\alpha\text{sc}}$.
    \end{Prop}
    \begin{proof}
    	It is enough to point out that every equivalence class of $\sim_{\text{safe}}$ is a subset of an equivalence class of $\sim_{\alpha\text{sc}}$
    		(and thus every equivalence class of the latter is a union of equivalence classes of the former).

    			This is clear as every safe-$\alpha$sc set is $\alpha$sc too.
    \end{proof}

    \subsection{Applications to Update Games and McNaughton-M\"uller Games}\label{sect:UG}
    An \emph{Update Game (UG)}~\cite{DK99,Din00,Bod01} is played on an $\alpha$graph $\A$ with vertex set $V$ and arc set $A$ for an infinite number of rounds.
    Here a \emph{play} is an infinite path $\rho=v_0 v_1 v_2 \ldots \in V^{\omega}$ such that $(v_i,v_{i+1}) \in A$ $\forall i\in \N$.
    Let $\text{Inf}(\rho)$ be the set of all the vertices $v\in V$ appearing infinitely often in $\rho$;
    namely, \[\text{Inf}[\rho]\doteq \big\{v\in V\mid \forall{j\in\N}\; \exists{k\in\N}, k > j, \text{ such that } v=v_k \big\},
    	\text{ provided } \rho = v_0 v_1 v_2 \dots v_k \ldots\in V^{\omega}.\]
    Player~$\square$ wins the UG played on $\A$ if and only if there exists $\sigma_{\square}\in\Sigma^{\A}_{\square}$ such that,
    	for every $\sigma_{\ocircle}\in\Sigma^{\A}_{\ocircle}$,
    every vertex is visited infinitely often in the unique play that is consistent with $\sigma_{\square}$ and $\sigma_{\ocircle}$,
    independently \wrt the starting position $s\in V$; namely, if and only if the following holds:
    	\[\exists{\sigma_{\square}\in\Sigma^{\A}_{\square}}\forall{\sigma_{\ocircle}\in\Sigma^{\A}_{\ocircle}} \forall{s\in V}\;
    		\text{Inf}\big[\rho_{\A}(s, \sigma_{\square}, \sigma_{\ocircle})\big]=V; \]
    otherwise, Player~$\ocircle$ wins. When Player~$\square$ wins an UG $\A$,	then $\A$ is called \emph{Update Network (UN)}~\cite{DK99,Din00,Bod01}.

    In order to decide who wins an UG, we can check whether the whole vertex set $V$ is either safe-$\alpha$sc or simply $\alpha$sc
    	(indifferently, as clearly the two conditions are equivalent for the whole vertex set,
    		\ie notice that the whole vertex set $V$ is safe-$\alpha$sc if and only if it is $\alpha$sc).
    \begin{Prop}\label{prop:UNequivsafe-asc}
    Let $\A$ be an UG on vertex set $V$.
    	Player~$\square$ wins the UG played on $\A$ if and only if $V$ is safe-$\alpha$sc; or equivalently (since $V$ is the whole vertex set),
    	if and only if $V$ is $\alpha$sc.
    \end{Prop}
    \begin{proof}
    If Player~$\square$ wins the UG played on $\A$, then $V$ is safe-$\alpha$sc (it follows directly from definitions,
    	as every vertex can be visited infinitely often then every vertex is $\alpha$reachable from any other one).
    Conversely, if $V$ is safe-$\alpha$sc, and $v_0, \ldots, v_{|V|-1}$ is a vertex ordering,
      for every $i$ there is $\sigma_\square(i)\in\Sigma^{\A}_\square$ such that $\sigma_\square(i) : v_i \leadsto v_{i'}$,
    			where $i'\doteq (i+1) \mod |V|$ for every $i\in \{0, \ldots, |V|-1\}$.
    Starting from any $v_i$, Player~$\square$ can visit infinitely often all vertices in $V$
      by playing forever $(\sigma_\square(i), \sigma_\square(i'), \sigma_\square(i''), \ldots)$ cascade.
    For the whole vertex set $V$, the same argument works if we consider $\alpha$sc instead of safe-$\alpha$sc.
    \end{proof}

    The fact is that we are not currently aware of any sub-quadratic	time algorithm for checking $\alpha$sc.
    Instead, our propsed solution for checking safe-$\alpha$sc runs in linear-time (as if we were computing all $|V|$ attractors in $O(|A|)$ aggregate time).
    Thus we employ safe-$\alpha$sc for solving UGs.

    Let us consider also
    	\emph{McNaughton-M\"uller Games (MMGs)}~\cite{Horn08}. They provide a useful model for the synthesis of controllers in reactive systems,
    but their complexity depends on the representation of the winning conditions.
    The most straightforward way to represent a M\"uller winning condition $\F\subseteq 2^V$
    is to provide an explicit list of subsets of vertices as in~\cite{Horn08},
    \ie $\F=\{\F_i\subseteq V\mid 1 \leq i \leq \ell \}$ for some $\ell\in\N$.

    A play $\rho\in V^\omega$ is winning for Player~$\square$ if and only if $\text{Inf}[\rho]\in \F$.
    So-called \emph{Explicit MMGs} can be solved in polynomial time, \eg with Horn's algorithm~\cite{Horn08}.
    Concerning time complexity, given an input $\alpha$graph $\A$ and explicit winning condition $\F$,
    there are at most $|\F|$ loops in a run of that algorithm,
    and the most time consuming operation at each iteration is precisely to decide an UG of size at most $|\A|+|\F|$, see~\cite{Horn08}.

    Thus deciding whether the whole vertex set of a game is safe-$\alpha$sc/$\alpha$sc is relevant to EMMGs~too.

    By Corollary~\ref{thm:UG_linear}, we can decide an UG in $\Theta(|\A|+|\F|)$ linear-time.
    As a consequence, the time complexity of Horn's algorithm \cite{Horn08} improves by a factor $|\A|+|\F|$ (\ie from cubic to quadratic).

    In summary, from Corollary~\ref{thm:UG_linear} and Horn's algorithm~\cite{Horn08}, we obtain Corollary~\ref{thm:explicit-mg} (\cfr Section~\ref{sect:results}).

    \section{Safe-Alternating Depth-First Search}\label{sect:aDFS}

This section introduces $\alpha$DFS, a depth-first reverse exploration algorithm for $\alpha$graphs, 
inspired by classical Depth-First Search (DFS). 

\subsection{Classical DFS}

Instead of presenting $\alpha$DFS in an abstract, standalone manner, 
we motivate its development by showing how it naturally arises when applying depth-first exploration principles to 
$\alpha$graphs. In particular, we build upon the classical DFS structures identified 
by~\cite{Tar72}, which we first recall before introducing $\alpha$DFS.

    \textbf{Palm-trees and jungles.}
In~\cite{Tar72}, several fundamental properties and applications of DFS were analyzed. 
In particular, the study introduced two key underlying graph structures, 
called \emph{palm-trees} and \emph{jungles}. 
These structures enabled the development of a now-classic linear-time algorithm for computing 
strongly connected components (SCCs), commonly referred to as \emph{The SCCs algorithm}.

    Following~\cite{Tar72, Cormen2001}, the recursive strategy of the DFS is to search deeper in the graph whenever possible.
    Initially all vertices are unexplored. Start from some vertex $u$ and choose an \emph{outgoing} arc to follow.
    Recursively, the arcs are firstly explored out of the most recently discovered vertex $v$ that still has unexplored arcs leaving it,
     by scanning the adjacency list of the already discovered vertex $v$.
    When all of $v$'s arcs have been explored, the search backtracks one step back to explore the
      remaining arcs leaving that vertex from which $v$ was discovered just before.
    This process continues until we have discovered all the vertices that are reachable from the original source vertex $u$.
    If any undiscovered vertices remain, then one of them is picked as a brand new source and the search is repeated from that.
    The entire process is repeated until all vertices are discovered.
    Besides exploring the graph the DFS also timestamps each vertex twice, where each timestamp is a natural number:
    the first one, named $\text{open}[v]:V\rightarrow \N$, records when $v\in V$ is first discovered; the second timestamp
    $\text{close}[v]$ records when the search finishes examining $v$'s adjacency list.
    These timestamps are used in many algorithms and are generally helpful in reasoning about the behavior of the DFS.
    Let us call it \emph{forward}-DFS, for, at each step the chosen arc is \emph{outgoing}.

    Concerning \emph{palm-trees}~\cite{Tar72}, observe in more detail what happens when DFS runs.
    The set of arcs $A_\pi$ first leading to an unexplored vertex, when traversed during the search, 
	forms a family of outward directed rooted \emph{trees} $T$, \ie where each arc is oriented away from the root.
    The predecessor\footnote{The symbol $\pi$ in $A_{\pi}$ simply stands for “predecessor”.} subgraph $(V,A_{\pi})$ of a DFS is thus a forest defined as:
	\[
	A_{\pi} \doteq \big\{ (\pi_v, v) \mid \pi_v, v \in V \text{ and $v$ is first discovered from $\pi_v$ during the DFS} \big\}.
	\]

    All of the other arcs of the input graph $G$ fall into four categories: (i) some arcs are running from ancestors to descendants in $T$,
    these may well be ignored as (even if we remove them from the graph) they do not affect the strongly-connectedness of $G$;
    still, (ii) some other arcs run from descendants to ancestors in $T$, these are quite relevant to determine strongly-connectedness instead,
    	and they are called \emph{fronds};
    (iii) other arcs run from one subtree to another within the same tree $T$, these are also relevant and named \emph{internal cross-links};
    (iv) suppose to continue the DFS until all arcs are explored,
    the process creates a family of trees which contains all vertices, \ie a \emph{spanning forest} $F=(V,A_\pi)$ of $G$,
     plus sets of (fronds and) cross-links which may also connect two different trees in $F$, and these would be called the \emph{external cross-links}.
    Notice that any (internal or external) cross-link $(u, v)$ always has $\textit{open}[u] > \textit{open}[v]$.

    Any tree $T$ of $F$, comprising fronds and cross-links, is called \emph{palm-tree}.

    A directed graph consisting of a spanning forest, plus fronds and cross-links, is named \emph{jungle},
    \ie a family of palm-trees plus external cross-links; this is a natural representation of the graph reachability relations of the input directed graph $G$.

    \begin{figure}[tb]
\centering
\subfloat[An $\alpha$graph $\A$.\label{fig:example_arena_dfs}]{
\begin{tikzpicture}[arrows=->, scale=.6]
	\node[circ] (A) {$C$};
	\node[squa, right=of A, yshift=0ex, xshift=0ex] (B) {$B$};
	\node[squa, below=of B] (C) {$A$};
	\node[squa, left=of C, xshift=0ex, yshift=0ex] (D) {$H$};
	\node[squa, above left=of A, xshift=0ex, yshift=0ex] (E) {$D$};
	\node[circ, above right=of B, yshift=0ex, xshift=0ex] (F) {$G$};
	\node[squa, below right=of C, xshift=0ex, yshift=0ex] (G) {$F$};
	\node[squa, below left=of D, yshift=0ex, xshift=0ex] (H) {$E$};

	\draw[] (A) to [xshift=0ex, yshift=0ex] node[below, xshift=0ex, yshift=0ex] {} (B);
	\draw[] (B) to [xshift=0ex, yshift=0ex] node[below, xshift=0ex, yshift=0ex] {} (C);
	\draw[] (D) to [xshift=0ex, yshift=0ex] node[below, xshift=0ex, yshift=0ex] {} (C);
	\draw[] (A) to [xshift=0ex, yshift=0ex] node[below, xshift=0ex, yshift=0ex] {} (D);
	\draw[] (E) to [xshift=0ex, yshift=0ex] node[below, xshift=0ex, yshift=0ex] {} (F);
	\draw[] (F) to [xshift=0ex, yshift=0ex] node[below, xshift=0ex, yshift=0ex] {} (G);
	\draw[] (G) to [xshift=0ex, yshift=0ex] node[below, xshift=0ex, yshift=0ex] {} (H);
	\draw[] (H) to [xshift=0ex, yshift=0ex] node[below, xshift=0ex, yshift=0ex] {} (E);
	\draw[] (E) to [xshift=0ex, yshift=0ex] node[below, xshift=0ex, yshift=0ex] {} (A);
	\draw[] (F) to [xshift=0ex, yshift=0ex] node[below, xshift=0ex, yshift=0ex] {} (B);
	\draw[] (C) to [xshift=0ex, yshift=0ex] node[below, xshift=0ex, yshift=0ex] {} (G);
	\draw[] (D) to [xshift=0ex, yshift=0ex] node[below, xshift=0ex, yshift=0ex] {} (H);
\end{tikzpicture}
}
\qquad
\subfloat[A reverse-palm-tree, with timestamps of vertices.\label{fig:example_palmtree_dfs}]{
\begin{tikzpicture}[arrows=->, scale=.45]
	\node[squa, label={right : \tiny $1|16$}] (tA) {$A$};
	\node[squa, below=of tA, label={right : \tiny $2|15$}, yshift=0ex, xshift=0ex] (tB) {$B$};
	\node[circ, below=of tB, label={right : \tiny $3|14$}] (tC) {$C$};
	\node[squa, below=of tC, label={right : \tiny $4|13$}, xshift=0ex, yshift=0ex] (tD) {$D$};
	\node[squa, below=of tD, label={right : \tiny $5|12$}, xshift=0ex, yshift=0ex] (tE) {$E$};
	\node[squa, below left=of tE, label={right : \tiny $6|9$}, yshift=0ex, xshift=0ex] (tF) {$F$};
	\node[circ, below left=of tF, label={right : \tiny $7|8$}, xshift=0ex, yshift=0ex] (tG) {$G$};
	\node[squa, below right=of tE, label={right : \tiny $10|11$}, yshift=0ex, xshift=0ex] (tH) {$H$};

	\draw[thick] (tB) to [xshift=0ex, yshift=0ex] node[ xshift=0ex, yshift=0ex] {\tiny\bf tree} (tA);
	\draw[thick] (tC) to [xshift=0ex, yshift=0ex] node[ xshift=0ex, yshift=0ex] {\tiny\bf tree} (tB);
	\draw[thick] (tD) to [xshift=0ex, yshift=0ex] node[ xshift=0ex, yshift=0ex] {\tiny\bf tree} (tC);
	\draw[thick] (tE) to [xshift=0ex, yshift=0ex] node[ xshift=0ex, yshift=0ex] {\tiny\bf tree} (tD);
	\draw[thick] (tG) to [xshift=0ex, yshift=0ex] node[ xshift=0ex, yshift=0ex] {\tiny\bf tree} (tF);
	\draw[thick] (tF) to [xshift=0ex, yshift=0ex] node[ xshift=0ex, yshift=0ex] {\tiny\bf tree} (tE);
	\draw[thick] (tH) to [xshift=0ex, yshift=0ex] node[ xshift=0ex, yshift=0ex] {\tiny\bf tree} (tE);

	\draw[dashed] (tA) to [bend right=20, xshift=0ex, yshift=0ex] node[ xshift=0ex, yshift=0ex] {\tiny frond} (tF);
	\draw[dashed] (tD) to [bend right=30, xshift=0ex, yshift=0ex] node[ xshift=0ex, yshift=0ex] {\tiny frond} (tG);
	\draw[dashed] (tC) to [bend left=55, xshift=0ex, yshift=0ex] node[ xshift=0ex, yshift=-1ex] {\tiny frond} (tH);
\end{tikzpicture}
}
\qquad
\subfloat[The order of arcs' exploration.]{
\begin{tikzpicture}[arrows=->, scale=.1, node distance = 0 and 0]
	\node[] (eA) {$\text{\tiny 1.} (B,A)$};
	\node[below=of eA, yshift=0ex, xshift=0ex] (eB) {$\text{\tiny 2.} (C,B)$};
	\node[below=of eB] (eC) {$\text{\tiny 3.} (D,C)$};
	\node[below=of eC, xshift=0ex, yshift=0ex] (eD) {$\text{\tiny 4.} (E,D)$};
	\node[below=of eD, xshift=0ex, yshift=0ex] (eE) {$\text{\tiny 5.} (F,E)$};
	\node[below=of eE, yshift=0ex, xshift=0ex] (eF) {$\text{\tiny 6.} (A,F)$};
	\node[right=of eA, xshift=0ex, yshift=0ex] (eG) {$\text{\tiny 7.} (G,F)$};
	\node[below=of eG, yshift=0ex, xshift=0ex] (eH) {$\text{\tiny 8.} (D,G)$};
	\node[below=of eH, xshift=0ex, yshift=0ex] (eI) {$\text{\tiny 9.} (H,E)$};
	\node[below=of eI, yshift=0ex, xshift=0ex] (eL) {$\text{\tiny 10.} (C,H)$};
	\node[below=of eL, xshift=0ex, yshift=0ex] (eM) {$\text{\tiny 11.} (G,B)$};
	\node[below=of eM, yshift=0ex, xshift=0ex] (eN) {$\text{\tiny 12.} (H,A)$};

\end{tikzpicture}
}
\caption{A reverse-palm-tree (b), generated by reverse-DFS (c) starting at $A$.}\label{fig:arena-palm-tree-example}
\end{figure}

    \textbf{Reverse-DFS, palm-trees and jungles.}
    As we are dealing with $\alpha$graphs, it turns out we need to impose an \emph{opposite} direction \wrt that in which the arcs are explored;
    \ie at each step of the DFS, we actually choose an \emph{ingoing} arc to follow instead of an outgoing one.
    This reversal is due to the fact that, on one side, Player~$\square$ has no control on the arc choices of the opponent, and on the other side,
    we still aim at exploring the $\alpha$graph in a depth-first manner but meanwhile preserving $\alpha$reachability relations;
     we will see that we can achieve this but we have to reverse the direction of exploration
      so that to mimics the backward tread of computing a $\square$-attractor.
    Let us call the corresponding search algorithm, \emph{reverse}-DFS (think of it as if we had reversed the direction of each arc).
    A moment's reflection reveals that, if run on a directed graph, all the basic properties of the DFS are still there (by symmetry).
    For instance, if the vertices are numbered in the order in which they are reached during the reverse-DFS,
    \eg by $\textit{open}[v]:V\rightarrow \N$,	now a cross-link $(u, v)$ always has $\textit{open}[u] < \textit{open}[v]$.
    A forest of inward directed \emph{reverse-palm-trees} is constructed during a reverse-DFS.
    Let us call \emph{reverse-jungle} the underlying predecessor subgraph structure,
     that is a family of reverse-palm-trees comprising \emph{fronds} and \emph{cross-links}.
     Also notice that, if run on an $\alpha$graph having $V_\ocircle=\emptyset$,
       the reverse-palm-tree of a reverse-DFS is actually a $\square$-attractor strategy.
    Since we will only deal with the reversed variants, from now on in the forthcoming sections
    we shall refer to them simply as ``DFS", ``palm-trees" and ``jungles".
	
\begin{Exa}\label{ex:arena-palm-tree-example}
Figure~\ref{fig:arena-palm-tree-example} illustrates an $\alpha$graph $\mathcal{A}$ and its corresponding 
reverse-palm-tree decomposition, derived from a reverse depth-first search traversal.

Subfigure~\ref{fig:example_arena_dfs} presents the original $\alpha$graph, where directed arcs represent 
possible transitions between vertices, controlled either by Player~$\square$ or Player~$\ocircle$. 
This structure defines the alternating reachability conditions within the game.

Subfigure~\ref{fig:example_palmtree_dfs} depicts the reverse-palm-tree decomposition of the same graph. 
The numbers next to each vertex indicate the opening and closing timestamps assigned during the DFS traversal. 
Thick arcs denote tree arcs, which form the backbone of the depth-first search tree, 
while other arcs represent different structural relationships between vertices. 

The reverse-palm-tree decomposition provides a hierarchical view of the graph, also 
revealing dependencies in alternating reachability. 
By examining the DFS tree, one can determine how each vertex is explored and how the game structure influences traversal order. 
This example visually demonstrates how depth-first search interacts with alternating graphs and 
lays the groundwork for subsequent observations.
\end{Exa}

Example~\ref{ex:arena-palm-tree-example} illustrates how reverse depth-first traversal structures 
an $\alpha$graph into a reverse-palm-tree. However, while classical DFS correctly captures graph reachability in 
standard directed graphs, it does not necessarily preserve alternating reachability in $\alpha$graphs. 
This raises the need for an adaptation that respects the strategic constraints imposed by alternating moves. 
To address this, we now introduce $\alpha$DFS, a depth-first exploration strategy tailored to $\alpha$graphs.

\subsection{Key Ideas of $\alpha$DFS()}

The reverse-palm-tree decomposition provides a structured way to explore an $\alpha$graph, 
but it remains essential to analyze whether this structure correctly reflects alternating reachability. 
In classical DFS, any vertex is reachable from its descendants within a palm-tree. 
However, in an $\alpha$graph, reachability is governed by the interplay between Player~$\square$ and Player~$\ocircle$, 
meaning that a vertex’s strategic accessibility may differ from its position in the palm-tree. 
This motivates a closer examination of reachability within jungles and palm-trees derived from DFS.

With this in mind, let us attempt again to explore an \emph{$\alpha$graph} $\mathcal{A}$ using a classical reverse DFS. 
Let $\mathcal{J}$ be the resulting jungle, and let $T$ be any palm-tree in $\mathcal{J}$. 
An example is provided in Figure~\ref{fig:example_arena_dfs}, with the corresponding palm-tree $T$ shown in Figure~\ref{fig:example_palmtree_dfs}. 
In this representation, each vertex $v$ is annotated with its opening and closing timestamps using the notation 
$\langle \textit{open}[v] \rangle \mid \langle \textit{close}[v] \rangle$. 

To begin, observe that in any palm-tree $T = (V_T, A_T)$, graph reachability is trivial: 
for any two vertices $u, v \in T$, if $v$ is an ancestor of $u$ in $T$, 
then there exists a simple path from $u$ to $v$ within $T$, \ie $v$ is reachable from $u$ in $T$.

At this point, let us shift our focus from graph reachability to \emph{alternating reachability}. 
Unlike in standard graph traversal, a palm-tree $T$ constructed via classical DFS does not necessarily preserve $\alpha$reachability. 
For instance, consider the two vertices $F, B \in V_{\square}$ in the palm-tree $T$ shown in Figure~\ref{fig:example_palmtree_dfs}. 
Although $B$ is an ancestor of $F$ in $T$, Player~$\square$ has no strategy to ensure reaching $B$ from $F$. 
Regardless of Player~$\square$'s choices, Player~$\ocircle$ can always force the play away from $B$:

- Any play starting from $F$ must first reach $D$. If Player~$\square$ then chooses $(D, G)$, 
  Player~$\ocircle$ can move back to $F$ via $(G, F)$.  
  
- Alternatively, if Player~$\square$ plays $(D, C)$, Player~$\ocircle$ can respond with $(C, H)$, reaching $H$. 
  Once at $H$, the play is forced to return to $D$, preventing access to $B$.  

Thus, starting from $F$, Player~$\ocircle$ can always prevent Player~$\square$ from reaching $B$. 
This example illustrates that classical DFS structures do not inherently respect alternating reachability, 
motivating the need for a more refined depth-first search method adapted to $\alpha$graphs.

Our goal is to extend classical DFS, along with palm-trees and jungles, to $\alpha$graphs in a way that preserves reachability 
\emph{within} a suitably adapted notion of palm-trees. Ideally, a well-defined ``DFS for $\alpha$graphs" should satisfy the following property:  
For any such palm-tree $T$, if $u, v \in T$ and $v$ is an ancestor of $u$ in $T$, 
there must exist a strategy $\sigma_{\square} \in \Sigma^{\mathcal{A}}_{\square}$ that allows Player~$\square$ to eventually reach $v$ 
starting from $u$, \emph{without leaving $T$}, regardless of the adversarial strategy chosen by Player~$\ocircle$. 
		
The discussion above highlights a fundamental issue: classical DFS structures, such as palm-trees, 
do not necessarily preserve $\alpha$reachability. This motivates the need for a refined approach that respects 
the strategic interplay between Player~$\square$ and Player~$\ocircle$. 
A crucial aspect of adapting DFS to $\alpha$graphs is understanding how vertices are discovered and processed 
during traversal. In standard DFS, the search progression is typically represented by a vertex coloring scheme, 
which helps track the state of exploration. Before introducing modifications for $\alpha$DFS, 
we briefly recall this basic coloring mechanism, as in~\cite{Cormen2001}.

Imagine reverse-DFS runs on a directed graph, color the vertices during the search to indicate their state.
Initially each vertex $v$ is \emph{white} to mean unexplored, then $v$ becomes \emph{grey} when it is first discovered
  (\ie when $\textit{open}[v]$ is assigned), then \emph{black} when the search backtracks (\ie when $\textit{close}[v]$ is assigned).
 Each vertex changes color only twice, from white to grey and then blackened.

    A fundamental underlying invariant property of \emph{Safe-Alternating Depth-First Search~($\alpha$DFS)} goes as prescribed in the box below
			(this will be formally proved in Proposition~\ref{prop:safe-reachability}).
Recall that the exploration of the vertices goes backward like in reverse-DFS meanwhile building up a palm-tree $T$ in post-ordering;
now the task is precisely to decide \emph{which} particular post-ordering to follow, \ie \emph{when} to explore any given vertex.
    \begin{mybox}
	\textbf{Condition for Attaching a Vertex to the $\alpha$DFS Palm-Tree}
	
    During the $\alpha$DFS() exploration of an input $\alpha$graph $\A$ on vertex set $V$, a new vertex $u\in V$ is visited and attached to the $\alpha$DFS's palm-tree $T$ under formation (\ie that one comprising at least one grey vertex) only when the \emph{$T$-safe-$\alpha$reachability} of its root $r_T$
    	becomes guaranteed starting from $u$ in such a way that any safe-$\alpha$reachability finite play path can only move through the non-white vertices of $T$.

    This happens only after that a certain set of out-neighbours of $u$ becomes non-white in $T$:
    \emph{all} of $u$'s out-neighbours must have been colored grey or black if $u\in V_\ocircle$; and \emph{at least one} if $u\in V_\square$.
    \end{mybox}

    (So, safe-$\alpha$reachability is invariantly preserved in the palm-trees instead of just graph reachability)

    Of course we will need additional (non-trivial) arguments to ensure the algorithm runs in linear-time.
    For instance, when a new vertex $u\in V_\ocircle$ attaches to the $\alpha$DFS's palm-tree $T$ under	formation,
    the parent $\gamma$ of $u$ in $T$ must be chosen very carefully. The following \emph{attraction-rule} stands out.

    \begin{mybox}
	\textbf{Parent Selection Rules in $\alpha$DFS Palm-Tree}
	
    	During the $\alpha$DFS() exploration of the given input $\alpha$graph $\A$ on vertex set $V$,
    		assume that a new vertex $u\in V$ now attaches to the $\alpha$DFS's palm-tree $T$ under formation (\ie that comprising at least one grey vertex).
    		If $u\in V_\square$, the parent $\gamma$ of $u$ in $T$ can be any of the grey out-neighbours of $u$ in $T$ (\emph{$\square$-attraction-rule});
    		otherwise, if $u\in V_\ocircle$, the parent $\gamma$ of $u$ in $T$ is precisely the 
			\emph{Lowest Common Ancestor (LCA)} (which is grey colored at that time) of all 
			the out-neighbours of $u$ in $T$ (and all these must be non-white colored at that time) (\emph{$\ocircle$-attraction-rule}).
    \end{mybox}

		So, safe-$\alpha$reachability is preserved from $u\in V_\ocircle$ to the LCA of its out-neighbours.

		\begin{figure}[t!]
			\centering
			\begin{tikzpicture}[scale=.7]
					\node [squa, fill=black!10] (0) at (0, 7) {$r$};
					\node [style=none] (1) at (-9, 0) {};
					\node [style=none] (2) at (9, 0) {};
					\node [style=none] (3) at (9, 0) {};
					\node [style=none] (4) at (3, 2) {};
					\node [squa, fill=black!10] (5) at (1, 4) {$\gamma$};
					\node [style=none] (6) at (0, 2) {};
					\node [style=none] (7) at (-2.75, 2) {};
					\node [squa, fill=black!10] (8) at (3, 2) {$v$};
					\node [circ, fill=black!40] (9) at (0, 2) {$y$};
					\node [squa, fill=black!40] (10) at (-2.5, 2) {$x$};
					\node [style=none] (11) at (-6, 0) {};
					\node [style=none] (12) at (5, 0) {};
					\node [circ] (14) at (6, .4) {$u$};
					\node [style=none] (15) at (1.75, 2.25) {};
					\node [style=none] (16) at (2.25, 3) {};
					\node [style=none] (17) at (1.25, 3.25) {};
					\node [style=none] (18) at (1.5, 4.75) {};
					\node [style=none] (19) at (0.25, 5) {};
					\node [style=none] (20) at (0.5, 6) {};
					\draw [thick] (0.south west) to (1.center);
					\draw (0.south east) to node[above, xshift=-10ex, yshift=10ex] {$T$}  (3.center);
					\draw [thick] (1.center) to (12.center);
					\draw (12.center) to (3.center);
					\draw [thick] (5.south west) to (10.north east);
					\draw (5.south) to (9.north);
					\draw [thick] (10.south west) to (11.center);
					\draw [thick] (8.south east) to (12.center);
					\draw [arrows=->] (8.west) to (15.center);
					\draw [arrows=->] (15.center) to (16.center);
					\draw [arrows=->] (16.center) to node[below, xshift=0ex, yshift=0ex] {\tiny tree} (17.center);
					\draw [arrows=->] (17.center) to (5.south);
					\draw [arrows=->, dotted, bend left=45] (14.west) to node[xshift=0ex, yshift=0ex] {\tiny stalk} (8.south);
					\draw [arrows=->, dotted, bend left, looseness=.75] (14.south west) to node[xshift=0ex, yshift=0ex] {\tiny stalk} (9.south east);
					\draw [arrows=->, dotted, bend right=315, looseness=0.75] (14.south) to node[xshift=0ex, yshift=0ex] {\tiny stalk} (10.south);
					\draw [arrows=->] (5.north) to (18.center);
					\draw [arrows=->] (18.center) to (19.center);
					\draw [arrows=->] (19.center) to (20.center);
					\draw [arrows=->] (20.center) to (0.south);
					\draw [arrows=->, dashed, thick, bend right=25] (14.north) to node[xshift=1.5ex, yshift=0ex] {\tiny tree} (5.east);
				\end{tikzpicture}
				  \caption{An illustration of the \emph{$\ocircle$-attraction-rule} during \textit{$\alpha$DFS()}}\label{fig:alphaDFS}
		\end{figure}

\begin{Exa}
Figure~\ref{fig:alphaDFS} illustrates the application of the \emph{$\ocircle$-attraction-rule} during 
the execution of $\alpha$DFS(). 
The diagram represents a partial $\alpha$DFS exploration, 
highlighting how a newly discovered vertex $u$ is incorporated into 
the $\alpha$DFS palm-tree $T$ under formation.

The tree structure $T$ is rooted at $r$ and consists of a set of vertices, including $\gamma$, 
$x$, $y$, and $v$, which are colored either 
grey or black. The vertex $u$, currently unprocessed, is depicted outside the main tree and is connected to 
multiple vertices 
within $T$ via \emph{stalk} arcs (dotted arrows), indicating that $u$ has out-neighbors that are 
already part of the search tree.

According to the \emph{$\ocircle$-attraction-rule}, since $u$ belongs to $V_\ocircle$, 
its parent in the palm-tree $T$ must be 
precisely the \emph{Lowest Common Ancestor (LCA)} of all its out-neighbors within $T$. 
In the figure, this corresponds to the 
vertex $\gamma$, which at the time of attachment is grey. 
The rule ensures that the attachment respects alternating reachability 
constraints, preventing premature or unsafe insertions into the tree.

The figure also differentiates between different types of arcs:

- tree arcs (solid arrows): represent standard depth-first traversal paths within the forming palm-tree.

- stalk arcs (dotted arrows): indicate connections from $u$ to existing vertices in $T$, representing potential entry points.

- dashed arcs: highlight paths that confirm attachment to the appropriate LCA.

This structure guarantees that, when $u$ is added to $T$, it does not disrupt the alternating reachability 
properties necessary for correct $\alpha$DFS execution. 
The $\ocircle$-attraction-rule thus plays a crucial role in maintaining the integrity of 
the depth-first exploration strategy tailored for $\alpha$graphs.

\end{Exa}

A detailed low-level description of the algorithm comes next, where some additional technical machinery
			(\eg counters, stacks, and disjoint-sets data structure) is employed for running time efficiency.

    As it starts to make sense, a major technical issue will be that to perform LCAs lookups efficiently.
	
    \subsection{Description of $\alpha$DFS}
The main procedure, \textit{$\alpha$DFS()}, 
is defined in Algorithm~\ref{algo:aDFS}, while vertex visitation is handled by the 
subprocedure \textit{$\alpha$DFS-visit()} (Procedure~\ref{algo:aDFS-visit}). 
The corresponding pseudocode is provided in Algorithm~\ref{algo:aDFS} and Procedure~\ref{algo:aDFS-visit} below.

    \begin{algorithm}[H]
    \caption{Safe-Alternating DFS}\label{algo:aDFS}
    \DontPrintSemicolon
    \nonl \SetKwProg{Fn}{Procedure}{}{}
    \normalsize
    \Fn{$\textit{$\alpha$DFS}(\mathcal{A})$}{
    		\SetKwInOut{Input}{input}
    		\SetKwInOut{Output}{output}

    \Input{An $\alpha$graph $\mathcal{A}=(V, A, \langle V_{\ocircle}, V_{\square} \rangle )$.}
    \Output{An $\alpha$jungle $\mathcal{J}_{\mathcal{A}}$.}
    $A_{\pi}, A_{f}, A_{s}, A_{c}\leftarrow \emptyset$; \label{algo:aDFS:l1}\;
    \ForEach{$u\in V$}{ \label{algo:aDFS:l2}
    	$\textit{open}[u]\leftarrow +\infty$; \label{algo:aDFS:l3}\;
    	$\textit{close}[u]\leftarrow +\infty$;\label{algo:aDFS:l4}\;
    	$\textit{rSt}[u]\leftarrow \emptyset$; \label{algo:aDFS:l5}\;
    	\If{$u\in V_{\ocircle}$}{ \label{algo:aDFS:l6}
    	$\textit{cnt}[u]\leftarrow |N_{\A}^{\text{out}}(u)|$; \label{algo:aDFS:l7}
    	}
    }
    $\textit{time}\leftarrow 0$; \tcp{global time variable} \label{algo:aDFS:l8}
    \ForEach{$u\in V_{\square}$ \label{algo:aDFS:l9}}{
    	\If{$\textit{open}[u]=+\infty$ \label{algo:aDFS:l10}}{
    		$\textit{$\alpha$DFS-visit}(u, \A)$;\label{algo:aDFS:l11} \;
    	}
    }
    \ForEach{$u\in V_{\ocircle}$\label{algo:aDFS:l12}}{
    	\If{$\textit{open}[u]=+\infty$\label{algo:aDFS:l13}}{
    		$\textit{open}[u]\leftarrow \textit{time}$; \; \label{algo:aDFS:l14}
    		$\textit{close}[u]\leftarrow \textit{time}$; \; \label{algo:aDFS:l14}
    		$\textit{time}\leftarrow \textit{time}+1$; \label{algo:aDFS:l15}
    	}
    }
    $A'\leftarrow A_{\pi} \cup A_{f} \cup A_{s} \cup A_{c}$; \label{algo:aDFS:l16}\;
    \Return{$\mathcal{J}_{\mathcal{A}}\leftarrow (V, A', (V_{\square}, V_{\ocircle}))$ }; \label{algo:aDFS:l17}
    }
    \end{algorithm}

    The starting point for describing how everything works is recalling the reverse-DFS. In fact
     $\textit{$\alpha$DFS}(\A)$ (Algo.~\ref{algo:aDFS}) can be viewed as a gamification of the latter,
     in the sense that, if $V_\ocircle=\emptyset$, it works like a reverse-DFS and the output forest $\J_{\A}$ is a jungle.

     Indeed, given an $\alpha$graph $\A$ on vertex set $V=V_\square\cup V_\ocircle$, a forest $\alpha$graph $\J_{\A}$
     can be built during the search process (like the traditional DFS constructs a jungle) and returned as output.
    So $\J_{\A}$ will comprise a forest of trees, each called \emph{alternating palm-tree ($\alpha$palm-tree)}, having \emph{fronds} and \emph{cross-links}.

    During the exploration, arcs $(u,v)\in A$ will be classified into four categories according to the state (color) of the tail vertex $u$ that is touched when
    the arc is first explored, namely, tree-arcs $A_{\pi}$ (white), fronds $A_{f}$ (grey), stalk-arcs $A_{s}$ (white $u\in V_\ocircle$), and cross-links $A_{c}$ (black);
    at the end, their union $A'$ will be the whole arc set of what we call the \emph{alternating jungle ($\alpha$jungle)} $\J_{\A}$.

    An index, named $\textit{open}:V\rightarrow\N\cup\{+\infty\}$,
    	timestamps the vertices in the order in which they are firstly visited
    (\ie the timestamp opens at the beginning of the visiting subprocedure);
    initially all vertices are unvisited, so $\forall{u\in V}\; \textit{open}[u]\leftarrow +\infty$.
    Another index, $\textit{close}:V\rightarrow\N\cup\{+\infty\}$, timestamps the vertices in the order in which they are closed
    (\ie the closing assignment happens at the end of the visiting subprocedure).
    In the pseudocode we assume $\textit{open}[], \textit{close}[], \textit{rSt}[], \textit{cnt}[], time$ are all global variables.

    \begin{procedurealgo}[h!]
    \caption{Visit Procedure of Safe-Alternating DFS}\label{algo:aDFS-visit}
    \DontPrintSemicolon
    \nonl \SetKwProg{Fn}{Procedure}{}{}
    \normalsize
    \Fn{$\textit{$\alpha$DFS-visit}(v, \A)$}{
    		\SetKwInOut{Input}{input}
    		\SetKwInOut{Output}{output}
     \Input{One vertex $v\in V$ of $\A$.}
    $\textit{open}[v]\leftarrow (\textit{time}\leftarrow \textit{time}+1)$; \label{algo:aDFS-visit:l1} \;
    \ForEach{$u\in N_\A^{\text{in}}(v)$}{ \label{algo:aDFS-visit:l2}
    	\If{$\textit{open}[u] = +\infty$}{ \label{algo:aDFS-visit:l3}
    		\If{$u\in V_\square$}{ \label{algo:aDFS-visit:l6}
    			add $(u,v)$ to $A_{\pi}$;\label{algo:aDFS-visit:l4}\;
    			$\textit{$\alpha$DFS-visit}(u, \A)$; \label{algo:aDFS-visit:l5} \;
    		}\Else{ \label{algo:aDFS-visit:l6}
    			$\textit{cnt}[u]\leftarrow \textit{cnt}[u]-1$; \label{algo:aDFS-visit:l7} \;
    			\If{$\textit{cnt}[u]=0$ \textbf{and} $\exists$(LCA of $N_{\A}^{\text{out}}(u)$ in $(V, A_{\pi})$)}{ \label{algo:aDFS-visit:l8}
    					$\gamma\leftarrow $ the LCA of $N_{\A}^{\text{out}}(u)$ in $(V, A_{\pi})$; \label{algo:aDFS-visit:l9}\;
    					$\textit{rSt}[\gamma].\textit{push}(u)$; \label{algo:aDFS-visit:l10} \;
    			}
    		}
    	}\ElseIf{$\textit{open}[u]<+\infty \textbf{ and } \textit{close}[u]=+\infty$ \label{algo:aDFS-visit:l11}}{
    		add $(u,v)$ to $A_{f}$; \label{algo:aDFS-visit:l12} \;
    		\lElse{
    			add $(u,v)$ to $A_{c}$; \label{algo:aDFS-visit:l13}
    		}
    	}
    }
	\tcp{Check the ready-stack of $v$, \ie $\textit{rSt}[v]$}
    \While{$\textit{rSt}[v]\neq\emptyset$}{ \label{algo:aDFS-visit:l14}
    	$u\leftarrow \textit{rSt}[v].\textit{pop}()$;\label{algo:aDFS-visit:l15} \tcp{$u\in V_{\ocircle}$}
    	add $(u,v)$ to $A_{\pi}$; \label{algo:aDFS-visit:l16}\;
    	\textit{\bf for each} $t\in N^{\text{out}}_{\mathcal{A}}(u)$ \textit{\bf{do}} add $(u,t)$ to $A_{\textit{stalk}}$; \label{algo:aDFS-visit:l17}\;
    	$\textit{$\alpha$DFS-visit}(u, \A)$; \label{algo:aDFS-visit:l18} \;
    }
    $\textit{close}[v]\leftarrow (\textit{time}\leftarrow \textit{time}+1)$;\label{algo:aDFS-visit:l19}\;
    }
    \end{procedurealgo} 

    We say vertex $u\in V$ is \emph{active (grey)} if $\textit{open}[u]<+\infty$ and $\textit{close}[u]=+\infty$,
    say that $u$ has been \emph{visited (black)} if $\textit{open}[u]<+\infty$ and $\textit{close}[u]<+\infty$,
    and that $u$ is \emph{unvisited (white)} if $\textit{open}[u]=\textit{close}[u]=+\infty$.

    Now, imagine that the search exploration proceeds by visiting and backtracking vertices like in a reverse-DFS.
		Any $u\in V_\square$ is visited, and so it joins $\J_{\A}$, as soon as it is firstly discovered in the in-neighbourhood of \emph{some} active vertex (\ie precisely as in the reverse-DFS).

    Let's say by convention that any vertex $u\in V$ \emph{joins} $\J_{\A}$ precisely when it becomes active and the tree-arc $(u,v)$ is added to $A_{\pi}$ for some $v\in V$.

    	The \emph{$\ocircle$-attraction-rule} (\ie that allowing any $u\in V_\ocircle$ to be visited) is more involved:
    	 any $u\in V_{\ocircle}$ becomes active joining $\J_{\A}$, by attaching to some parent vertex $\pi_u$,
    	only when \emph{all} of $u$'s out-neighbours $v\in N^{\text{out}}_{\A}(u)$ have already~done so.
			So the visiting step of any circled $u$ has to be delayed \wrt the (possibly repeated) discovery of $u$ as an in-neighbour of (possibly many) active vertices $v$ such that $(u,v)\in A$; the exact moment being when the search backtracks, after the lastly visited out-neighbour $v\in N^{\text{out}}_{\A}(u)$, up to the corresponding parent vertex $\pi_u$.
    And when $u\in V_{\ocircle}$ joins $\J_\A$ with parent $\pi_u$ (\ie if $u\in V_{\ocircle}$
    	\emph{and} $(u,\pi_u)\in A_{\pi}$ for some $\pi_u\in V$), then $\pi_u$ is prescribed by the $\ocircle$-rule to be the \emph{LCA} $\gamma$
    	of $N_{\A}^{\text{out}}(u)$ in the $\alpha$palm-tree under formation; at that point
    	\emph{all} of the original outgoing arcs of $u$ are labeled~\emph{stalk-arcs}.

    	Indeed besides fronds and cross-links, $\alpha$palm-trees have an additional arc category:
    	\emph{stalk-arcs}, that are all the original outgoing arcs of any $u\in V_\ocircle$ which joined $\J_\A$.

    	Notice that if $u\in V_{\ocircle}$ joins $\J_\A$ with parent $\pi_u$, and since $\pi_u$ is the LCA of $N_{\A}^{\text{out}}(u)$,
    		then the arc $(u,\pi_u)\in A_{\pi}$ may be a totally brand new arc,
    			\ie it might not have been in the original arc set $A$ of the input $\alpha$graph $\A$
    				(in that case $A_\pi \not\subseteq A$ and $(u,\pi_u)$ is \emph{not} labeled as a stalk-arc).
    				The possibility that $A_\pi \not\subseteq A$ is a distinctive point with respect to reverse-DFS and Tarjan's jungles,
    				where all tree-arcs belonged to the input directed graph.

    In order to implement the \emph{$\ocircle$-attraction-rule} efficiently, an additional counter of out-neighbour vertices $\textit{cnt}:V_\ocircle \rightarrow \N$ is employed, constantly checked and updated.
    		The following invariant $I_\textit{cnt}$ is kept maintained:
    \[\forall{u\in V_\ocircle}\; \textit{cnt}[u]=\big|\{v\in N^{\text{out}}_\A(u)\mid \textit{open}[v]=+\infty\}\big|. \tag{$I_\textit{cnt}$}\]
    Also, for each $v\in V$ it is employed a LIFO stack of vertices named $\textit{rSt}[v]$ (named, the ready stack).
    Its role, during the $\textit{$\alpha$DFS-visit()}$ subprocedure, is to memorize that a certain vertex $\pi_u\in V$ had been
    	identified as the parent of some other vertex $u\in V_\ocircle$ (\ie when $\textit{cnt}[u]=0$ and $\pi_u=\gamma$
    		is the LCA of $N	_{\A}^{\text{out}}(u)$); at that point $u$ would be promptly pushed to the ready stack $\textit{rSt}[\pi_u]$.
    	Then $u$ will have to join $\J_\A$ when visited by the search, this happens when the visit backtracks from $u$ up to his parent $\pi_u$.

By construction, the $\ocircle$-attraction-rule ensures that safe-$\alpha$reachability 
is preserved within the $\alpha$palm-tree decomposition, as shown in Proposition~\ref{prop:safe-reachability}. 
In particular, the tree structure induced by $\alpha$DFS guarantees that every vertex remains safely connected 
to its ancestors according to the game rules.

Additionally, the graph \( (V, A_{\pi}) \), where each vertex \( u \) follows its designated 
parent \( \pi(u) \) as assigned by $\alpha$DFS, forms a forest. This follows from the fact that every 
vertex (except the root of each tree) is attached exactly once and never revisited. 
A more detailed structural characterization of the resulting jungle \( \mathcal{J}_{\mathcal{A}} \) 
is deferred to Proposition~\ref{prop:aDFS-jungle}.

The following proposition formalizes this safe reachability guarantee within each 
$\alpha$palm-tree $\mathcal{P}_i$ of the decomposition. 
Here, $\{\mathcal{P}_i\}_{i=1}^k$ denotes the vertex-disjoint collection of $\alpha$palm-trees 
constructed during the $\alpha$DFS exploration.

    		\begin{Prop}\label{prop:safe-reachability}
    			Assume that $\alpha$DFS() runs on a given input $\alpha$graph $\A$.
    		Consider the forest of $\alpha$palm-trees $\{\mathcal{P}_i\}_{i=1}^k$ that are constructed during the visiting process; say that
    		$\mathcal{P}_i=(V_i, A_i, \langle {V_{\square}}_i, {V_{\ocircle}}_i \rangle )$ is the $i$-th $\alpha$palm-tree,
				on vertex set $V_i$ and arc set $A_i$ for each $i \in [k]$.
    		For any two vertices $u,v\in V_i$ any $i\in [k]$, if $u$ is a descendant of $v$ in $\mathcal{P}_i$,
    			then $v$ is $V_i$-safe-$\alpha$reachable from $u$ \wrt $\A$. Particularly,
    				this holds thanks to the following strategy $\sigma_{\square} \in \Sigma^{\A}_{\square}$,
    				where $\pi(u)$ denotes the parent of any $u\in V_\square$ in the forest $(V, A_{\pi})$:
    				\[
    				\forall{u\in V_\square}\;	\sigma_{\square}(u) \doteq \left\{
    					\begin{array}{ll}
    						\pi(u), 	&
    								\text{if } u \text{ is \emph{not} the root of any $\alpha$palm-tree } \mathcal{P}_i;  \\
    						\text{any } u'\in N^{\text{out}}_{\A}(u), & \text{if } u \text{ is the root of some $\alpha$palm-tree } \mathcal{P}_i.
    					\end{array}
    						\right.
    				\]
    		\end{Prop}
    		\begin{proof}
    		Assume $u,v\in V_i$ where $u$ is a descendant of $v$ in the $\alpha$palm-tree $\mathcal{P}_i$, for some $i\in [k]$ fixed arbitrarily.
    		Recall that during the $\alpha$DFS() all vertices are given an index so that
    			$\textit{open}[v]<\textit{open}[u]$ if $v$ is a proper ancestor of $u$ in some $\alpha$palm-tree.
    		Let us proceed arguing by induction on $\textit{open}[u]$.
    		Let $z\doteq \min_{x\in V_{i}} \textit{open}[x]$ be the vertex with minimum index in $\mathcal{P}_i$.
    		Assume $\textit{open}[u]=z$ as a base case. So, $u$ is the root of $\mathcal{P}_i$.
    		Then $v=u$, so there is actually nothing to prove. Now, let $\textit{open}[u]>z$. Let w.l.o.g $u\neq v$.
    		Assume as induction hypothesis the thesis for every vertex $x\in V_i$ such that $\textit{open}[x] < \textit{open}[u]$.

    		Let us break the forthcoming analysis in two cases, according to whether $u\in V_\square$ or $u\in V_\ocircle$.
    		\begin{itemize}
    		\item If $u\in V_\square$, since $u$ is not the root of $\mathcal{P}_i$,
    			then $\sigma_{\square}(u)=\pi(u)$.
    		By construction, $\textit{open}[\pi(u)] < \textit{open}[u]$.
    		Since $\pi(u)$ is the parent of $u$ in $\mathcal{P}_i$ and $u\neq v$,
    			then $\pi(u)$ is still a descendant of $v$ in $\mathcal{P}_i$ (possibly, $\pi(u)=v$);
    		thus, by induction hypothesis:
    			\[\sigma_{\square}:\pi(u)\overset{V_{i}}{\leadsto} v.\]
    		Since $\sigma_{\square} : u \overset{V_{i}}{\leadsto} \pi(u)$
    			and $\sigma_{\square}:\pi(u)\overset{V_{i}}{\leadsto} v$, therefore by composition $\sigma_{\square} : u \overset{V_{i}}{\leadsto} v$.

    		\item If $u\in V_\ocircle$, recall that by definition of $\alpha$DFS,	$\pi(u)$ is the LCA of the out-neighbours of $u$~in~$\A$,
    		\ie the LCA of $N^{\text{out}}_{\A}(u) = \{u'\in V\mid (u,u')\in A_{s}\}$.
    		Fix some $u'\in N^{\text{out}}_{\A}(u)$, arbitrarily.
    		Notice that $u'$ is still a descendant of $\pi(u)$ in $\mathcal{P}_i$ (possibly, $u' = \pi(u)$),
    		just because $\pi(u)$ is the LCA of $N^{\text{out}}_{\A}(u)$ in $\mathcal{P}_i$.
    		Thus, since $\pi(u)$ is a descendant of $v$ in $\mathcal{P}_i$ (possibly, $\pi(u)=v$),
    			then by transitivity $u'$ is also a descendant of $v$ in $\mathcal{P}_i$.
    		And, by definition of $\alpha$DFS, it must be that $\textit{open}[u']<\textit{open}[u]$.
    		Therefore, by induction hypothesis: \[\sigma_{\square}:u'\overset{V_{i}}{\leadsto} v.\]
    		Since $u'$ was chosen arbitrarily, the latter assertion holds for every $u'\in N^{\text{out}}_{\A}(u)$;	so, $\sigma_{\square}:u\overset{V_{i}}{\leadsto} v$.
    		\end{itemize}
    		This concludes the inductive step of the proof. So, anyway, $\sigma_{\square}:u\overset{V_{i}}{\leadsto} v$.

        It's also clear at this point that, at anytime during the execution of $\alpha$DFS(), any such safe-$\alpha$reachability finite play path
          (that goes from descendants up to ancestors) can only move through the non-white vertices of its $\alpha$palm-tree.
    		\end{proof}

    			\paragraph{More Details.}
    								Let us further provide some lower-level implementation details of $\alpha$DFS~(Algo.~\ref{algo:aDFS}).

    	  Concerning stacks and counters, $\textit{rSt}[u]$ is initialized to be empty for every $u\in V$ and,
				for every $u\in V_\ocircle$, it is initialized $\textit{cnt}[u]\leftarrow |N_{\A}^{\text{out}}(u)|$
				(see lines~\ref{algo:aDFS:l5}-\ref{algo:aDFS:l7} of Algo.~\ref{algo:aDFS}).
    	  Then $\textit{cnt}[u]$ is decremented whenever some out-neighbour $v$ of $u$ is visited during the search process.
    	  When $\textit{cnt}[u]=0$ (see line~\ref{algo:aDFS-visit:l8} of Proc.~\ref{algo:aDFS-visit}), all out-neighbours of $u$ have already joined the $\alpha$jungle $\J_\A$.

    			Notice, if any two out-neighbours of $u$ belong to two distinct $\alpha$palm-trees in $\J_\A$,
    			there is no way to preserve safe-$\alpha$reachability because Player~$\ocircle$ might choose to move from $u$ to any of the two shafts at will,
    			and the LCA $\gamma$ of $N^{\text{out}}_\A(u)$ might not exist in $(V, A_{\pi})$;
    			still, if all out-neighbours of $u$ belong to the same $\alpha$palm-tree, the LCA $\gamma$ does exist in $(V, A_{\pi})$.
    			So, when $\textit{cnt}[u]=0$,	firstly we seek for the LCA $\gamma$ and if it exists	we push $u$ on top of $\textit{rSt}[\gamma]$ (\cfr lines~\ref{algo:aDFS-visit:l8}-\ref{algo:aDFS-visit:l10} of $\textit{$\alpha$DFS-visit()}$, Proc.~\ref{algo:aDFS-visit}).

    			In so doing, $u\in V_\ocircle$ will join $\J_\A$ only when \textit{$\alpha$DFS-visit()} backtracks,
    		  from the last out-neighbour $v$ of $u$ that has been visited, up to $\gamma$ (possibly $\gamma=v$).
    			At that point (see lines~\ref{algo:aDFS-visit:l14}-\ref{algo:aDFS-visit:l18}), as $\textit{rSt}[\gamma]$ will be checked and $u$ will be found therein, $(u,\gamma)$ will be added to $A_{\pi}$;
    			and, for each $t\in N^{\text{out}}_\A(u)$ the arc $(u,t)$ will be added to $A_{s}$ (possibly, $(u,\gamma)\in A_{\pi}\cap A_{s}$).
    			Finally $\textit{$\alpha$DFS-visit}(u, \A)$ will be invoked for recursively visiting $u$. 
				In this way every vertex is visited exactly once.

    		During $\textit{$\alpha$DFS-visit}(v, \A)$, when it is explored some in-neighbour $u$ of $v$ such that $\textit{open}[u]\neq+\infty$,
    		if $u$ is still active (grey) then $(u,v)$ is added to the fronds $A_{f}$,
    		otherwise $u$ is inactive (black) and $(u,v)$ goes to cross-links $A_{c}$.

    There's still one detail which is worth mentioning as it helps keeping the presentation smooth.
    Firstly all $u\in V_\square$ are considered as roots of the $\alpha$palm-trees,
    \ie no $u\in V_\ocircle$ ever becomes a root of an $\alpha$palm-tree due
    to lines~\ref{algo:aDFS:l9}-\ref{algo:aDFS:l11} of $\textit{$\alpha$DFS()}$ (Algo.~\ref{algo:aDFS}).
    After the visiting is completed, for each $u\in V_{\ocircle}$ which still remained unvisited,
    	$\textit{open}[u]$ is assigned incrementally and the visiting process is not invoked anymore.

    Indeed, w.l.o.g we can assume that for all $v\in V$ $|N^{\text{out}}_\A(v)|\geq 2$. For this we just preprocess $\A$ as follows:
    	for any $v\in V$, if $N^{\text{out}}_\A(v)=\emptyset$, remove $v$ from the $\alpha$graph; if $N^{\text{out}}_\A(v)=\{v'\}$ is a singleton,
    	add $(u,v')$ to $A$ for each $u\in N^{\text{in}}_\A(v)$ and then remove $v$ from the $\alpha$graph.
    So doing, observe that even if $\textit{$\alpha$DFS-visit}(v,\A)$ would've been invoked for some $v\in V_{\ocircle}$,
    	say at line~\ref{algo:aDFS:l14} of $\textit{$\alpha$DFS()}$,
    there would've been no actual $\alpha$palm-tree to visit,
    	\ie no vertex $u$ such that $(u,v)\in A_{\pi}$. Of course all reachability relations are preserved after the preprocessing.
    So this self-reduction is fine, and it keeps simpler the presentation of the algorithm.

    This ends the detailed description of $\textit{$\alpha$DFS()}$ (Algo.~\ref{algo:aDFS}).
    Let us now begin to analyze its complexity.

    \begin{figure}[tb]
\centering
\subfloat[An $\alpha$graph $\mathcal{A}$.\label{fig:example_arena_A-PT}]{
\begin{tikzpicture}[arrows=->, scale=.7]
	\node[circ] (A) {$C$};
	\node[squa, right=of A, yshift=0ex, xshift=0ex] (B) {$B$};
	\node[squa, below=of B] (C) {$A$};
	\node[squa, left=of C, xshift=0ex, yshift=0ex] (D) {$H$};
	\node[squa, above left=of A, xshift=0ex, yshift=0ex] (E) {$E$};
	\node[squa, above right=of B, yshift=0ex, xshift=0ex] (F) {$D$};
	\node[squa, below right=of C, xshift=0ex, yshift=0ex] (G) {$G$};
	\node[circ, below left=of D, yshift=0ex, xshift=0ex] (H) {$F$};

	\draw[] (A) to [xshift=0ex, yshift=0ex] node[below, xshift=0ex, yshift=0ex] {} (B);
	\draw[] (B) to [xshift=0ex, yshift=0ex] node[below, xshift=0ex, yshift=0ex] {} (C);
	\draw[] (D) to [xshift=0ex, yshift=0ex] node[below, xshift=0ex, yshift=0ex] {} (C);
	\draw[] (A) to [xshift=0ex, yshift=0ex] node[below, xshift=0ex, yshift=0ex] {} (D);
	\draw[] (E) to [xshift=0ex, yshift=0ex] node[below, xshift=0ex, yshift=0ex] {} (F);
	\draw[] (G) to [xshift=0ex, yshift=0ex] node[below, xshift=0ex, yshift=0ex] {} (F);
	\draw[] (H) to [xshift=0ex, yshift=0ex] node[below, xshift=0ex, yshift=0ex] {} (G);
	\draw[] (H) to [xshift=0ex, yshift=0ex] node[below, xshift=0ex, yshift=0ex] {} (E);
	\draw[] (A) to [xshift=0ex, yshift=0ex] node[below, xshift=0ex, yshift=0ex] {} (E);
	\draw[] (F) to [xshift=0ex, yshift=0ex] node[below, xshift=0ex, yshift=0ex] {} (B);
	\draw[] (C) to [xshift=0ex, yshift=0ex] node[below, xshift=0ex, yshift=0ex] {} (G);
	\draw[] (H) to [xshift=0ex, yshift=0ex] node[below, xshift=0ex, yshift=0ex] {} (D);
\end{tikzpicture}
}
\qquad
\subfloat[The $\alpha$palm-tree generated by $\alpha$DFS rooted at $A$,
		with timestamps of vertices and labelled arcs.\label{fig:example_alternating-palm-tree} ]{
\begin{tikzpicture}[arrows=->, scale=.6]
	\node[squa, label={above : \tiny $1|16$}] (tA) {$A$};
	\node[squa, below=of tA, label={right : \tiny $2|9$}, yshift=0ex, xshift=0ex] (tB) {$B$};
	\node[circ, below left=of tB, label={below left, xshift=1.5ex, yshift=.25ex: \tiny $12|13$}, yshift=0ex, xshift=0ex] (tC) {$C$};
	\node[squa, below=of tB, label={right : \tiny $3|8$}, yshift=0ex, xshift=0ex] (tD) {$D$};
	\node[squa, below left=of tD, label={right : \tiny $4|5$}, yshift=-3ex, xshift=0ex] (tE) {$E$};
	\node[squa, below right=of tD, label={right :\tiny $6|7$}, yshift=-1ex, xshift=0ex] (tG) {$G$};
	\node[circ, below right=of tG, label={right : \tiny $14|15$}, yshift=1ex, xshift=-1ex] (tF) {$F$};
	\node[squa, left=of tC, label={below left, xshift=2ex, yshift=.5ex: \tiny $10|11$}, yshift=-2ex, xshift=-5ex] (tH) {$H$};

	\draw[thick] (tB) to [xshift=0ex, yshift=0ex] node[right,xshift=0ex, yshift=0ex] {\tiny\bf tree} (tA);
	\draw[dotted] (tC) to [xshift=0ex, yshift=0ex] node[xshift=0ex, yshift=0ex] {\tiny stalk} (tB);
	\draw[thick] (tD) to [xshift=0ex, yshift=0ex] node[right,xshift=0ex, yshift=0ex] {\tiny\bf tree} (tB);
	\draw[thick] (tE) to [xshift=0ex, yshift=0ex] node[right,xshift=0ex, yshift=0ex] {\tiny\bf tree} (tD);
	\draw[dotted] (tC) to [xshift=0ex, yshift=0ex] node[xshift=0ex, yshift=-1ex] {\tiny stalk} (tE);
	\draw[dotted] (tF) to [bend left=20, xshift=0ex, yshift=0ex] node[xshift=0ex, yshift=0ex] {\tiny stalk} (tE);
	\draw[thick] (tG) to [xshift=0ex, yshift=0ex] node[right,xshift=0ex, yshift=0ex] {\tiny\bf tree} (tD);
	\draw[dashed] (tA) to [bend left = 40, xshift=0ex, yshift=0ex] node[xshift=0ex, yshift=0ex] {\tiny frond} (tG);
	\draw[dotted] (tF) to [xshift=0ex, yshift=0ex] node[xshift=0ex, yshift=0ex] {\tiny stalk} (tG);
	\draw[thick] (tH) to [bend left=10, xshift=0ex, yshift=0ex] node[left,xshift=0ex, yshift=0ex] {\tiny\bf tree} (tA);
	\draw[dotted] (tC) to [xshift=0ex, yshift=0ex] node[xshift=0ex, yshift=.5ex] {\tiny stalk} (tH);
	\draw[dotted] (tF) to [bend left=40, xshift=0ex, yshift=0ex] node[xshift=0ex, yshift=0ex] {\tiny stalk} (tH);
	\draw[thick] (tC) to [bend left=10, xshift=0ex, yshift=0ex] node[xshift=-2ex, yshift=0ex] {\tiny\bf tree} (tA);
	\draw[thick] (tF) to [bend right=50, xshift=0ex, yshift=0ex] node[xshift=-.75ex, yshift=-2ex] {\tiny\bf tree} (tA);
\end{tikzpicture}
}
\qquad
\subfloat[The order of arcs' exploration.\label{fig:example_alternating-palm-tree-order}]{
\begin{tikzpicture}[arrows=->, scale=.1, node distance = 0 and 0]
	\node[] (eA) {$\text{\tiny 1.} (B,A)$};
	\node[below=of eA, yshift=0ex, xshift=0ex] (eB) {$\text{\tiny 2.} (D,B)$};
	\node[below=of eB, yshift=0ex, xshift=0ex] (eC) {$\text{\tiny 3.} (E,D)$};
	\node[below=of eC, yshift=0ex, xshift=0ex] (eD) {$\text{\tiny 4.} (C,E)$};
	\node[below=of eD, yshift=0ex, xshift=0ex] (eE) {$\text{\tiny 5.} (F,E)$};
	\node[below=of eE, yshift=0ex, xshift=0ex] (eE2) {$\text{\tiny 6.} (G,D)$};
	\node[right=of eA, yshift=0ex, xshift=0ex] (eG) {$\text{\tiny 7.} (A,G)$};
	\node[below=of eG, yshift=0ex, xshift=0ex] (eGA) {$\text{\tiny 8.} (F,G)$};
	\node[below=of eGA, yshift=0ex, xshift=0ex] (eF) {$\text{\tiny 9.} (C,B)$};
	\node[below=of eF, yshift=0ex, xshift=0ex] (eH) {$\text{\tiny 10.} (H,A)$};
	\node[below=of eH, yshift=0ex, xshift=0ex] (eHC) {$\text{\tiny 11.} (C,H)$};
	\node[below=of eHC, yshift=0ex, xshift=0ex] (eHF) {$\text{\tiny 12.} (F,H)$};
\end{tikzpicture}
}
\caption{An $\alpha$graph (a) and an $\alpha$palm-tree (b) generated by $\alpha$DFS (c).}\label{fig:palm_tree_example}
\end{figure}

    \begin{figure}[t!]
\begin{minipage}{\columnwidth}
\centering
\subfloat[An $\alpha$graph $\mathcal{A}$.\label{fig:example_arena_II}]{
\begin{tikzpicture}[arrows=->, scale=.5]
	\node[circ] (A) {$C$};
	\node[squa, right=of A, yshift=0ex, xshift=0ex] (B) {$B$};
	\node[squa, below=of B] (C) {$A$};
	\node[squa, left=of C, xshift=0ex, yshift=0ex] (D) {$G$};
	\node[squa, above left=of A, xshift=0ex, yshift=0ex] (E) {$H$};
	\node[circ, above right=of B, yshift=0ex, xshift=0ex] (F) {$D$};
	\node[squa, below right=of C, xshift=0ex, yshift=0ex] (G) {$E$};
	\node[circ, below left=of D, yshift=0ex, xshift=0ex] (H) {$F$};

	\draw[] (B) to [xshift=0ex, yshift=0ex] node[below, xshift=0ex, yshift=0ex] {} (A);
	\draw[] (B) to [xshift=0ex, yshift=0ex] node[below, xshift=0ex, yshift=0ex] {} (C);
	\draw[] (D) to [xshift=0ex, yshift=0ex] node[below, xshift=0ex, yshift=0ex] {} (C);
	\draw[] (A) to [xshift=0ex, yshift=0ex] node[below, xshift=0ex, yshift=0ex] {} (D);
	\draw[] (F) to [xshift=0ex, yshift=0ex] node[below, xshift=0ex, yshift=0ex] {} (E);
	\draw[] (G) to [xshift=0ex, yshift=0ex] node[below, xshift=0ex, yshift=0ex] {} (F);
	\draw[] (H) to [xshift=0ex, yshift=0ex] node[below, xshift=0ex, yshift=0ex] {} (G);
	\draw[] (H) to [xshift=0ex, yshift=0ex] node[below, xshift=0ex, yshift=0ex] {} (E);

	\draw[] (A) to [xshift=0ex, yshift=0ex] node[below, xshift=0ex, yshift=0ex] {} (E);
	\draw[] (F) to [xshift=0ex, yshift=0ex] node[below, xshift=0ex, yshift=0ex] {} (B);
	\draw[] (G) to [xshift=0ex, yshift=0ex] node[below, xshift=0ex, yshift=0ex] {} (C);
	\draw[] (D) to [xshift=0ex, yshift=0ex] node[below, xshift=0ex, yshift=0ex] {} (H);
\end{tikzpicture}
}
\subfloat[The $\alpha$palm-tree generated by an $\alpha$DFS rooted at $A$, with timestamps of vertices and labelled arcs.]{
\begin{tikzpicture}[arrows=->, scale=.5]
	\node[squa, label={above : \tiny $1|8$}] (tA) {$A$};
	\node[squa, below=of tA, label={right : \tiny$2|3$}, yshift=0ex, xshift=0ex] (tB) {$B$};
	\node[color=black!15,circ, below left=of tB, label={above : \tiny\bf n.a.}, yshift=0ex, xshift=0ex] (tC) {$C$};
	\node[color=black!15,circ, below=of tB, label={right : \tiny\bf n.a.}, yshift=0ex, xshift=0ex] (tD) {$D$};
	\node[squa, below right=of tA, label={right : \tiny$4|5$}, yshift=0ex, xshift=5ex] (tE) {$E$};
	\node[color=black!15,circ, below=of tE, label={above right : \tiny\bf n.a.}, yshift=0ex, xshift=0ex] (tF) {$F$};
	\node[squa, below left=of tC, label={left : \tiny$6|7$}, yshift=0ex, xshift=0ex] (tG) {$G$};

	\draw[thick] (tB) to [xshift=0ex, yshift=0ex] node[xshift=0ex, yshift=0ex] {\tiny\bf tree} (tA);
	\draw[dotted] (tD) to [xshift=0ex, yshift=0ex] node[below,xshift=0ex, yshift=0ex] {} (tB);
	\draw[thick] (tE) to [xshift=0ex, yshift=0ex] node[right,xshift=0ex, yshift=0ex] {\tiny\bf tree} (tA);
	\draw[dotted] (tF) to [bend left=0, xshift=0ex, yshift=0ex] node[xshift=0ex, yshift=0ex] {} (tE);
	\draw[thick] (tG) to [bend left=20, xshift=0ex, yshift=0ex] node[left,xshift=0ex, yshift=0ex] {\tiny\bf tree} (tA);
	\draw[dotted] (tC) to [bend left=10, xshift=0ex, yshift=0ex] node[xshift=0ex, yshift=0ex] {} (tG);
\end{tikzpicture}
}
\qquad
\subfloat[The order of arcs' exploration.]{
\begin{tikzpicture}[arrows=->, scale=.1, node distance = 0 and 0]
	\node[] (eA) {\text{\tiny{1.}}$(B,A)$};
	\node[below=of eA, yshift=0ex, xshift=0ex] (eC) {\text{\tiny{2.}}$(D,B)$};
	\node[below=of eC, yshift=0ex, xshift=0ex] (eD) {\text{\tiny{3.}}$(E,A)$};
	\node[below=of eD, yshift=0ex, xshift=0ex] (eE2) {\text{\tiny{4.}}$(F,E)$};
	\node[below=of eE2, yshift=0ex, xshift=0ex] (eG) {\text{\tiny{5.}}$(G,A)$};
	\node[below=of eG, yshift=0ex, xshift=0ex] (eGA) {\text{\tiny{6.}}$(C,G)$};
\end{tikzpicture}
}
\end{minipage}
\begin{minipage}{\columnwidth}
\centering
\subfloat[The $\alpha$palm-tree generated by an $\alpha$DFS rooted at $H$.]{
\begin{tikzpicture}[arrows=->, scale=.45]
	\node[color=black!15, squa, label={above: \tiny $1|8$}] (tA) {$A$};
	\node[color=black!15, squa, below=of tA, label={right: \tiny$2|3$}, yshift=0ex, xshift=0ex] (tB) {$B$};
	\node[color=black!15, circ, below left=of tB, label={above: \tiny\bf n.a.}, yshift=0ex, xshift=0ex] (tC) {$C$};
	\node[squa, below=of tC, label={ below, xshift=.75ex, yshift=.75ex : \tiny$9|10$}, yshift=0ex, xshift=0ex] (tH) {$H$};
	\node[color=black!15, circ, below=of tB, label={right : \tiny\bf n.a.}, yshift=0ex, xshift=0ex] (tD) {$D$};
	\node[color=black!15, squa, below right=of tA, label={right: \tiny$4|5$}, yshift=0ex, xshift=5ex] (tE) {$E$};
	\node[color=black!15, circ, below=of tE, label={right : \tiny\bf n.a.}, yshift=0ex, xshift=0ex] (tF) {$F$};
	\node[color=black!15, squa, below left=of tC, label={left: \tiny$6|7$}, yshift=0ex, xshift=0ex] (tG) {$G$};

	\draw[color=black!15,thick] (tB) to [xshift=0ex, yshift=0ex] node[xshift=0ex, yshift=0ex] {\tiny\bf tree} (tA);
	\draw[color=black!15,dotted] (tD) to [xshift=0ex, yshift=0ex] node[below,xshift=0ex, yshift=0ex] {} (tB);
	\draw[color=black!15,thick] (tE) to [xshift=0ex, yshift=0ex] node[right,xshift=0ex, yshift=0ex] {\tiny\bf tree} (tA);
	\draw[color=black!15,dotted] (tF) to [bend left=0, xshift=0ex, yshift=0ex] node[xshift=0ex, yshift=0ex] {} (tE);
	\draw[color=black!15,thick] (tG) to [bend left=20, xshift=0ex, yshift=0ex] node[left,xshift=0ex, yshift=0ex] {\tiny\bf tree} (tA);
	\draw[color=black!15,dotted] (tC) to [bend left=10, xshift=0ex, yshift=0ex] node[xshift=0ex, yshift=0ex] {} (tG);
	\draw[dotted] (tC) to [xshift=0ex, yshift=0ex] node[right, xshift=0ex, yshift=0ex] {} (tH);
	\draw[dotted] (tD) to [xshift=0ex, yshift=0ex] node[right, xshift=0ex, yshift=0ex] {} (tH);
	\draw[dotted] (tF) to [bend left=30, xshift=0ex, yshift=0ex] node[xshift=0ex, yshift=0ex] {} (tH);
\end{tikzpicture}
}
\qquad
\subfloat[The order of arcs' exploration.]{
\begin{tikzpicture}[arrows=->, scale=.1, node distance = 0 and 0]
	\node[] (eA) {\text{\tiny{7.}}$(C,H)$};
	\node[below=of eA] (eB) {\text{\tiny{8.}}$(D,H)$};
	\node[below=of eB] (eC) {\text{\tiny{9.}}$(F,H)$};
\end{tikzpicture}
}
\qquad
\subfloat[The $\alpha$palm-trees generated by an $\alpha$DFS rooted at $C, D, F$.]{
\begin{tikzpicture}[arrows=->, scale=.4]
	\node[color=black!15, squa, label={above : \tiny $1|8$}] (tA) {$A$};
	\node[color=black!15, squa, below=of tA, label={right : \tiny $2|3$}, yshift=0ex, xshift=0ex] (tB) {$B$};
	\node[circ, below left=of tB, label={above,xshift=.75ex,yshift=-.5ex :\tiny $11|12$}, yshift=0ex, xshift=0ex] (tC) {$C$};
	\node[color=black!15, squa, below=of tC, label={below, xshift=.75ex, yshift=.75ex : \tiny $9|10$}, yshift=0ex, xshift=0ex] (tH) {$H$};
	\node[circ, below=of tB, label={below,yshift=.5ex : \tiny $13|14$}, yshift=0ex, xshift=0ex] (tD) {$D$};
	\node[color=black!15, squa, below right=of tA, label={right :\tiny $4|5$}, yshift=0ex, xshift=5ex] (tE) {$E$};
	\node[circ, below=of tE, label={right :\tiny $15|16$}, yshift=0ex, xshift=0ex] (tF) {$F$};
	\node[color=black!15, squa, below left=of tC, label={left :\tiny $6|7$}, yshift=0ex, xshift=0ex] (tG) {$G$};

	\draw[color=black!15,thick] (tB) to [xshift=0ex, yshift=0ex] node[xshift=0ex, yshift=0ex] {\tiny\bf tree} (tA);
	\draw[dotted] (tB) to [xshift=0ex, yshift=0ex] node[xshift=0ex, yshift=0ex] {} (tC);
	\draw[color=black!15,dotted] (tD) to [xshift=0ex, yshift=0ex] node[below, xshift=0ex, yshift=1.5ex] {} (tB);
	\draw[color=black!15,thick] (tE) to [xshift=0ex, yshift=0ex] node[right,xshift=0ex, yshift=0ex] {\tiny\bf tree} (tA);
	\draw[color=black!15,dotted] (tF) to [bend left=0, xshift=0ex, yshift=0ex] node[xshift=0ex, yshift=0ex] {} (tE);
	\draw[dotted] (tE) to [bend left=0, xshift=0ex, yshift=0ex] node[xshift=0ex, yshift=0ex] {} (tD);
	\draw[color=black!15,thick] (tG) to [bend left=20, xshift=0ex, yshift=0ex] node[left,xshift=0ex, yshift=0ex] {\tiny\bf tree} (tA);
	\draw[color=black!15,dotted] (tC) to [bend left=10, xshift=0ex, yshift=0ex] node[xshift=0ex, yshift=0ex] {} (tG);
	\draw[dotted] (tG) to [bend right=70, xshift=0ex, yshift=0ex] node[xshift=0ex, yshift=0ex] {} (tF);
	\draw[color=black!15,dotted] (tC) to [xshift=0ex, yshift=0ex] node[right, xshift=0ex, yshift=0ex] {} (tH);
	\draw[color=black!15,dotted] (tD) to [xshift=0ex, yshift=0ex] node[xshift=1ex, yshift=0ex] {} (tH);
	\draw[color=black!15,dotted] (tF) to [bend left=30, xshift=0ex, yshift=0ex] node[xshift=0ex, yshift=0ex] {} (tH);
\end{tikzpicture}
}
\qquad
\subfloat[The order of arcs' exploration.]{
\begin{tikzpicture}[arrows=->, scale=.1, node distance = 0 and 0]
	\node[] (eA) {\text{\tiny{10.}}$(B,C)$};
	\node[below=of eA, yshift=0ex, xshift=0ex] (eB) {\text{\tiny{11.}}$(E,D)$};
	\node[below=of eB, yshift=0ex, xshift=0ex] (eC) {\text{\tiny{12.}}$(G,F)$};
\end{tikzpicture}
}
\end{minipage}
\begin{minipage}{\columnwidth}
\centering
\subfloat[The resulting $\alpha$jungle, which is generated by multiple
$\alpha$DFSs rooted at $A,H,C,D$ and $F$.]{
\begin{tikzpicture}[arrows=->, scale=.5]
	\node[squa, label={above : \tiny $1|8$}] (tA) {$A$};
	\node[squa, below=of tA, label={right : \tiny $2|3$}, yshift=0ex, xshift=0ex] (tB) {$B$};
	\node[circ, below left=of tB, label={above, xshift=.75ex, yshift=-.5ex : \tiny $11|12$}, yshift=0ex, xshift=0ex] (tC) {$C$};
	\node[squa, below=of tC, label={below, xshift=.75ex, yshift=.75ex: \tiny $9|10$}, yshift=0ex, xshift=0ex] (tH) {$H$};
	\node[circ, below=of tB, label={below,yshift=.5ex : \tiny $13|14$}, yshift=0ex, xshift=0ex] (tD) {$D$};
	\node[squa, below right=of tA, label={right : \tiny $4|5$}, yshift=0ex, xshift=5ex] (tE) {$E$};
	\node[circ, below=of tE, label={right : \tiny $15|16$}, yshift=0ex, xshift=0ex] (tF) {$F$};
	\node[squa, below left=of tC, label={left : \tiny $6|7$}, yshift=0ex, xshift=0ex] (tG) {$G$};

	\draw[thick] (tB) to [xshift=0ex, yshift=0ex] node[xshift=0ex, yshift=0ex] {\tiny\bf tree} (tA);
	\draw[dotted] (tB) to [xshift=0ex, yshift=0ex] node[xshift=0ex, yshift=0ex] {} (tC);
	\draw[dotted] (tD) to [xshift=0ex, yshift=0ex] node[below, xshift=0ex, yshift=1.5ex] {} (tB);
	\draw[thick] (tE) to [xshift=0ex, yshift=0ex] node[right,xshift=0ex, yshift=0ex] {\tiny\bf tree} (tA);
	\draw[dotted] (tF) to [bend left=0, xshift=0ex, yshift=0ex] node[xshift=0ex, yshift=0ex] {} (tE);
	\draw[dotted] (tE) to [bend left=0, xshift=0ex, yshift=0ex] node[xshift=0ex, yshift=0ex] {} (tD);
	\draw[thick] (tG) to [bend left=20, xshift=0ex, yshift=0ex] node[left,xshift=0ex, yshift=0ex] {\tiny\bf tree} (tA);
	\draw[dotted] (tC) to [bend left=10, xshift=0ex, yshift=0ex] node[xshift=0ex, yshift=0ex] {} (tG);
	\draw[dotted] (tG) to [bend right=70, xshift=0ex, yshift=0ex] node[xshift=0ex, yshift=0ex] {} (tF);
	\draw[dotted] (tC) to [xshift=0ex, yshift=0ex] node[right, xshift=0ex, yshift=0ex] {} (tH);
	\draw[dotted] (tD) to [xshift=0ex, yshift=0ex] node[xshift=1ex, yshift=0ex] {} (tH);
	\draw[dotted] (tF) to [bend left=30,xshift=0ex, yshift=0ex] node[xshift=0ex, yshift=0ex] {} (tH);
\end{tikzpicture}
}
\end{minipage}
\caption{An $\alpha$graph (a), and the construction of a corresponding $\alpha$jungle (b-h).}\label{fig:tr-jungle}
\end{figure}

    \begin{Prop}\label{prop:aDFS_complexity}
    Assume that $\alpha$DFS() (Algo.~\ref{algo:aDFS}) runs on a given input $\alpha$graph $\A$ on vertex set $V$ and arc set $A$.
    Each vertex $v\in V$ is timestamped by $\textit{open}[v]$ exactly once,
    and the algorithm halts in time $\Theta\big(|V| + |A| + \text{Time}[\text{\footnotesize LCA}]\big)$,
    	consuming space $\Theta\big(|V| + |A| + \text{Space}[\text{\footnotesize LCA}]\big)$,
    where $\text{Time}[\text{\footnotesize LCA}]$ ($\text{Space}[\text{\footnotesize LCA}]$)
    	is the aggregate total time (space) taken by all LCA computations that are done at lines~\ref{algo:aDFS-visit:l8}-\ref{algo:aDFS-visit:l9} of $\textit{$\alpha$DFS-visit}()$~(Proc.~\ref{algo:aDFS-visit}).
    \end{Prop}
    \begin{proof}
    	The initialization phase takes $\Theta(|V|+|A|)$ time (see lines~\ref{algo:aDFS:l1}-\ref{algo:aDFS:l7} of Algo.~\ref{algo:aDFS}).
    Recall Algo.~\ref{algo:aDFS} performs multiple calls to $\textit{$\alpha$DFS-visit}(v,\A)$ (Proc.~\ref{algo:aDFS-visit}), each for some $v\in V$.
    Any of these happens if and only if $\textit{open}[v]=+\infty$, and then $\textit{open}[v]$ is set to some non-zero value.
    Thus, the total number of invocations of \textit{$\alpha$DFS-visit()} (Proc.~\ref{algo:aDFS-visit}) is at most $|V|$.
(Indeed, calls are issued only for vertices that actually get visited; vertices in $V_{\ocircle}$ that remain unvisited are timestamped in the final loop without calling \textit{$\alpha$DFS-visit()}.) 
In any case, each vertex $v\in V$ is assigned $\textit{open}[v]$ exactly once (either during a visit or in the final loop).

    Concerning time complexity, consider each of such visits independently from one another, where the in-neighbourhood
    	$N^{\text{in}}_\A(v)$ is explored. For some $u\in N^{\text{in}}_\A(v)\cap V_\ocircle$, the LCA of $N^{\text{out}}_\A(u)$ might be computed, but notice that all the other operations about $N^{\text{in}}_\A(v)$ can be done in constant time per single $u\in N^{\text{in}}_\A(v)$.
    At the end of each visit the stack $\textit{rSt}[v]$ is emptied, still, due to the condition $\textit{cnt}[u]=0$ any $u\in V_\ocircle$ can be pushed on $\textit{rSt}[v]$ at most once and for at most one $v\in V$.
    Therefore, the $\Theta\big(|V| + |A| + \text{Time}[\text{\footnotesize LCA}]\big)$ aggregate time bound holds.

    Concerning space usage, a similar argument shows that the aggregate total space of storing $\{\textit{rSt}[v]\}_{v\in V}$ is $O(|V|)$.
    Also, the total size of $\textit{open}[]$, $\textit{close}[]$ and $\textit{cnt}[]$ is $\Theta(|V|)$, and that of $A'$ is $\Theta(|A|)$.
    \end{proof}

    Later on in [Section~\ref{subsect:lca-dsf}, Theorem~\ref{thm:aDFS:complexityRAM}],
    	 the aggregate total time and space of all LCA computations	(\ie $\text{Time}[\text{\footnotesize LCA}]$ and $\text{Space}[\text{\footnotesize LCA}]$) will be bounded linearly.
    	Before that, in the following Section~\ref{subsect:graphstructures}, let us read out and carefully analyze the graph structure of the $\alpha$jungle $\J_\A$.
    
	\begin{Exa}
Figure~\ref{fig:palm_tree_example} illustrates a complete execution of the $\alpha$DFS algorithm on an 
example $\alpha$graph. Subfigure~\ref{fig:example_arena_A-PT} shows the initial $\alpha$graph $\mathcal{A}$, 
comprising vertices labeled from $A$ to $H$, 
partitioned into vertices controlled by Player~$\square$ or Player~$\ocircle$, 
connected by directed arcs representing available moves in the alternating game.

Subfigure~\ref{fig:example_alternating-palm-tree} presents the resulting $\alpha$palm-tree formed by 
performing the $\alpha$DFS traversal rooted at vertex $A$. Each vertex is annotated with timestamps indicating 
when it was first discovered and when the search finished exploring all its descendants, 
denoted by $\langle\textit{open}[v]\rangle \mid \langle\textit{close}[v]\rangle$. 
Arcs within this palm-tree are categorized into distinct types: \emph{tree arcs} (thick solid lines) indicating the 
DFS backbone structure, \emph{stalk arcs} (dotted lines) that represent connections from vertices 
to ancestors in the palm-tree, and possibly other arcs (dashed lines) highlighting 
further structural relations relevant to $\alpha$reachability. 
Subfigure~\ref{fig:example_alternating-palm-tree-order} explicitly enumerates the exact order in which the arcs were explored during the 
execution of the $\alpha$DFS algorithm. 

Together, these subfigures provide a comprehensive visualization of how $\alpha$DFS systematically 
constructs an $\alpha$palm-tree, highlighting the interplay between alternating reachability constraints 
and depth-first traversal principles.
	\end{Exa}
	
	\subsection{Graph Structures}\label{subsect:graphstructures}
    Let's start by formalizing the structural properties of the $\alpha$palm-trees.	Examples are given in \figref{fig:palm_tree_example}~and~\ref{fig:tr-jungle}.

    \begin{Def}\label{def:tr-pt}
    An \emph{alternating palm-tree ($\alpha$palm-tree)} is a triplet $(\mathcal{P},\textit{open}[], \textit{close}[])$, where:

    (i) $\mathcal{P}= (V, A, \langle V_{\square}, V_{\ocircle}\rangle)$ is an $\alpha$graph on
    $V = V_{\square}\cup V_{\ocircle}$ and
    	$A = A_{\pi}\cup A_{f} \cup A_{s} \cup A_{c}$, so
    	the vertex set is split in squares and circles whereas the arc set into four categories.

    (ii) $\textit{open}[], \textit{close}[]:V\rightarrow \N$ timestamp the vertex set $V$ in pre and post order respectively;

    (iii) the following four main properties hold:
    \begin{enumerate}
    	\item[($\alpha$pt-1)] $\mathcal{T}_{\mathcal{P}}\doteq (V, A_{\pi})$ is an inward directed rooted tree such that:

    		(a) the root $r_{ \mathcal{T}_{\mathcal{P}}}$ of $\mathcal{T}_{\mathcal{P}}$ is controlled by Player~$\square$,
    					\ie $r_{ \mathcal{T}_{\mathcal{P}} } \in V_{\square}$;

    		(b) $\textit{open}[v]<\textit{open}[u]<\textit{close}[u]<\textit{close}[v]$ whenever $(u,v)\in A_{\pi}$,
    			\ie if $v=\pi(u)$ is the parent of $u$ in $\mathcal{T}_{\mathcal{P}}$;

    		\item[($\alpha$pt-2)] Each frond-arc $(u,v)\in A_{f}$ connects
    			some $u\in V_{\square}$ to one of its proper descendants $v\in V$ in $\mathcal{T}_{\mathcal{P}} $;

    		\item[($\alpha$pt-3)] Each stalk-arc $(u,v)\in A_{s}$ connects some $u\in V_{\ocircle}$
    		to one of the descendants $v$ of its parent $\pi(u)$ (\ie possibly to $\pi(u)$ itself);
    		particularly, given any $u\in V_{\ocircle}$, the following three properties hold:

    		(a) $\{v\in V \mid (u, v)\in A_{s}\} \cup \{\pi(u)\} = N^{\text{out}}_{\mathcal{P}}(u)$;

    		(b) $\pi(u)$ is the LCA of $\{v\in V \mid (u, v)\in A_{s}\}$ in $\mathcal{T}_{\mathcal{P}}$;

    		(c) $\textit{open}[v]<\textit{close}[v]<\textit{open}[u]<\textit{close}[u]$ for every $v\in N^{\text{out}}_{\mathcal{P}}(u)\setminus \{\pi(u)\}$.

    	\item[($\alpha$pt-4)] Each cross-arc $(u,v)\in A_{c}$ connects some $u\in V_{\square}$ to some $v\in V$ such that:

    		(a) $v$ is not a descendant of $u$ in $\mathcal{T}_{\mathcal{P}}$;

    		(b) either $v$ is a proper ancestor of $u$ in $\mathcal{T}_{\mathcal{P}}$ (in that case $\textit{open}[v]<\textit{open}[u]<\textit{close}[u]<\textit{close}[v]$),
    			or $\textit{open}[u]<\textit{close}[u]<\textit{open}[v]<\textit{close}[v]$.
    \end{enumerate}
    \end{Def}

    An $\alpha$jungle) is formed by a disjoint union of $\alpha$palm-trees (see $\alpha$jn-1 and $\alpha$jn-2),
    	possibly with external cross-arcs connecting two distinct $\alpha$palm-trees (see $\alpha$jn-3),
    	plus a (possibly empty) set of circled vertices each one having out-neighbours lying in at least two distinct $\alpha$palm-trees (see $\alpha$jn-4).
    \begin{Def}\label{def:tr-jungle}
    An \emph{alternating jungle ($\alpha$jungle)} is an $\alpha$graph $\mathcal{J} = (V, A, \langle V_{\square}, V_{\ocircle}\rangle)$
    comprising a family of vertex-disjoint $\alpha$palm-trees $\{(\mathcal{P}_i, \textit{open}[]_i, \textit{close}[]_i)\}_{i=1}^k$,
    	whose vertices are timestamped, and these hold:

    ($\alpha$jn-1) $\forall{i\in [k]}$ $\mathcal{P}_i = (V_i, A_i, \langle{V_{\square}}_i, {V_{\ocircle}}_i\rangle)$,
    			where ${V_{\square}}_i\subseteq V_{\square}, {V_{\ocircle}}_i\subseteq V_{\ocircle}, A_i\subseteq A$;

    ($\alpha$jn-2) $\forall{i,j\in [k]}$ $V_i\cap V_j=\emptyset$ if $i\neq j$;

    ($\alpha$jn-3) If $(u,v)\in A$ for some $u\in V_i$ and $v\in V_j$ such that $i\neq j$, then $u\in {V_{\square}}_i$ and $i < j$;

    ($\alpha$jn-4) If $v\in V\setminus \bigcup_{i=1}^k V_i$, then $v\in V_{\ocircle}$ and $N^{\text{out}}_{\J}(v)\subseteq V_i$ for \emph{no} $i\in [k]$.
    \end{Def}

    Proposition~\ref{prop:aDFS-jungle} shows that $\alpha$DFS() (Algo.~\ref{algo:aDFS}) really constructs an $\alpha$jungle.
    	It's worth introducing a technical but conceptually simple notion, that of \emph{support} for an $\alpha$jungle.
    The support of $\J$ is just the same $\alpha$graph deprived of all the arcs in $\{(u,v)\in A_{\pi}\setminus A_s\mid u\in V_\ocircle\}$,
    	\ie those arcs that are added by $\textit{$\alpha$DFS()}$ (Algo.~\ref{algo:aDFS}) but that were not in the input $\alpha$graph.
    More formal details below.
    \begin{Def}
    	Given an $\alpha$palm-tree $(\mathcal{P}, \textit{open}[], \textit{close}[])$, 
		for $\mathcal{P}=(V, A, \langle V_{\square}, V_{\ocircle}\rangle)$,
    	$A=A_{\pi}\cup A_{f} \cup A_{s} \cup A_{c}$,
    	the \emph{support} of $\mathcal{P}$ is the $\alpha$graph $\mathcal{P}_*\doteq (V, A_*, \langle V_{\square}, V_{\ocircle}\rangle )$, where
    	$A_* \doteq  \big\{(u,v)\in A\mid u\in V_{\square}\big\} \cup A_{s}$.

    Notice that $A_* = A\setminus \big\{(u,v)\in A_{\pi}\setminus A_s\mid u\in V_\ocircle\big\}$ holds by ($\alpha$pt-3).

    Given an $\alpha$jungle $\mathcal{J}$ with family of $\alpha$palm-trees $\{\mathcal{P}_i\}_{i=1}^k$,
    let $\overline{V}\doteq V\setminus \bigcup_{i=1}^k V_i$ (where $V_i$ is the vertex set of $\mathcal{P}_i$).
    The \emph{support} of $\mathcal{J}$ is the $\alpha$graph $\mathcal{J}_*$ obtained from $\mathcal{J}$
    by replacing each $\mathcal{P}_i$ with its support $({\mathcal{P}_i})_*$, and by leaving intact
    all the vertices in $\overline{V}$ and all arcs $(u,v)$ of $\mathcal{J}$ such that:
    either, (i) $u\in V_i$ and $v\in V_j$ for some $i\neq j$ (\ie all external cross-arcs);
    or, (ii) $u\in \overline{V}$ or $v\in \overline{V}$ (possibly both).
    \end{Def}

    Let us now argue more formally that an $\alpha$jungle really traces down the behaviour of $\alpha$DFS() (Algo.~\ref{algo:aDFS}).
    \begin{Prop}\label{prop:aDFS-jungle}
    Let $\mathcal{A}=(V, A, \langle  V_{\square}, V_{\ocircle} \rangle )$ be an $\alpha$graph. The following two propositions hold.
    \begin{enumerate}
    	\item Let $J$ be the $\alpha$graph constructed by executing $\textit{$\alpha$DFS}(\mathcal{A})$ (Algo.~\ref{algo:aDFS}). Then, $J$ is an $\alpha$jungle.
    	\item Let $J$ be an $\alpha$jungle with support $J_*$.
    		Then, $\textit{$\alpha$DFS}(J_*)$ (Algo.~\ref{algo:aDFS}) reconstructs $J$ itself, \ie $\J_{J_*}=J$.
    \end{enumerate}
    \end{Prop}
    \begin{proof}[Proof of (1)]
    Recall, $\textit{$\alpha$DFS}(\mathcal{A})$ (Algo.~\ref{algo:aDFS}) performs a sequence of invocations to $\textit{$\alpha$DFS-visit}(\cdot, \A)$ (Proc.~\ref{algo:aDFS-visit}).
    Let $k$ be the total number of times that $\textit{$\alpha$DFS-visit()}$ is invoked \emph{only} at line~\ref{algo:aDFS:l11} of $\textit{$\alpha$DFS()}$ (Algo.~\ref{algo:aDFS}).
    For each $i=1, 2, \ldots, k$, let $u_i\in V_\square$ be the vertex that is passed as a parameter to the $i$-th invocation,
    	\ie assume $\textit{$\alpha$DFS-visit}(u_i,\A)$ is the $i$-th call;
    notice $u_i\in V_\square$ by line~\ref{algo:aDFS:l9} of $\textit{$\alpha$DFS()}$ (Algo.~\ref{algo:aDFS}).
    Let $V_i\subseteq V$ be the set of all vertices timestamped by $\textit{open}[]$ during the $i$-th invocation (recursive calls included).
    Similarly, let $A_i$ be the set of arcs that are explored during that invocation (recursive calls included),
    	and consider the \emph{internal arcs} \ie ${A_i}_{\text{int}}\doteq \{(a,b) \in A_i\mid \text{both } a,b\in V_i\}$.
    Finally let $\mathcal{P}_i\doteq (V_i, {A_i}_{\text{int}}, \langle V_\square \cap V_i, V_\ocircle \cap V_i \rangle )$.
    It is easy to check that $\mathcal{P}_i$ is an $\alpha$palm-tree since it	satisfies all properties from \emph{($\alpha$pt-1)} to \emph{($\alpha$pt-4)}.
    We also claim that $J$ is an $\alpha$jungle with $\alpha$palm-tree family $\{\mathcal{P}_i\}_{i\in [k]}$.
    Clearly, we are given a family $\{\mathcal{P}_i\}_{i\in [k]}$ of vertex-disjoint $\alpha$palm-trees,
    	so properties \emph{($\alpha$jn-1)} and \emph{($\alpha$jn-2)} hold.
    Concerning \emph{($\alpha$jn-3)}, let $(u,v)\in A$ by any arc such that $u\in V_i$ and $v\in V_j$ with $i\neq j$;
    then $u\in {V_\square}_i$ (we can argue this by exclusion:
    	since $\mathcal{P}_i$ is an $\alpha$palm-tree, \emph{($\alpha$pt-3)} holds for $V_\ocircle$,
    		so the tail $u$ of an external cross-link connecting two distinct $\alpha$palm-trees must be a square);
    also, $i<j$ since otherwise $u$ would've joined $\mathcal{P}_j$ instead of $\mathcal{P}_i$
    	(\cfr lines~\ref{algo:aDFS-visit:l3}-\ref{algo:aDFS-visit:l5} of $\textit{$\alpha$DFS-visit}()$).
    Concerning \emph{($\alpha$jn-4)}, let $v\in V\setminus \bigcup_{i=1}^k V_i$,
    	then $v\in V_{\ocircle}$ (\cfr lines~\ref{algo:aDFS:l9}-\ref{algo:aDFS:l15} of $\textit{$\alpha$DFS}()$);
    also, $N^{\text{out}}_{\J}(v)\subseteq V_i$ holds for \emph{no} $i\in [k]$, otherwise $v$ would've joined $\mathcal{P}_i$ thanks to lines~9-11 and~15-19 of $\textit{$\alpha$DFS-visit}()$.
    All in, $J$ is an $\alpha$jungle.
    \end{proof}
    \begin{proof}[Proof of (2)]
    Recall that the support $J_*$ can be obtained from $J$ simply by removing from the $\alpha$palm-trees of $J$
    	all the arcs $(u,v)\in A_{\pi}\setminus A_s$ such that $u\in V_\ocircle$.
     Consider the total ordering $<_{\textit{open}}$ on the vertex set $V$ induced by the opening timestamp $\textit{open}[]$ of $J$,
     \ie $\forall{a,b\in V}\, a<_{\textit{open}} b \iff \textit{open}[a]<\textit{open}[b]$.
    	Encode an adjacency list of $J_*$ such that: (i) the main list of vertices is ordered according to $<_{\textit{open}}$;
    		(ii) for each $u\in V$, also the in-neighbourhood $N^{\text{in}}_{J_*}(u)$ is ordered according to $<_{\textit{open}}$.
    	Since $J$ satisfies all properties from \emph{($\alpha$jn-1)} to \emph{($\alpha$jn-4)} and
    		their $\alpha$palm-trees satisfy all properties from \emph{($\alpha$pt-1)} to \emph{($\alpha$pt-4)},
    		it's straightforward to check inductively that $\textit{$\alpha$DFS}(J_*)=J$.
    \end{proof}

    Still it remains to be seen how to perform efficiently, \ie in linear-time, all the LCAs computations that are needed
    	at lines~9-10 of $\textit{$\alpha$DFS-visit()}$ (Proc.~\ref{algo:aDFS-visit}).
    In the next subsection, we suggest to adopt a disjoint-set forest data structure with a non-ranked union and a classical \textit{Find} primitive based on path-compression.

    \subsection{Computing LCAs by Disjoint-Set Forest}\label{subsect:lca-dsf}
    A \emph{disjoint-set forest (dsf)} data structure~\cite{Tar75}, hereby denoted $\mathcal{D}$,
    is a data structure that keeps track of a set of elements partitioned into a number of disjoint (non-overlapping) subsets,
    	each of which is represented by a rooted tree. This is also known as union-find data structure or merge-find set.

    The following three operations are supported:
    $\mathcal{D}.\textit{MakeSet}(\cdot)$, $\mathcal{D}.\textit{Union}(\cdot, \cdot)$ and $\mathcal{D}.\textit{Find}(\cdot)$, where:

    \emph{(dsf-1)} The representative element of each disjoint set is the root of that set's tree;

    \emph{(dsf-2)} $\textit{MakeSet}(v)$ initializes the parent of a vertex $v\in V$ to be $v$ itself, \ie a singleton vertex tree;

    \emph{(dsf-3)} $\textit{Union}(u,v)$ combines two trees, $T_1$ rooted at $u$ and $T_2$ rooted at $v$,
    		into a new tree $T_3$ which is still rooted at $v$, \ie $u$ simply becomes a child of $v$ (this is a \emph{non-ranked} union).

    \emph{(dsf-4)} $\textit{Find}(v)$, starting from $v$, traverses the ancestors of $v$ until the root $r$ of the tree containing $v$ is finally reached.
    		While doing this, $\textit{Find}(v)$ changes each ancestor's parent reference to directly point to $r$ (this is \emph{path-compression});
    			the resulting tree is much flatter, speeding up future operations, not only on these traversed elements but also on those referencing them from the downstairs of the tree.

    Let us describe how to implement the LCAs computations at lines~9-10 of 
	$\textit{$\alpha$DFS-visit}()$ (Proc.~\ref{algo:aDFS-visit}).
    The resulting algorithm is named \textit{dsf-$\alpha$DFS}, based on a global dsf data structure $\mathcal{D}$.

    The main procedure of dsf-$\alpha$DFS() (Algo.~\ref{algo:dsf-aDFS}) is 
	almost the same as $\alpha$DFS() (Algo.~\ref{algo:aDFS}), 
		in the pseudocode the lines that differ are highlighted in grey to emphasize the modifications, 
	the only additions being:

    \emph{(dsf-init-1)} $\mathcal{D}.\textit{MakeSet}(v)$ is executed for each $v\in V$;

    \emph{(dsf-init-2)} An array indexed by circled vertices 
	$\textit{low\_ready}[]:V_\ocircle\rightarrow \N\cup\{+\infty\}$ is initialized as $\textit{low\_ready}[v]\leftarrow~+~\infty$ for every $v\in V_{\ocircle}$.
    Its role is tracking the $\textit{open}[]$ timestamp of the unique out-neighbour of $v\in V_{\ocircle}$
    	which is visited firstly and before all other out-neighbours (\ie the out-neighbour having minimum index).
    So, given $\A$ in input, the following invariant property will be maintained:
    		\[
    			\forall{v\in V_{\ocircle}}\; \textit{low\_ready}[v] =
    				 	\min\big\{\textit{open}[u]\in\N\cup\{+\infty\} \mid u\in N^{\text{out}}_{\A}(v) \big\}.
    					 \tag{$\text{I}_{\text{low}}$}
    		\]
    Lemma~\ref{lemma:gamma_is_LCA} shows that $\textit{low\_ready}[v]$ can be used as a compass needle for making LCA lookups;
    	indeed, because of the two forthcoming rules, the LCA that we need to find turns out to be the root of the disjoint set tree
    			containing precisely the vertex indexed by $\textit{low\_ready}[v]$.

    Let us now describe in more detail the distinctive rules of the dsf-$\alpha$DFS() algorithm. 
	The visiting subprocedure is given in Proc.~\ref{algo:dsf-aDFS-visit}, it goes as follows.	
	Let $v\in V$, then:

    \emph{(dsf-visit-1)} Whenever the visiting subprocedure, 
	$\textit{dsf-$\alpha$DFS-visit}(v, \A)$ (Proc.~\ref{algo:dsf-aDFS-visit}), makes a recursive call on some ingoing
    	neighbour $u\in N^{\text{in}}_\A(v)\cup \textit{rSt}[v]$ (see lines~6~and~19 of Proc.~\ref{algo:aDFS-visit}),
    	soon after that, it is executed $\mathcal{D}.\textit{Union}(u, v)$.
    Doing so, as soon as the recursive call on $u$ returns, the disjoint set tree of the child $u$ is merged with that of its parent~$v$;
    thus, parent-children ordering relations are preserved.
    This allows for fast lookup of the subtrees' roots (\ie the LCAs) that are needed in the \emph{(dsf-visit-2)} rule coming next.

    \emph{(dsf-visit-2)} Suppose that $\textit{dsf-$\alpha$DFS-visit}(v, \A)$ is currently visiting some $v\in V$,
    and that it comes to consider some in-neighbour $u\in N^{\text{in}}_{\A}(v)\cap V_\ocircle$ (at line~3~and~7).
    Then, assume at line~8, $\textit{low\_ready}$ is updated as follows: \[\textit{low\_ready}[u]\leftarrow \min(\textit{low\_ready}[u], \textit{open}[v]);\]
     this aims at satisfying the $\text{I}_{\text{low}}$ invariant.
    Next, $\textit{cnt}[u]$ is decremented (\cfr at line~8 of Proc.~\ref{algo:aDFS-visit}).

    If the condition $\textit{cnt}[u]=0$ is met at line~11 of $\textit{dsf-$\alpha$DFS-visit}(v, \A)$ (Proc.~\ref{algo:dsf-aDFS-visit}), the following is done:

    (a) It is identified the unique $x\in N^{\text{out}}_\A(u)$ s.t. $\textit{open}[x]=\textit{low\_ready}[u]$, and it is assigned to $\textit{low\_}v\leftarrow x$;

    (b) Then, we lookup for the root $\gamma$ of the corresponding disjoint set tree: $\gamma\leftarrow \mathcal{D}.\textit{Find}(\textit{low\_}v)$;

    (c) We say that any \( u \in V \) is \emph{active} if the predicate  
\( \textit{active}[u] \doteq (\textit{open}[u] < +\infty \ \text{and} \ \textit{close}[u] = +\infty) \) holds.
	So, if $\textit{active}[\gamma]=\textit{true}$, then $\gamma$ is pushed to the ready stack $\textit{rSt}[\gamma]$;
    	indeed, in that case, we can prove (see Lemma~\ref{lemma:gamma_is_LCA})
    		that the LCA of $N^{\text{out}}_{\A}(u)$ in $(V, A_{\pi})$ exists and it is really $\gamma$ (\ie the root of $\textit{low\_}v$).

    The rest of dsf-$\alpha$DFS-visit() (Proc.~\ref{algo:dsf-aDFS-visit}) is the same as Proc.~\ref{algo:aDFS-visit},  
	 the lines that differ are highlighted in grey to emphasize modifications. 
	This ends the description of dsf-$\alpha$DFS() (Algo.~\ref{algo:dsf-aDFS}).

    \begin{algorithm}[H]
    \caption{Safe-Alternating DFS with Disjoin-Set Forest}\label{algo:dsf-aDFS}
    \DontPrintSemicolon
    \nonl \SetKwProg{Fn}{Procedure}{}{}
    \normalsize
    \Fn{$\textit{dsf-$\alpha$DFS}(\mathcal{A})$}{
    		\SetKwInOut{Input}{input}
    		\SetKwInOut{Output}{output}

    \Input{An $\alpha$graph $\mathcal{A}=(V, A, \langle V_{\ocircle}, V_{\square} \rangle )$.}
    \Output{An $\alpha$jungle $\mathcal{J}_{\mathcal{A}}$.}
    $A_{\pi}, A_{f}, A_{s}, A_{c}\leftarrow \emptyset$; \label{algo:dsf-aDFS:l1}\;
    \ForEach{$u\in V$}{ \label{algo:dsf-aDFS:l2}
    	$\textit{open}[u]\leftarrow +\infty$; \label{algo:dsf-aDFS:l3}\;
    	$\textit{close}[u]\leftarrow +\infty$;\label{algo:dsf-aDFS:l4}\;
    	$\textit{rSt}[u]\leftarrow \emptyset$; \label{algo:dsf-aDFS:l5}\;
		
		\colorbox{gray!20}{$\mathcal{D}.\textit{make\_set}(u)$; \tcp{\emph{(dsf-init-1)}} } \label{algo:STCC:l5b}\;
      $\textit{rSt}[u]\leftarrow \emptyset$; \label{algo:STCC:l6b}\;
		
    	\If{$u\in V_{\ocircle}$}{ \label{algo:dsf-aDFS:l6}
		\colorbox{gray!20}{$\textit{low\_ready}[u]\leftarrow +\infty$; \tcp{\emph{(dsf-init-2)}} } \label{algo:STCC:l8b}\;
    	$\textit{cnt}[u]\leftarrow |N_{\A}^{\text{out}}(u)|$; \label{algo:dsf-aDFS:l7}
    	}
    }
    $\textit{time}\leftarrow 0$; \tcp{global time variable} \label{algo:dsf-aDFS:l8}
    \ForEach{$u\in V_{\square}$ \label{algo:dsf-aDFS:l9}}{
    	\If{$\textit{open}[u]=+\infty$ \label{algo:dsf-aDFS:l10}}{
    		$\textit{$\alpha$DFS-visit}(u, \A)$;\label{algo:dsf-aDFS:l11} \;
    	}
    }
    \ForEach{$u\in V_{\ocircle}$\label{algo:dsf-aDFS:l12}}{
    	\If{$\textit{open}[u]=+\infty$\label{algo:dsf-aDFS:l13}}{
    		$\textit{open}[u]\leftarrow \textit{time}$; \; \label{algo:dsf-aDFS:l14}
    		$\textit{close}[u]\leftarrow \textit{time}$; \; \label{algo:dsf-aDFS:l14}
    		$\textit{time}\leftarrow \textit{time}+1$; \label{algo:dsf-aDFS:l15}
    	}
    }
    $A'\leftarrow A_{\pi} \cup A_{f} \cup A_{s} \cup A_{c}$; \label{algo:dsf-aDFS:l16}\;
    \Return{$\mathcal{J}_{\mathcal{A}}\leftarrow (V, A', (V_{\square}, V_{\ocircle}))$ }; \label{algo:dsf-aDFS:l17}
    }
    \end{algorithm}

    \begin{procedurealgo}[h!]
    \caption{Visit Procedure of Safe-Alternating DFS with Disjoin-Set Forest}\label{algo:dsf-aDFS-visit}
    \DontPrintSemicolon
    \nonl \SetKwProg{Fn}{Procedure}{}{}
    \normalsize
    \Fn{$\textit{dsf-$\alpha$DFS-visit}(v, \A)$}{
    		\SetKwInOut{Input}{input}
    		\SetKwInOut{Output}{output}
     \Input{One vertex $v\in V$ of $\A$.}
    $\textit{open}[v]\leftarrow (\textit{time}\leftarrow \textit{time}+1)$; \label{algo:dsf-aDFS-visit:l1} \;
    \ForEach{$u\in N_\A^{\text{in}}(v)$}{ \label{algo:dsf-aDFS-visit:l2}
    	\If{$\textit{open}[u] = +\infty$}{ \label{algo:dsf-aDFS-visit:l3}
    		\If{$u\in V_\square$}{ \label{algo:dsf-aDFS-visit:l6}
    			add $(u,v)$ to $A_{\pi}$;\label{algo:dsf-aDFS-visit:l4}\;
    			$\textit{dsf-$\alpha$DFS-visit}(u, \A)$; \label{algo:dsf-aDFS-visit:l5}\;
				\colorbox{gray!20}{$\mathcal{D}.\textit{Union}(u, v)$; \tcp{\emph{(dsf-visit-1)}}} \label{algo:dsf-aDFS-visit:l5b} \;
    		}\Else{ \label{algo:dsf-aDFS-visit:l6}
				\colorbox{gray!20}{$\textit{low\_ready}[u]\leftarrow \min(\textit{low\_ready}[u], \textit{open}[v])$; \tcp{\emph{(dsf-visit-2)}}} \label{algo:dsf-aDFS-visit:l7b} \;
    			$\textit{cnt}[u]\leftarrow \textit{cnt}[u]-1$; \label{algo:dsf-aDFS-visit:l7}\;
				
				\If{\colorbox{gray!20}{$\textit{cnt}[u]=0$}}{ \label{algo:dsf-aDFS-visit:l8}
					\colorbox{gray!20}{$\textit{low\_}v\leftarrow $ the unique $x$ such that $\textit{open}[x]=\textit{low\_ready}[u]$;} \label{algo:dsf-aDFS-visit:l8b} \;
					\colorbox{gray!20}{$\gamma\leftarrow \mathcal{D}.\textit{find}(\textit{low\_}v)$;} \label{algo:dsf-aDFS-visit:l9} \;
    				\If{\colorbox{gray!20}{$\textit{active}[\gamma]=\textit{true}$} \label{algo:dsf-aDFS:l10b}}{
						$\textit{rSt}[\gamma].\textit{push}(u)$; \label{algo:dsf-aDFS-visit:l10} \;
					}
    			}
    		}
    	}\ElseIf{$\textit{open}[u]<+\infty \textbf{ and } \textit{close}[u]=+\infty$ \label{algo:dsf-aDFS-visit:l11}}{
    		add $(u,v)$ to $A_{f}$; \label{algo:dsf-aDFS-visit:l12} \;
    		\lElse{
    			add $(u,v)$ to $A_{c}$; \label{algo:dsf-aDFS-visit:l13}
    		}
    	}
    }
	\tcp{Check the ready-stack of $v$, \ie $\textit{rSt}[v]$}
    \While{$\textit{rSt}[v]\neq\emptyset$}{ \label{algo:dsf-aDFS-visit:l14}
    	$u\leftarrow \textit{rSt}[v].\textit{pop}()$;\label{algo:dsf-aDFS-visit:l15} \tcp{$u\in V_{\ocircle}$}
    	add $(u,v)$ to $A_{\pi}$; \label{algo:dsf-aDFS-visit:l16}\;
    	\textit{\bf for each} $t\in N^{\text{out}}_{\mathcal{A}}(u)$ \textit{\bf{do}} add $(u,t)$ to $A_{\textit{stalk}}$; \label{algo:dsf-aDFS-visit:l17}\;
    	$\textit{dsf-$\alpha$DFS-visit}(u, \A)$; \label{algo:dsf-aDFS-visit:l18} \;
		\colorbox{gray!20}{$\mathcal{D}.\textit{Union}(u, v)$; \tcp{\emph{(dsf-visit-1)}} } \label{algo:dsf-aDFS-visit:l18b} \;
    }
    $\textit{close}[v]\leftarrow (\textit{time}\leftarrow \textit{time}+1)$;\label{algo:dsf-aDFS-visit:l19}\;
    }
    \end{procedurealgo} 

    At this point we prove that the above mentioned claim concerning $\gamma$ and LCAs really holds.

    \begin{Lem}\label{lemma:gamma_is_LCA}
    Suppose $\textit{dsf-$\alpha$DFS-visit}(v, \A)$ visits some $v\in V$
     and considers an in-neighbour $u\in N^{\text{in}}_{\A}(v)\cap V_{\ocircle}$.
    Assume that $u$ is still unvisited, \ie $\textit{open}[u]=+\infty$,
    	and that $v$ is the last out-neighbour of $u$ that is being visited,
    	\ie that $\textit{cnt}[u]=0$.
    Let $\gamma$ be the vertex returned by $\mathcal{D}.\textit{find}(\textit{low\_}v)$, \ie the root of the disjoint set tree of $\textit{low\_}v$,
    where $\textit{low\_}v$ is the unique $x\in V$ such that $\textit{open}[x]=\textit{low\_ready}[u]$.
    If $\textit{active}[\gamma]=\textit{true}$ holds at that time, then the LCA of $N^{\text{out}}_{\A}(u)$ in $(V, A_{\pi})$ is really $\gamma$.
    \end{Lem}
		\begin{figure}[h!]
			\centering
			\begin{tikzpicture}[scale=.65]
					\node [squa, fill=black!10] (0) at (0, 7) {$r$};
					\node [style=none] (1) at (-9, 0) {};
					\node [style=none] (2) at (9, 0) {};
					\node [style=none] (3) at (9, 0) {};
					\node [style=none] (4) at (3, 2) {};
					\node [squa, fill=black!10] (5) at (1, 4) {$\gamma$};
					\node [style=none] (6) at (0, 2) {};
					\node [style=none] (7) at (-2.75, 2) {};
					\node [squa, fill=black!10] (8) at (3, 2) {$v$};
					\node [circ, fill=black!40] (9) at (0, 2) {$x$};
					\node [squa, fill=black!40] (10) at (-2.5, 2) {$\textit{low\_}v$};
					\node [style=none] (11) at (-6, 0) {};
					\node [style=none] (12) at (5, 0) {};
					\node [circ] (14) at (6, .4) {$u$};
					\node [style=none] (15) at (1.75, 2.25) {};
					\node [style=none] (16) at (2.25, 3) {};
					\node [style=none] (17) at (1.25, 3.25) {};
					\node [style=none] (18) at (1.5, 4.75) {};
					\node [style=none] (19) at (0.25, 5) {};
					\node [style=none] (20) at (0.5, 6) {};
					\draw [thick] (0.south west) to (1.center);
					\draw (0.south east) to node[above, xshift=-10ex, yshift=10ex] {$\mathcal{T}_\mathcal{P}$}  (3.center);
					\draw [thick] (1.center) to (12.center);
					\draw (12.center) to (3.center);
					\draw [thick] (5.south west) to (10.north east);
					\draw (5.south) to (9.north);
					\draw [thick] (10.south west) to (11.center);
					\draw [thick] (8.south east) to node[above, xshift=5ex, yshift=2.5ex] {$\mathcal{T}_\mathcal{P}\setminus \mathcal{T}_v$} (12.center);
					\draw [arrows=->] (8.west) to (15.center);
					\draw [arrows=->] (15.center) to (16.center);
					\draw [arrows=->] (16.center) to node[below, xshift=0ex, yshift=0ex] {\tiny tree} (17.center);
					\draw [arrows=->] (17.center) to (5.south);
					\draw [arrows=->, dotted, bend left=45] (14.west) to node[xshift=0ex, yshift=0ex] {\tiny stalk} (8.south);
					\draw [arrows=->, dotted, bend left, looseness=.75] (14.south west) to node[xshift=0ex, yshift=0ex] {\tiny stalk} (9.south east);
					\draw [arrows=->, dotted, bend right=315, looseness=0.75] (14.south) to node[xshift=0ex, yshift=0ex] {\tiny stalk} (10.south);
					\draw [arrows=->] (5.north) to (18.center);
					\draw [arrows=->] (18.center) to node[above, xshift=0ex, yshift=0ex] {$p_v$} (19.center);
					\draw [arrows=->] (19.center) to node[above, xshift=-6ex, yshift=-4ex] {$\mathcal{T}_v$} (20.center);
					\draw [arrows=->] (20.center) to (0.south);
				\end{tikzpicture}
					\caption{An illustration of Lemma~\ref{lemma:gamma_is_LCA}}\label{fig:lemma4}
		\end{figure}

    \begin{proof}
    Notice that $(V, A_{\pi})$ still grows as a forest during the execution of dsf-$\alpha$DFS().
    Indeed, if a new arc $(u,v)$ is added to $A_{\pi}$ it still holds that $\textit{open}[u]=+\infty$
    and $\textit{open}[v]<+\infty$; no cycle can be formed.
    Thus, assuming $\textit{dsf-$\alpha $DFS-visit}(v, \A)$ is invoked for some $v\in V$,
    	we can consider the unique maximal tree $\mathcal{T}_v$ in $(V, A_{\pi})$ containing $v$ and comprising only non-white vertices
    		-- \ie constructed \emph{until} the time of that particular invocation.
    Let $p_v$ be the path in $\mathcal{T}_v$ going from $v$ up to the root $r$ of $\mathcal{T}_v$.
    By properties \emph{(dsf-visit-1, dsf-visit-2)} and by the definition of $\textit{low\_}v$,
    	and since $\gamma=\mathcal{D}.\textit{find}(\textit{low\_}v)$	and $\gamma$ is active by hypothesis,
    		then $\gamma$ lies on $p_v$.
    Thus, $\gamma$ must be the LCA of $\textit{low\_}v$ and $v$ in $\mathcal{T}_v$ (possibly $\gamma=\textit{low\_}v$).
    We argue that $N^{\text{out}}_{\A}(u)\subseteq {\mathcal{T}^\gamma_v}$,
    where ${\mathcal{T}^\gamma_v}$ is the maximal subtree of $\mathcal{T}_v$ rooted at $\gamma$.
    Indeed, by \emph{(dsf-visit-2)}, the $\text{I}_{\text{low}}$ invariant holds:
    \[\textit{open}[\textit{low\_}v]=\min\big\{ \textit{open}[x]\mid x\in N^{\text{out}}_{\A}(u)\big\}.\]
    So, when $\textit{cnt}[u]=0$, and since $\gamma$ is an ancestor of $\textit{low\_}v$, then:
    	\[ \forall{x\in N^{\text{out}}_{\A}(u)}\textit{open}[\gamma] \leq \textit{open}[\textit{low\_}v] \leq \textit{open}[x]<+\infty.\]
    Notice all vertices in $\mathcal{T}_v$ which are not descendants of $\gamma$ still
    	have a smaller opening timestamp than $\gamma$ (\ie they were all visited before $\gamma$),
    and all those which are proper descendants of $\gamma$ have a greater opening timestamp than $\gamma$.
    All these combined, it must be $N^{\text{out}}_{\A}(u)\subseteq {\mathcal{T}^{\gamma}_v}$.
    So, $\gamma$ is a common ancestor of all out-neighbours of $u$ in $\mathcal{T}_v$; 
	but $\gamma$ is also the LCA of $\{\textit{low\_}v,v\}\subseteq N^{\text{out}}_{\A}(u)$,
    	this means that $\gamma$ is the LCA of all $N^{\text{out}}_{\A}(u)$ in $\mathcal{T}_v$.
    \end{proof}

    By Lemma~\ref{lemma:gamma_is_LCA}, Proposition~\ref{prop:aDFS-jungle} holds even for dsf-$\alpha$DFS, proving its correctness.

		Concerning time complexity, by relying on technical results offered in~\cite{GT85},
			\textit{dsf-$\alpha$DFS()} can be implemented so to run in linear-time on a RAM machine.
		Concretely, \cite{GT85} showed that the \emph{incremental-tree set-union problem} can be solved in linear-time on RAMs.
		The disjoint-sets union-tree $T$ of $\mathcal{D}$ is revealed one vertex at a time by attaching new singleton vertices to $T$ incrementally and in interleaving with the $\mathcal{D}.\textit{Find}()$ operations (that can possibly be performed on those vertices that have already been revealed previously). The vertices $u$ that are incrementally revealed and attached must be new singletons that were never attached before (\ie there is only one underlying union-tree $T$, that keeps growing, and many singleton vertices attached incrementally).

A moment's reflection reveals that the incremental-tree set-union problem does encompass the way in which \textit{dsf-$\alpha$DFS()} grows the union-tree.
Indeed recall that our vertices are always attached incrementally to the union-tree during the backtracking and, thus, in a post-ordering.
Concretely, recall that \textit{dsf-$\alpha$DFS()} performs a depth-first search of the union-tree,
and observe that once the docking points of the circled vertices have been decided and the post-order visiting of the vertices has been fixed,
then one can also forget about the fact that the vertices belong to two players and reason about the union-tree downstream of that, as if it were a traditional dfs-tree of uncolored vertices.
By this we mean that the distinction of the vertices into color classes, squares and circles, only affects the particular post-ordering that is being chosen (\ie the particular order in which the $\mathcal{D}.\textit{Union}()$ operations are performed) but not the underlying fundamental graph structure.
So~\cite{GT85} applies. Notice that we would need just a rather special case of the incremental-tree set-union problem, \ie the one in which
the $\mathcal{D}.\textit{Union}()$ operations are always done in a post-ordering simultaneously with the dfs backtracking. The following holds.
\begin{Thm}\label{thm:aDFS:complexityRAM}
	Given an input $\alpha$graph $\A$ on vertex set $V$ and arc set $A$, \textit{dsf-$\alpha$DFS}($\A$) halts in $\Theta(|V|+|A|)$ linear-time on a RAM machine, provided that the dfs data structure $\mathcal{D}$ is implemented as proposed in~\cite{GT85}.
\end{Thm}

If the dsf data structure $\mathcal{D}$ is implemented more traditionally as proposed in~\cite{Tar75}, \ie with ranked-unions and path-compressions,
then \textit{dsf-$\alpha$DFS()} runs in Ackermann-linear-time even on a pointer machine.
As one would expect, due to it's simplicity, this variant would also perform well in practice.

	\begin{Thm}\label{thm:aDFS:complexityPTM}
		Given an input $\alpha$graph $\A$ on vertex set $V$ and arc set $A$, \textit{dsf-$\alpha$DFS}($\A$) halts in $O(|V|+|A|\alpha(|A|, |V|))$ Ackermann-linear-time on a pointer machine, provided that the dfs data structure $\mathcal{D}$ is implemented with ranked-unions and path-compressions as in~\cite{Tar75}.
	\end{Thm}

		We leave open the question of whether \textit{dsf-$\alpha$DFS()} can be implemented so that to run in $\Theta(|V|+|A|)$ linear-time on pointer machines,
		 see Section~\ref{ref:relatedfutureworks} for further discussion.

    \section{Linear-Time Algorithm for Safe-Alternating SCCs}\label{sect:algo-slip}
    In order to offer a linear-time safe-$\alpha$SCCs decomposition algorithm some more technical machinery is still needed, the catalyst being Definition~\ref{def:lowlink} below.

    It is shown that the problem of computing safe-$\alpha$SCCs of a given $\alpha$graph $\A$
      can be tackled by finding the roots of the components' subtrees in the $\alpha$jungle $\J_\A$, this is
    reminiscent to what happens in Tarjan's algorithm for the classical problem of decomposing a directed graph into SCCs.

    So we have identified an efficient procedure to decide whether a vertex is the root of a safe-$\alpha$SCC subtree in $\J_\A$.
    It is based on a \emph{lowlink} indexing gamifying the lowlink calculation proposed in~\cite{Tar72}.
	
	Thus, before presenting our linear-time algorithm for computing safe-$\alpha$SCC, 
	it is instructive to recall the classical algorithm introduced by Tarjan \cite{Tar72} 
	for computing strongly connected components (SCCs) in directed graphs.

\subsection{Tarjan's Algorithm and Lowlink Computation}

Before presenting our linear-time algorithm for safe-alternating strongly connected components, 
it is instructive to recall the classical algorithm introduced by Tarjan~\cite{Tar72}, 
which uses a standard (forward) depth-first traversal to compute strongly connected components (SCCs) of a directed graph.

Tarjan's algorithm performs a single depth-first traversal of a directed graph \( G=(V,A) \), 
assigning a discovery index \(\textit{open}[v]\) to each vertex \( v \) at the moment it is first visited by DFS.

Additionally, it computes the value \(\textit{lowlink}[v]\) for each vertex \( v \), defined as 
the smallest opening time \(\textit{open}[u]\) of any vertex \( u \) that lies in the same strongly connected component as \( v \) 
and is reachable from \( v \) by traversing zero or more tree arcs followed by at most one frond or cross-link arc.

Finally, Tarjan's algorithm identifies a vertex \( v \) as the root of an SCC if and only if 
\(\textit{lowlink}[v] = \textit{open}[v]\). 
Thus, all vertices belonging to the same SCC form a subtree in the DFS tree, 
with the component root being the subtree's root.
This classical definition and algorithm serve as the foundation and intuition for the 
safe-alternating SCC algorithm presented in the following sections.

\subsection{$\alpha$Lowlink}

The concept of \emph{$\alpha$lowlink} adapts Tarjan's classical lowlink indexing to the 
alternating setting of an $\alpha$jungle $\J$ constructed from an $\alpha$graph $\A$. 

The intuitive idea behind the definition of $\alpha\textit{lowlink}(v)$ mirrors the classical lowlink concept, 
transposed into the reverse DFS setting.
Specifically, $\alpha\textit{lowlink}(v)$ is the smallest opening time $\textit{open}[u]$ 
among all vertices $u$ belonging to the same safe-$\alpha$SCC as $v$, 
from which vertex $v$ can be reached by traversing at most one frond or cross-link arc, 
followed by a (possibly empty) sequence of tree arcs.

However we cannot practically carry out this computation directly, 
since we do not yet have an explicit method to determine whether two vertices belong to the same safe-$\alpha$SCC.

To address this issue, Definition~\ref{def:lowlink} explicitly incorporates additional structural constraints: 
specifically, it requires the existence of a common ancestor vertex $\gamma$ shared by $u$ and $v$, 
such that $u$ and $\gamma$ lie within the same safe-$\alpha$SCC. 

While somewhat stronger and initially more complex in form, 
this will become instrumental in enabling both practical algorithmic efficiency in the procedures 
and theoretical clarity in the correctness proofs.

    \begin{Def}\label{def:lowlink}
    Let $\J$ be an $\alpha$jungle constructed over an $\alpha$graph $\A$ on vertex set $V$.
    Let the vertices be timestamped by $\textit{open}[]:V\rightarrow \N$, and let $\{\mathcal{P}_i\}_{i=1}^k$ be the $\alpha$palm-trees of $\J$
    each having vertex set $V_i$ and arc set $A_i={A_i}_{\pi}\cup {A_i}_{f}\cup {A_i}_{c}\cup {A_i}_{s}$.

     $\textit{$\alpha$lowlink}_\J:V\rightarrow \N$ is defined as the following minimum index for every $v\in V$:
    \begin{align*}
     \textit{$\alpha$lowlink}_\J(v) \doteq
    	\min \big\{\textit{open}[v]\big\}\cup & \big\{ \textit{open}[u] \mid u\in V\setminus\{v\} \text{ and } \exists{i\in [k]}\; \text{such that the following two hold:}\; \\
    		& \textit{($\alpha$ll-1) } \exists{t\geq 1}\, \exists{( u, v_1, \ldots, v_{t-1}, (v_t = v) ) \in ({V_i})^+} \text{such that:}  \\
    		&\;\;\;\;\;\;\;\;\;\;\;\;\;\;\; \textit{(a) } (u,v_1)\in {A_i}_{f} \cup {A_i}_{c}; \text{(\ie $(u,v_1)$ is either a frond or cross-arc)}  \\
    		&\;\;\;\;\;\;\;\;\;\;\;\;\;\;\; \textit{(b) } \text{ if } t\geq 2, \forall{j\in \{1, \ldots, t-1\}}\,\text{ it holds } (v_j, v_{j+1})\in {A_i}_{\pi}. \\
    	\text{ and }\;\; & \textit{($\alpha$ll-2) } \exists{\gamma\in V_i} \text{ such that: } \\
      	&\;\;\;\;\;\;\;\;\;\;\;\;\;\;\; \textit{(a) } \gamma \text{ is a common ancestor of } u \text{ and } v \text{ in } (V_i, {A_i}_{\pi}); \\
    		&\;\;\;\;\;\;\;\;\;\;\;\;\;\;\; \textit{(b) } \gamma \text{ and } u \text{ are in the same safe-$\alpha$SCC of $\A$, \ie $\gamma\in\C_u$.} \\
    		& \big\}.
    \end{align*}
    (where, for any $u\in V$, $\C_u$ denotes the unique safe-$\alpha$SCC of $\A$ which includes vertex $u$; 
		also notice that it is possible that $u=\gamma$ or $\gamma=v$)
    \end{Def}

		\begin{figure}[h!]
			\centering
			\begin{tikzpicture}[scale=.75]
					\node [squa] (0) at (0, 9) {$r$};
					\node [style=none] (1) at (-10, 0) {};
					\node [style=none] (2) at (2, 0) {};
					\node [squa] (3) at (-0.5, 6.5) {$\gamma$};
					\node [squa] (4) at (-1, 4.5) {$v$};
					\node [squa] (5) at (-2.75, 0.5) {$v_1$};
					\node [squa] (6) at (-2.5, 1.5) {$v_2$};
					\node [squa] (7) at (-1.25, 3.25) {$v_{t-1}$};
					\node [squa] (8) at (-4.25, 2.75) {$u$};
					\node [style=none] (9) at (-1.5, 2.5) {$\iddots$};

					\draw (0.south west) to (1.center);
					\draw (1.center) to (2.center);
					\draw [thick] (0.south) to (3.north);
					\draw [thick] (3.south) to (4.north);
					\draw [thick] (4.south) to (7.north);
					\draw [thick] (7.south) to (6.north);
					\draw [thick] (6.south) to (5.north);
					\draw [thick] (3.south west) to (8.north east);
					\draw [arrows=->, thick, bend right, in=120, out=-70, looseness=1.75, dashed]
						(3.west) to node[above, xshift=-3.5ex, yshift=-.5ex] {$\gamma\in\C_u$} (8.east);
					\draw [arrows=->, dashed, bend right=45] (8.south) to node[xshift=-1ex, yshift=0ex] {\tiny frond or cross} (5.west);
					\draw (0.south east) to (2.center);
			\end{tikzpicture}
		\caption{An illustration of Definition~\ref{def:lowlink}}
		\end{figure}

  Now, in order to proceed on this route, we must overcome some obstructions.
  Unfortunately it's not generally true that, if $\C\subseteq V$ is a safe-$\alpha$SCC of an $\alpha$graph $\A$, then, $\C$ induces a subtree $T_\C$ in $\J_\A$ --
  	if $\J_\A$ is the $\alpha$jungle constructed during an $\alpha$DFS() as defined in Section~\ref{sect:aDFS}.
  And even when it's true, say by chance, still it is not generally true that a vertex $v$ of $\A$ is the root of some
  \emph{safe-}$\alpha$SCC if and only if $\textit{$\alpha$lowlink}_{\J_\A}(v)=\textit{open}[v]$ as it was in the SCCs algorithm.

  Still, for all this to happen, we claim that a conceptually simple (but technically non-trivial) modification to the $\alpha$DFS() (and,	thus, to the structure of the $\alpha$jungle) can be introduced.
  To better illustrate the issue, let us first consider the following Examples~\ref{ex:safe-aSCC1}~and~\ref{ex:safe-aSCC}.

    \begin{figure}[t!]
    	\begin{minipage}{\columnwidth}
    	\centering
    	\subfloat[An $\alpha$graph $\A_1$\label{fig:ex:safe-aSCC:agraph1}]{
    	\begin{tikzpicture}[arrows=->, scale=.75]
    		\node[squa, label={}] (tAA) {$a$};
    		\node[squa, below=of tAA, label={}] (tA) {$b$};
    		\node[squa, right=of tB, label={}, yshift=0ex, xshift=0ex] (tC) {$c$};
    		\node[squa, below=of tA, label={}, yshift=0ex, xshift=0ex] (tB) {$e$};
    		\node[circ, below=of tB, label={}, yshift=-6ex, xshift=0ex] (tD) {$g$};
    		\node[squa, below right=of tA, label={}, yshift=0ex, xshift=1ex] (tE) {$f$};
    		\node[squa, left=of tB, label={}, yshift=0ex, xshift=-1ex] (tF) {$d$};
    		\node[squa, below=of tD, label={}, yshift=0ex, xshift=0ex] (tG) {$h$};

    		\draw[] (tA) to [xshift=0ex, yshift=0ex] node[xshift=2ex, yshift=-1ex] {} (tAA);
    		\draw[] (tB) to [xshift=0ex, yshift=0ex] node[xshift=2ex, yshift=-1ex] {} (tA);
    		\draw[] (tC) to [bend right=40, xshift=0ex, yshift=0ex] node[xshift=2ex, yshift=-1ex] {} (tAA);
    		\draw[] (tD) to [bend left=40, xshift=0ex, yshift=0ex] node[below, xshift=0ex, yshift=1.5ex] {} (tF);
    		\draw[] (tD) to [bend right=40, xshift=0ex, yshift=0ex] node[below, xshift=0ex, yshift=1.5ex] {} (tE);
    		\draw[] (tD) to [bend right=0, xshift=0ex, yshift=0ex] node[below, xshift=0ex, yshift=1.5ex] {} (tB);
    		\draw[] (tE) to [xshift=0ex, yshift=0ex] node[right,xshift=0ex, yshift=0ex] {} (tA.east);
    		\draw[] (tG) to [xshift=0ex, yshift=0ex] node[right,xshift=0ex, yshift=0ex] {} (tD);
    	  \draw[] (tF) to [xshift=0ex, yshift=0ex] node[right,xshift=0ex, yshift=0ex] {} (tA.west);
    		\draw[] (tAA.west) to [out=170, in=180, xshift=-5ex, yshift=0ex] node[left,xshift=0ex, yshift=0ex] {} (tG.west);
    		\draw[] (tG.east) to [out=0, in=0, xshift=-5ex, yshift=0ex] node[left,xshift=0ex, yshift=0ex] {} (tC.east);
    	\end{tikzpicture}
    	}
    \qquad
    \subfloat[The corresponding $\alpha$jungle $\mathcal{\J}_{\A_1}$ traced by $\alpha$DFS(), showing timestamps and $\alpha$lowlinks\label{fig:ex:safe-aSCC:algorun1}]{
    \begin{tikzpicture}[arrows=->, scale=.75]
    	\node[squa, label={[yshift=-.5ex] \tiny $1|16|(1)$}] (tAA) {$a$};
    	\node[squa, below=of tAA, label={[xshift=3ex, yshift=-.5ex] \tiny $2|13|(1)$}] (tA) {$b$};
    	\node[squa, right=of tA, label={[xshift=3.5ex, yshift=-.5ex] \tiny $14|15|(10)$}, yshift=0ex, xshift=0ex] (tC) {$c$};
    	\node[squa, below=of tA, label={[xshift=2.5ex, yshift=-.5ex] \tiny $5|6|(5)$}, yshift=0ex, xshift=0ex] (tB) {$e$};
    	\node[circ, below=of tB, label={[xshift=2ex, yshift=-.5ex] \tiny $9|12|(1)$}, yshift=-6ex, xshift=0ex] (tD) {$g$};
    	\node[squa, below right=of tA, label={[xshift=1.5ex, yshift=-.5ex] \tiny $7|8|(7)$}, yshift=0ex, xshift=1ex] (tE) {$f$};
    	\node[squa, left=of tB, label={[xshift=-1ex, yshift=-.5ex] \tiny $3|4|(3)$}, yshift=0ex, xshift=-1ex] (tF) {$d$};
    	\node[squa, below=of tD, label={[xshift=0ex, yshift=-5ex] \tiny $10|11|(1)$}, yshift=0ex, xshift=0ex] (tG) {$h$};

    	\draw[thick] (tA) to [xshift=0ex, yshift=0ex] node[xshift=2ex, yshift=0ex] {\tiny tree} (tAA);
    	\draw[thick] (tB) to [xshift=0ex, yshift=0ex] node[xshift=2ex, yshift=.5ex] {\tiny tree} (tA);
    	\draw[thick] (tC) to [bend right=40, xshift=0ex, yshift=0ex] node[xshift=3ex, yshift=-1ex] {\tiny tree} (tAA);
    	\draw[dotted] (tD) to [bend left=40, xshift=0ex, yshift=0ex] node[below, xshift=0ex, yshift=1.5ex] {\tiny stalk} (tF);
    	\draw[dotted] (tD) to [bend right=40, xshift=0ex, yshift=0ex] node[below, xshift=0ex, yshift=1.5ex] {\tiny stalk} (tE);
    	\draw[dotted] (tD) to [bend right=0, xshift=0ex, yshift=0ex] node[below, xshift=0ex, yshift=1.5ex] {\tiny stalk} (tB);
    	\draw[thick] (tE) to [xshift=0ex, yshift=0ex] node[right,xshift=0ex, yshift=0ex] {\tiny tree} (tA.east);
    	\draw[thick] (tG) to [xshift=0ex, yshift=0ex] node[left,xshift=0ex, yshift=0ex] {\tiny tree} (tD);
    	\draw[thick] (tF) to [xshift=0ex, yshift=0ex] node[xshift=-2ex, yshift=0ex] {\tiny tree} (tA.west);
    	\draw[dashed] (tAA.west) to [out=170, in=180, xshift=-5ex, yshift=0ex] node[left,xshift=1ex, yshift=0ex] {\tiny frond} (tG.west);
    	\draw[dashed] (tG.east) to [out=0, in=0, xshift=-5ex, yshift=0ex] node[right,xshift=0ex, yshift=0ex] {\tiny cross} (tC.east);
    	\draw[thick] (tD) to [bend left=30, xshift=-5ex, yshift=0ex] node[xshift=-1.5ex, yshift=-1ex] {\tiny tree} (tA);
    \end{tikzpicture}
    }
    \end{minipage}
    \caption{An illustration of Example~\ref{ex:safe-aSCC1}}\label{fig:ex:safe-aSCC1}
    \end{figure}

    \begin{Exa}\label{ex:safe-aSCC1}
    Consider the $\alpha$graph $\A_1=(V, A, \langle V_\square, V_\ocircle\rangle)$ shown in \figref{fig:ex:safe-aSCC:agraph1} where
    $V_\square=\{a,b,c,d,e,f,h\}$ and $V_\ocircle=\{g\}$, where $V=V_\square \cup V_\ocircle$ and
    $A=\{(a,h), (b,a), (c,a), (d,b), (e,b), (f,b), (g,d), (g,e), (g,f), (h,c)\}$.

    \figref{fig:ex:safe-aSCC:algorun1} shows the $\alpha$jungle $\mathcal{\J}_{\A_1}$ tracing the execution of $\alpha$DFS() on input $\A_1$.
    Timestamps and $\alpha$lowlinks are shown above each vertex
    	(denoted: $\langle \text{open}[v]\rangle | \langle \text{close}[v]\rangle | (\langle \alpha\text{lowlink}[v]\rangle)$, for $v\in V$).
    Notice $(g,b)$ is an arc in $\J_{\A_1}$ but not in $\A_1$.
    Concerning the safe-$\alpha$SCCs of $\A_1$, a moment's reflection reveals that they are $\C_1=\{a,h,c\}$
    and all of the remaining vertices are singleton safe-$\alpha$SCCs.
    Notice that $b$ is always $\alpha$reachable from $\C_1$, but Player~$\ocircle$ decides how to reach it by controlling $g$.

    The main issue here is that $\C_1$ doesn't induce a subtree in $\J_{\A_1}$ because $(a,h)$ is a frond, $(h,c)$ is a cross-arc,
    and $g\not\in\C_1$ act as a vertex in the middle between $h$ and $a$.
    The reason being that $g$ joined the $\alpha$jungle $\J_{\A_1}$ by attaching to parent $b$
    (which is fine for deciding just safe-$\alpha$reachability relations but it's not for identifying safe-$\alpha$SCCs).
    \end{Exa}

    \begin{figure}[h!]
    	\begin{minipage}{\columnwidth}
    	\centering
    	\subfloat[An $\alpha$graph $\A_2$\label{fig:ex:safe-aSCC:agraph}]{
    	\begin{tikzpicture}[arrows=->, scale=.75]
    		\node[squa, label={}] (tA) {$a$};
    		\node[squa, below=of tA, label={}, yshift=0ex, xshift=0ex] (tB) {$c$};
    		\node[squa, below=of tB, label={}, yshift=2ex, xshift=0ex] (tC) {$e$};
    		\node[circ, below=of tB, label={}, yshift=-6ex, xshift=0ex] (tD) {$f$};
    		\node[squa, below right=of tA, label={}, yshift=0ex, xshift=1ex] (tE) {$d$};
    		\node[squa, left=of tB, label={}, yshift=0ex, xshift=-1ex] (tF) {$b$};
    		\node[squa, below=of tD, label={}, yshift=0ex, xshift=0ex] (tG) {$g$};
    		\node[squa, below=of tG, label={}, yshift=0ex, xshift=0ex] (tH) {$h$};

    		\draw[] (tB) to [xshift=0ex, yshift=0ex] node[xshift=2ex, yshift=-1ex] {} (tA);
    		\draw[] (tD) to [xshift=0ex, yshift=0ex] node[below, xshift=0ex, yshift=1.5ex] {} (tC);
    		\draw[] (tC) to [xshift=0ex, yshift=0ex] node[below, xshift=2ex, yshift=1.25ex] {} (tB);
    		\draw[] (tB) to [bend right=55, xshift=0ex, yshift=0ex] node[below, xshift=-2.25ex, yshift=1.25ex] {} (tC);
    		\draw[] (tD) to [bend left=40, xshift=0ex, yshift=0ex] node[below, xshift=0ex, yshift=1.5ex] {} (tF);
    		\draw[] (tD) to [bend right=40, xshift=0ex, yshift=0ex] node[below, xshift=0ex, yshift=1.5ex] {} (tE);
    		\draw[] (tE) to [xshift=0ex, yshift=0ex] node[right,xshift=0ex, yshift=0ex] {} (tA.east);
    		\draw[] (tG) to [xshift=0ex, yshift=0ex] node[right,xshift=0ex, yshift=0ex] {} (tD);
    		\draw[] (tH) to [xshift=0ex, yshift=0ex] node[right,xshift=0ex, yshift=0ex] {} (tG);
    	  \draw[] (tF) to [xshift=0ex, yshift=0ex] node[right,xshift=0ex, yshift=0ex] {} (tA.west);
    		\draw[] (tG) to [bend right=40, xshift=0ex, yshift=0ex] node[left,xshift=0ex, yshift=0ex] {} (tH);
    		\draw[] (tA.west) to [out=170, in=180, xshift=-5ex, yshift=0ex] node[left,xshift=0ex, yshift=0ex] {} (tH.west);
    	\end{tikzpicture}
    	}
    \qquad
    \subfloat[The corresponding $\alpha$jungle $\mathcal{\J}_{\A_2}$ traced by $\alpha$DFS(), showing timestamps and $\alpha$lowlinks\label{fig:ex:safe-aSCC:algorun}]{
    \begin{tikzpicture}[arrows=->, scale=.75]
    	\node[squa, label={above : \tiny $1|16|(1)$}] (tA) {$a$};
    	\node[squa, below=of tA, label={left, xshift=.75ex : \tiny $2|5|(2)$}, yshift=0ex, xshift=0ex] (tB) {$c$};
    	\node[squa, below=of tB, label={left, xshift=.75ex : \tiny $3|4|(2)$}, yshift=2ex, xshift=0ex] (tC) {$e$};
    	\node[circ, below=of tB, label={right: \tiny $10|15|(1)$}, yshift=-7ex, xshift=0ex] (tD) {$f$};
    	\node[squa, below right=of tA, label={above : \tiny $8|9|(8)$}, yshift=0ex, xshift=1ex] (tE) {$d$};
    	\node[squa, left=of tB, label={above : \tiny $6|7|(6)$}, yshift=0ex, xshift=-1ex] (tF) {$b$};
    	\node[squa, below=of tD, label={right : \tiny $11|14|(1)$}, yshift=0ex, xshift=0ex] (tG) {$g$};
    	\node[squa, below=of tG, label={right : \tiny $12|13|(1)$}, yshift=0ex, xshift=0ex] (tH) {$h$};

    	\draw[thick] (tB) to [xshift=0ex, yshift=0ex] node[xshift=1.5ex, yshift=-1.5ex] {\tiny tree} (tA);
    	\draw[dotted] (tD) to [xshift=0ex, yshift=0ex] node[below, xshift=0ex, yshift=1ex] {\tiny stalk} (tC);
    	\draw[thick] (tC) to [xshift=0ex, yshift=0ex] node[below, xshift=2ex, yshift=1.25ex] {\tiny tree} (tB);
    	\draw[dashed] (tB) to [bend right=55, xshift=0ex, yshift=0ex] node[below, xshift=-2.25ex, yshift=1.25ex] {\tiny frond} (tC);
    	\draw[thick] (tD) to [bend right=45, xshift=0ex, yshift=0ex] node[below, xshift=0ex, yshift=2.5ex] {\tiny tree} (tA);
    	\draw[dotted] (tD) to [bend left=40, xshift=0ex, yshift=0ex] node[below, xshift=0ex, yshift=1.5ex] {\tiny stalk} (tF);
    	\draw[dotted] (tD) to [bend right=40, xshift=0ex, yshift=0ex] node[below, xshift=0ex, yshift=1.5ex] {\tiny stalk} (tE);
    	\draw[thick] (tE) to [xshift=0ex, yshift=0ex] node[right,xshift=-1ex, yshift=.75ex] {\tiny tree} (tA.east);
    	\draw[thick] (tG) to [xshift=0ex, yshift=0ex] node[right,xshift=0ex, yshift=0ex] {\tiny tree} (tD);
    	\draw[thick] (tH) to [xshift=0ex, yshift=0ex] node[right,xshift=-.5ex, yshift=0ex] {\tiny tree} (tG);
      \draw[thick] (tF) to [xshift=0ex, yshift=0ex] node[left,xshift=1ex, yshift=.75ex] {\tiny tree} (tA.west);
    	\draw[dashed] (tG) to [bend right=40, xshift=0ex, yshift=0ex] node[left,xshift=.75ex, yshift=0ex] {\tiny frond} (tH);
    	\draw[dashed] (tA.west) to [out=170, in=180, xshift=-5ex, yshift=	0ex] node[left,xshift=1ex, yshift=0ex] {\tiny frond} (tH.west);
    \end{tikzpicture}
    }
    \end{minipage}
    \caption{An illustration of Example~\ref{ex:safe-aSCC}}\label{fig:ex:safe-aSCC}
    \end{figure}

    \begin{Exa}\label{ex:safe-aSCC}
    Consider the $\alpha$graph $\A_2=(V, A, \langle V_\square, V_\ocircle\rangle)$ of \figref{fig:ex:safe-aSCC:agraph}:
    $V_\square=\{a,b,c,d,e,g,h\}$ and $V_\ocircle=\{f\}$, where $V=V_\square \cup V_\ocircle$ and
    $A=\{(a,h), (b,a), (c,a), (c,e), (d,a), (f,b), (f,e), (f,d), (g,f), (g,h), (h,g)\}$.

    \figref{fig:ex:safe-aSCC:algorun} shows the $\alpha$jungle $\mathcal{\J}_{\A_2}$ tracing the execution of $\alpha$DFS() on input $\A_2$.
    	Timestamps and $\alpha$lowlinks are shown above each vertex
    	(denoted: $\langle \text{open}[v]\rangle | \langle \text{close}[v]\rangle | (\langle \alpha\text{lowlink}[v]\rangle)$, for $v\in V$).
    	Notice that the arc $(f,a)$ belongs to $\J_{\A_2}$ but not to $\A_2$.
    Concerning the safe-$\alpha$SCCs of $\A_2$, a moment's reflection reveals that they are $\C_1=\{c,e\}$, $\C_2=\{g,h\}$ and
    	all of the remaining vertices are singleton safe-$\alpha$SCCs.
    Notice that both $\C_1$ and $\C_2$ induce a subtree in $\J_{\A_2}$.
    	Notice $c$ is the root of $\C_1$ and $g$ is that of $\C_2$.

    	But	$\textit{$\alpha$lowlink}_{\J_\A}(g)=1\neq 11=\textit{open}[g]$, so $g$ can't be recognized as a root simply by testing the $\alpha$lowlink.
    	The issue is still that $f$ joined $\J_{\A_2}$ by attaching to parent $a$.
    \end{Exa}
	
	 \begin{algorithm}[t]
    \caption{safe-$\alpha$SCC}\label{algo:STCC}
    \DontPrintSemicolon
    \nonl \SetKwProg{Fn}{Procedure}{}{}
	\normalsize
    \Fn{$\textit{safe-$\alpha$SCC}(\mathcal{A})$}{
        \SetKwInOut{Input}{input}
        \SetKwInOut{Output}{output}
    \Input{An $\alpha$graph $\mathcal{A}=(V, A, ( V_{\ocircle}, V_{\square} ) )$.}
    \Output{The safe-$\alpha$SCC of $\mathcal{A}$.}
    \ForEach{$u\in V$}{ \label{algo:STCC:l1}
      $\textit{open}[u]\leftarrow +\infty$; \label{algo:STCC:l2}\;
      \colorbox{gray!20}{$\textit{$\alpha$lowlink}[u]\leftarrow +\infty$; \label{algo:STCC:l3}}\;
      \colorbox{gray!20}{$\textit{on\_stack}[u]\leftarrow\textit{false}$; \label{algo:STCC:l4}}\;
      $\mathcal{D}.\textit{make\_set}(u)$; \label{algo:STCC:l5}\;
      $\textit{rSt}[u]\leftarrow \emptyset$; \label{algo:STCC:l6}\;
      \If{$u\in V_{\ocircle}$}{ \label{algo:STCC:l7}
        $\textit{low\_ready}[u]\leftarrow +\infty$; \label{algo:STCC:l8}\;
        $\textit{cnt}[u]\leftarrow |N^{\text{out}}_\A(u)|$; \label{algo:STCC:l9} \;
      }
    }
    $\textit{time}\leftarrow 1$; $\textit{cSt}\leftarrow \emptyset$;\; \label{algo:STCC:l10}
    \ForEach{$u\in V_{\square}$}{ \label{algo:STCC:l11}
      \If{$\textit{open}[u]=+\infty$}{ \label{algo:STCC:l12}
        $\textit{safe-$\alpha$SCC-visit}(u, \A)$;\label{algo:STCC:l13} \;
      }
    }
    \ForEach{$u\in V_{\ocircle}$ \label{algo:STCC:l14}}{
      \If{$\textit{open}[u]=+\infty$ \label{algo:STCC:l15}}{
      $\textit{open}[u]\leftarrow (\textit{time}\leftarrow\textit{time}+1)$; \label{algo:STCC:l16} \;
      \colorbox{gray!20}{$\textit{$\alpha$lowlink}[u]\leftarrow \textit{open}[u]$; \label{algo:STCC:l18}} \;
      }
    }
    }
    \end{algorithm}

    A revision of the $\alpha$DFS() is next provided in order to decompose a graph into safe-$\alpha$SCC.
    		Based on dsf-$\alpha$DFS() (Algo.~\ref{algo:aDFS}),
    			but still, with three additional and distinctive rules for identifying the components:

    	(r1) All vertices that have already been visited during the search,
    				but whose safe-$\alpha$SCC has not been identified yet, are stored on an auxiliary stack named $\textit{cSt}$ (\ie the component stack);

    	(r2) $\textit{cSt}$ shrinks back when the condition $\textit{$\alpha$lowlink}(v)=\textit{open}[v]$
    				is met at the end of the visiting subprocedure (see Propositions~\ref{prop:alowlink_correct}~and~\ref{prop:STCC:output} below for correctness),
    					at that point a brand new safe-$\alpha$SCC $\C$ is identified and detached.

    	(r3) The \emph{$\ocircle$-attraction-rule} that allows circled vertices to join $\J_\A$ is revised by restriction: now a circled vertex $u\in \textit{rSt}[v]$ joins $\J_\A$ as a child of $v$ if and only if \emph{all} of its out-neighbours are still found on the component stack $\textit{cSt}$; otherwise, $u$ is discarded. (see lines~\ref{algo:STCCs-visit:l22}-26 of Proc.~\ref{algo:STCCs-visit}, particularly, line~23)

    			\textbf{Remark.} The safe-$\alpha$SCC algorithm doesn't need to build the $\alpha$jungle $\J_\A$ explicitly
    		(\ie in principle there might be no real need to store it in memory;
    	still, an $\alpha$jungle is defined implicitly just by following the trace of vertices that are visited and closed during the search.
    	As it will be convenient to consider the $\alpha$jungle $\J_\A$ during the correctness proof,
    		we shall continue refer to it anyways. \qed

    More details follow. The main procedure is now renamed $\textit{safe-$\alpha$SCC()}$ (Algo.~\ref{algo:STCC}).
    Given an $\alpha$graph $\A$ in input, it aims at identifying and printing out all the safe-$\alpha$SCC $\C_1, \ldots, \C_k$ of $\A$ without repetitions.

    A subprocedure named $\textit{safe-$\alpha$SCC-visit()}$ (Proc.~\ref{algo:STCCs-visit}) is also employed for visiting the vertices.

    $\textit{safe-$\alpha$SCC()}$ goes like \textit{dsf-$\alpha$DFS()},
    	the major distinction being that now there is also an addional component stack $\textit{cSt}$ (which is initialized empty) and
    an additional flag vector $\textit{on\_stack}:V\rightarrow \{\textit{true}, \textit{false}\}$ (where all flags are initially $\textit{false}$).

    $\textit{safe-$\alpha$SCC-visit}(v,\A)$ (Proc.~\ref{algo:STCCs-visit}) goes like \textit{dsf-$\alpha$DFS-visit()},
    	but now there are some new features for computing the $\alpha$lowlinks and for keeping track of the components.

    \begin{procedurealgo}[H]
    \caption{$\textit{safe-$\alpha$SCC-visit()}$}\label{algo:STCCs-visit}
    \DontPrintSemicolon
    \nonl \SetKwProg{Fn}{Procedure}{}{}
    \normalsize
    \Fn{$\textit{safe-$\alpha$SCC-visit}(v,\A)$}{
    		\SetKwInOut{Input}{input}
    		\SetKwInOut{Output}{output}
     \Input{A vertex $v\in V$.}
    $\textit{open}[v]\leftarrow (\textit{time}\leftarrow\textit{time}+1)$; \label{algo:STCCs-visit:l1} \;
    \colorbox{gray!20}{$\textit{$\alpha$lowlink}[v]\leftarrow \textit{time}$; \label{algo:STCCs-visit:l2}} \;
    \colorbox{gray!20}{$\textit{cSt}.\textit{push}(v)$; \label{algo:STCCs-visit:l4}} \;
    \colorbox{gray!20}{$\textit{on\_stack}[v]\leftarrow \textit{true}$ \label{algo:STCCs-visit:l5}} \;
    \tcp{Check the in-neighbourhood of $v$}
    \ForEach{$u\in N^{\text{in}}_\A(v)$}{ \label{algo:STCCs-visit:l6}
    	\If{$\textit{open}[u] = +\infty$}{ \label{algo:STCCs-visit:l7}
    		\If{$u\in V_\square$}{ \label{algo:STCCs-visit:l8}
    			$\textit{safe-$\alpha$SCC-visit}(u,\A)$; \label{algo:STCCs-visit:l9} \;
    			\colorbox{gray!20}{$\textit{$\alpha$lowlink}[v]\leftarrow \min(\textit{$\alpha$lowlink}[v], \textit{$\alpha$lowlink}[u])$; \label{algo:STCCs-visit:l10}} \;
    			$\mathcal{D}.\textit{Union}(u, v)$; \label{algo:STCCs-visit:l11} \;
    		}\Else{ 
    		$\textit{low\_ready}[u]\leftarrow \min(\textit{low\_ready}[u], \textit{open}[v])$; \label{algo:STCCs-visit:l13} \;
    		$\textit{cnt}[u]\leftarrow \textit{cnt}[u]-1$; \label{algo:STCCs-visit:l14} \;
    			\If{$\textit{cnt}[u]=0$}{ \label{algo:STCCs-visit:l15}
    		$\textit{low\_}v\leftarrow $ the unique $x$ such that $\textit{open}[x]=\textit{low\_ready}[u]$; \label{algo:STCCs-visit:l16} \;
    		$\gamma\leftarrow \mathcal{D}.\textit{find}(\textit{low\_}v)$; \label{algo:STCCs-visit:l17} \;
    				\If{\colorbox{gray!20}{$\textit{on\_stack}[\gamma] = \textit{true}$ \label{algo:STCCs-visit:l18} }}{
    					\colorbox{gray!20}{$\textit{rSt}[\gamma].\textit{push}(u)$; \label{algo:STCCs-visit:l19}} \;
    				}
    			}
    		}
    	}\ElseIf{\colorbox{gray!20}{$\textit{on\_stack}[u] = \textit{true}$}}{ \label{algo:STCCs-visit:l20}
    		\colorbox{gray!20}{$\textit{$\alpha$lowlink}[v]\leftarrow \min(\textit{$\alpha$lowlink}[v], \textit{open}[u])$; \label{algo:STCCs-visit:l21}} \;
    	}
    }
    \tcp{Check the ready-stack of $v$, \ie $\textit{rSt}[v]$}
    \While{$\textit{rSt}[v]\neq\emptyset$}{ \label{algo:STCCs-visit:l22}
    	$u\leftarrow \textit{rSt}[v].\textit{pop()}$; \tcp{$u\in V_{\ocircle}$} \label{algo:STCCs-visit:l23}
    	\If{$\forall{u'\in N^{\text{out}}_{\A}(u)} \textit{on\_stack}[u']=\textit{true}$}{
    		$\textit{safe-$\alpha$SCC-visit}(u, \A)$; \label{algo:STCCs-visit:l25} \;
    		$\textit{$\alpha$lowlink}[v]\leftarrow \min(\textit{$\alpha$lowlink}[v], \textit{$\alpha$lowlink}[u])$; \label{algo:STCCs-visit:l26} \;
    		$\mathcal{D}.\textit{union}(u, v)$; \label{algo:STCCs-visit:l27} \;
    	}
    }
	$\textit{close}[v]\leftarrow (\textit{time}\leftarrow \textit{time}+1)$;\label{algo:STCCs-visit:l27bis}\;
    \tcp{\colorbox{gray!20}{Check for a new safe-$\alpha$SCC}}
    \If{\colorbox{gray!20}{$\textit{$\alpha$lowlink}[v] = \textit{open}[v]$}}{ \label{algo:STCCs-visit:l28}
    	\colorbox{gray!20}{$\C\leftarrow\emptyset$;}\;\label{algo:STCCs-visit:l29}
    	\Repeat{\colorbox{gray!20}{$u=v$}}{ \label{algo:STCCs-visit:l30}
    		\colorbox{gray!20}{$u\leftarrow \textit{cSt}.pop()$;} \label{algo:STCCs-visit:l31} \;
    		\colorbox{gray!20}{$\textit{on\_stack}[u] \leftarrow \textit{false}$;} \label{algo:STCCs-visit:l32} \;
    		\colorbox{gray!20}{\textit{add} $u$ to $\C$;} \label{algo:STCCs-visit:l33}
    	}
    	\colorbox{gray!20}{$\textit{output}(\C)$;} \label{algo:STCCs-visit:l34} \;
    }
    }
    \end{procedurealgo} 

\newpage
    The idea for computing the $\alpha$lowlinks being that to keep an eye just on the indices
    coming from active frond-arcs and cross-arcs, \ie to pick the minimum $\textit{$\alpha$lowlink}$ that can be found in the following neighbours of the currently visited $v$:
    \[
    N^{\text{in}}_{\A}[\textit{cSt}](v)\doteq
    \big\{u\in N^{\text{in}}_\A(v) \mid
    	u\in \textit{cSt} \text{ when line~\ref{algo:STCCs-visit:l22} of }
    	\textit{safe-$\alpha$SCC-visit}(v, \A)\text{ (Proc.~\ref{algo:STCCs-visit}) is executed} \big\},
    \]
    or from the recursive children of the currently visited vertex $v$, \ie picking the minimum:
    \[\min\{\textit{$\alpha$lowlink}(c)\mid c \text{ is a child of } v \text{ in } \J_\A\}.\]

    In order to identify the components, $\textit{safe-$\alpha$SCC-visit()}$
    tests whether $\textit{$\alpha$lowlink}[v] = \textit{open}[v]$ (this is reminiscent to the SCC algorithm~\cite{Tar72}).
    If that's the case a brand new safe-$\alpha$SCC $\C$ is identified;	thus some vertices $u$ will be repeatedly
    removed from $\textit{cSt}$ and added to $\C$, until $u=v$ ($v$ comprised).

    However, in order for this test to be sound and complete,
      we have to overcome the issues observed before in Examples~\ref{ex:safe-aSCC1}~and~\ref{ex:safe-aSCC}.
    As mentioned in (r3) above, the proposed solution is conceptually simple.
    Soon after that the whole in-neighbourhood of any $v\in V$	has been visited by $\textit{safe-$\alpha$SCC-visit(v, $\A$)}$,
    a circled vertex $u\in \textit{rSt}[v]$ is visited with a recursive call
    (and thus attached to $\J_\A$ as a child of $v$) if and only if \emph{all} of
    its out-neighbours are still on the component stack $\textit{cSt}$ (see lines~\ref{algo:STCCs-visit:l22}-26 of Proc.~\ref{algo:STCCs-visit}, particularly, line~23); otherwise,
    	$u$ is simply discarded and becomes a singleton component at the end of the search.
	
	Intuitively, this works because if some of the neighbours of $u$ is no longer on the stack at that point,
    then (by reasoning inductively) it has already been detached into another component that has been fully identified already,
    	so it would not be possible to guarantee safe-$\alpha$reachability from $u\in V_\ocircle$ to the parent $\pi(u)=v$
    		 within the safe-$\alpha$SCC of $v$ that is currently under formation.
    Along the lines of this intuitive observation, soundness and completeness is formally established in the proofs (see Appendix-A)
      of the forthcoming Propositions~\ref{prop:alowlink_correct}~and~\ref{prop:STCC:output}.

    \textbf{Remark.}
    Notice that with (r3), the $\alpha$jungle underlying safe-$\alpha$SCC (Algo.~\ref{algo:STCC})
    might be different w.r.t. the $\alpha$jungle $\J_\A$ that would have been built by running $\alpha$DFS:
    like if some of the $\alpha$palm-trees of $\J_\A$ were pruned and partitioned into subtrees,
    where the cutting points are precisely those arcs $(u,v)\in A_{\pi}$ on circled vertices $u\in \textit{rSt}[v]\cap V_\ocircle$ that can no longer join $\J_\A$ because at that point $u'\not\in\textit{rSt}[v]$ for some $u'\in N^{\text{out}}_\A(u)$).
    However, a moment's reflection reveals that this is just a minor structural refinement of $\J_\A$, the
    resulting graph structure still satisfies the foundamental properties of an $\alpha$jungle given in Definitions~\ref{def:tr-pt}~and~\ref{def:tr-jungle}.
    The only partial exception being property ($\alpha$jn-4), now there might be circled vertices $u$ that can no longer join
  	$\J_\A$ even if all out-neighbours belong to the same $\alpha$palm-tree (\cfr vertex $g$ in Example~\ref{ex:safe-aSCC1} and vertex $f$ in Example~\ref{ex:safe-aSCC}) -- but this property would be still satisfied if only we imagine that, as soon as a safe-$\alpha$SCC is identified,
  	the corresponding subtree detaches from the maximal tree to which it belongs.
    With this in mind the resulting graph structure is really an $\alpha$jungle, so we will continue to denote it by $\J_\A$ as the local context of $\textit{safe-$\alpha$SCC()}$ (Algo.~\ref{algo:STCC}) supersedes possible confusion.
    \qed

    Let us now provide some more implementation details of $\textit{safe-$\alpha$SCC-visit(v,$\A$)}$ (Proc.~\ref{algo:STCCs-visit}).

    At the very beginning, the vertex $v$ which is currently being visited is pushed on top of the component stack $\textit{cSt}$ and
     flagged $\textit{on\_stack}[v]\leftarrow\textit{true}$ (see lines~\ref{algo:STCCs-visit:l4}-\ref{algo:STCCs-visit:l5} of Proc.~\ref{algo:STCCs-visit}).

    Then, whenever some in-neighbour $u\in N^{\text{in}}_\A(v)$ is visited, and as soon as
    	the child recursive call $\textit{safe-$\alpha$SCC-visit}(u, \A)$ returns, the $\textit{$\alpha$lowlink}$ is updated as follows:
    \[\textit{$\alpha$lowlink}[v]\leftarrow \min(\textit{$\alpha$lowlink}[v], \textit{$\alpha$lowlink}[u])
      \;\;\;\;\;\; \text{(see lines~9~and~25 of Proc.~\ref{algo:STCCs-visit} )}\]
    besides executing a $\mathcal{D}.\textit{Union}(u, v)$ to update the disjoint-set forest as before in \textit{dsf-$\alpha$DFS-visit()}.

    Next, when exploring the in-neighbourhood $N^{\text{in}}_\A(v)$ aiming at visiting unexplored vertices:
    if an in-neighbour $u\in N^{\text{in}}_\A(v)\cap V_\ocircle$ is still unvisited (\ie if $\textit{open}[u] = +\infty$),
    and it happens that $\textit{cnt}[u]=0$, then $u$ is pushed to the ready stack $\textit{rSt}[\gamma]$
      if and only if $\textit{on\_stack}[\gamma] = \textit{true}$
    (we now have the additional stack $\textit{cSt}$ flagged by $\textit{on\_stack}$, and indeed we can use it to check whether $\gamma$ is still active);
    else, if $u\in N^{\text{in}}_\A(v)$ has been already visited (\ie if $\textit{open}[u]\neq +\infty$), and if $\textit{on\_stack}[u]=\textit{true}$,
    then the $\textit{$\alpha$lowlink}$ of $v$ is updated as follows:
    \[\textit{$\alpha$lowlink}[v]\leftarrow \min(\textit{$\alpha$lowlink}[v], \textit{open}[u])
      \;\;\;\;\;\; \text{(see lines~\ref{algo:STCCs-visit:l20}-\ref{algo:STCCs-visit:l21} of Proc.~\ref{algo:STCCs-visit} )}\]

   Soon after that the in-neighbourhood of $v$ has been visited (see lines~\ref{algo:STCCs-visit:l22}-\ref{algo:STCCs-visit:l27} of Proc.~\ref{algo:STCCs-visit}),
    $\textit{rSt}[v]$ is managed almost as it was in \textit{dsf-$\alpha$DFS-visit()};
    the only difference being that, as already mentioned,
    	a circled vertex $u\in \textit{rSt}[v]$ is visited with a recursive call if and only if
     all of its out-neighbours are still on $\textit{cSt}$.
     Of course  when such an $u$ gets visited the disjoint-set forest is updated as usual by $\mathcal{D}.\textit{Union}(u, v)$,
     but now also the $\textit{$\alpha$lowlink}$ of $v$ is updated by taking the minimum, \ie
    $\textit{$\alpha$lowlink}[v]\leftarrow \min(\textit{$\alpha$lowlink}[v], \textit{$\alpha$lowlink}[u])$.

    This concludes the description of $\textit{safe-$\alpha$SCC-visit()}$ (Proc.~\ref{algo:STCCs-visit}) and that of Algorithm~\ref{algo:STCC}.

    Let us revise Examples~\ref{ex:safe-aSCC1}~and~\ref{ex:safe-aSCC} to illustrate how safe-$\alpha$SCC (Algo.~\ref{algo:STCC}) runs on the $\alpha$graphs $\A_1$ and $\A_2$; the resulting $\alpha$jungles are shown in \figref{fig:ex:safe-aSCC:agraphrev} and \figref{fig:ex:safe-aSCC:algorunrev} (respectively).

    Concerning Example~\ref{ex:safe-aSCC1}, \figref{fig:ex:safe-aSCC:agraphrev} shows that all vertices in the safe-$\alpha$SCC $\C_1=\{a,h,c\}$ have $\alpha$lowlink equal to $\textit{open}[a]=1$ and all other vertices are singletons. Now $\C_1$ induces a subtree in the $\alpha$palm-tree.

    Similarly for Example~\ref{ex:safe-aSCC}, \figref{fig:ex:safe-aSCC:algorunrev} shows that all vertices in the
    safe-$\alpha$SCC $\C_1=\{c,e\}$ have $\alpha$lowlink equal to $\textit{open}[c]=2$,
      and all vertices in $\C_2=\{g,h\}$ have an $\alpha$lowlink equal to $\textit{open}[g]=6$.
    All of the remaining vertices are singleton safe-$\alpha$SCCs.
    Again both $\C_1=\{c,e\}$ and $\C_2=\{g,h\}$ induce a subtree, rooted at $c$ and $g$ respectively.

    \begin{figure}[t!]
    \centering
    \begin{minipage}{\columnwidth}
    \centering
    \subfloat[The revised $\alpha$jungle $\mathcal{\J}_{\A_1}$ traced by safe-$\alpha$SCC(), showing timestamps and $\alpha$lowlinks\label{fig:ex:safe-aSCC:agraphrev}]{
    \begin{tikzpicture}[arrows=->, scale=.75]
    	\node[squa, label={[yshift=-.5ex] \tiny $1|14|(1)$}] (tAA) {$a$};
    	\node[squa, below=of tAA, label={[xshift=-3ex, yshift=-.5ex] \tiny $2|9|(2)$}] (tA) {$b$};
    	\node[squa, right=of tA, label={[xshift=3.25ex, yshift=-.5ex] \tiny $10|13|(1)$}, yshift=0ex, xshift=0ex] (tC) {$c$};
    	\node[squa, below=of tA, label={[xshift=2.75ex, yshift=-5ex] \tiny $5|6|(5)$}, yshift=0ex, xshift=0ex] (tB) {$e$};
    	\node[circ, below=of tB, label={[xshift=1.5ex, yshift=-.5ex] \tiny $15|16|(15)$}, yshift=-6ex, xshift=0ex] (tD) {$g$};
    	\node[squa, below right=of tA, label={[xshift=1.5ex, yshift=-.5ex] \tiny $7|8|(7)$}, yshift=0ex, xshift=1ex] (tE) {$f$};
    	\node[squa, left=of tB, label={[xshift=-1ex, yshift=-.5ex] \tiny $3|4|(3)$}, yshift=0ex, xshift=-1ex] (tF) {$d$};
    	\node[squa, below=of tD, label={[xshift=0ex, yshift=-5ex] \tiny $11|12|(1)$}, yshift=0ex, xshift=0ex] (tG) {$h$};

    	\draw[thick] (tA) to [xshift=0ex, yshift=0ex] node[xshift=2ex, yshift=0ex] {\tiny tree} (tAA);
    	\draw[thick] (tB) to [xshift=0ex, yshift=0ex] node[xshift=2ex, yshift=-.5ex] {\tiny tree} (tA);
    	\draw[thick] (tC) to [bend right=40, xshift=0ex, yshift=0ex] node[xshift=2ex, yshift=0ex] {\tiny tree} (tAA);
    	\draw[dotted] (tD) to [bend left=40, xshift=0ex, yshift=0ex] node[below, xshift=0ex, yshift=1.5ex] {} (tF);
    	\draw[dotted] (tD) to [bend right=40, xshift=0ex, yshift=0ex] node[below, xshift=0ex, yshift=1.5ex] {} (tE);
    	\draw[dotted] (tD) to [bend right=0, xshift=0ex, yshift=0ex] node[below, xshift=0ex, yshift=1.5ex] {} (tB);
    	\draw[thick] (tE) to [xshift=0ex, yshift=0ex] node[right,xshift=0ex, yshift=0ex] {\tiny tree} (tA.east);
    	\draw[dashed] (tG) to [xshift=0ex, yshift=0ex] node[left,xshift=2ex, yshift=0ex] {} (tD);
    	\draw[thick] (tF) to [xshift=0ex, yshift=0ex] node[xshift=-2ex, yshift=0ex] {\tiny tree} (tA.west);
    	\draw[dashed] (tAA.west) to [out=170, in=180, xshift=-5ex, yshift=0ex] node[left,xshift=.5ex, yshift=0ex] {\tiny frond} (tG.west);
    	\draw[thick] (tG.east) to [out=0, in=0, xshift=-5ex, yshift=0ex] node[right,xshift=0ex, yshift=0ex] {\tiny tree} (tC.east);
    \end{tikzpicture}
    }
    \qquad
    \subfloat[The revised $\alpha$jungle $\mathcal{\J}_{\A_2}$ traced by safe-$\alpha$SCC(), showing timestamps and $\alpha$lowlinks\label{fig:ex:safe-aSCC:algorunrev}]{
    \begin{tikzpicture}[arrows=->, scale=.75]
    	\node[squa, label={above : \tiny $1|10|(1)$}] (tA) {$a$};
    	\node[squa, below=of tA, label={[xshift=3.75ex, yshift=-2.5ex] \tiny $2|5|(2)$}, yshift=0ex, xshift=-1ex] (tB) {$c$};
    	\node[squa, below=of tB, label={[xshift=3.75ex, yshift=-2.5ex] \tiny $3|4|(2)$}, yshift=2ex, xshift=0ex] (tC) {$e$};
    	\node[circ, below=of tB, label={right: \tiny $15|16|(15)$}, yshift=-7ex, xshift=0ex] (tD) {$f$};
    	\node[squa, below right=of tA, label={[xshift=1ex] \tiny $8|9|(8)$}, yshift=0ex, xshift=1ex] (tE) {$d$};
    	\node[squa, left=of tB, label={[xshift=-.25ex] \tiny $6|7|(6)$}, yshift=0ex, xshift=-1ex] (tF) {$b$};
    	\node[squa, below=of tD, label={right : \tiny $11|14|(11)$}, yshift=0ex, xshift=0ex] (tG) {$g$};
    	\node[squa, below=of tG, label={right : \tiny $12|13|(11)$}, yshift=0ex, xshift=0ex] (tH) {$h$};

    	\draw[thick] (tB) to [xshift=0ex, yshift=0ex] node[xshift=2ex, yshift=-1ex] {\tiny tree} (tA);
    	\draw[dotted] (tD) to [xshift=0ex, yshift=0ex] node[below, xshift=0ex, yshift=1ex] {} (tC);
    	\draw[thick] (tC) to [xshift=0ex, yshift=0ex] node[below, xshift=2ex, yshift=1.25ex] {\tiny tree} (tB);
    	\draw[dashed] (tB) to [bend right=55, xshift=0ex, yshift=0ex] node[below, xshift=-2.25ex, yshift=1.25ex] {\tiny frond} (tC);
    	\draw[dotted] (tD) to [bend left=40, xshift=0ex, yshift=0ex] node[below, xshift=0ex, yshift=1.5ex] {} (tF);
    	\draw[dotted] (tD) to [bend right=40, xshift=0ex, yshift=0ex] node[below, xshift=0ex, yshift=1.5ex] {} (tE);
    	\draw[thick] (tE) to [xshift=0ex, yshift=0ex] node[right,xshift=-.25ex, yshift=0ex] {\tiny tree} (tA.east);
    	\draw[dashed] (tG) to [xshift=0ex, yshift=0ex] node[right,xshift=0ex, yshift=0ex] {} (tD);
    	\draw[thick] (tH) to [xshift=0ex, yshift=0ex] node[right,xshift=-.5ex, yshift=0ex] {\tiny tree} (tG);
      \draw[thick] (tF) to [xshift=0ex, yshift=0ex] node[left,xshift=0ex, yshift=0ex] {\tiny tree} (tA.west);
    	\draw[dashed] (tG) to [bend right=40, xshift=0ex, yshift=0ex] node[left,xshift=.75ex, yshift=0ex] {\tiny frond} (tH);
    	\draw[dashed] (tA.west) to [out=170, in=170, xshift=-8ex, yshift=	0ex] node[left,xshift=1ex, yshift=0ex] {\tiny cross} (tH.west);
    \end{tikzpicture}
    }
    \end{minipage}
    \caption{The $\alpha$jungles of Examples~\ref{ex:safe-aSCC1}~and~\ref{ex:safe-aSCC} as revised by safe-$\alpha$SCC() (Algo.~\ref{algo:STCC})}\label{fig:ex:safe-aSCC}
    \end{figure}

    In summary, safe-$\alpha$SCC() (Algo.~\ref{algo:STCC}) enjoys the following major properties (as proved in Appendix~A).

    \begin{Prop}\label{prop:alowlink_correct}
    Assume \textit{safe-$\alpha$SCC()} (Algo.~\ref{algo:STCC}) runs on a given input $\alpha$graph $\A$,
    and let $\J_\A$ be the corresponding $\alpha$jungle, then the $\textit{$\alpha$lowlink}_{\J_\A}$ indexing is correctly computed as given in Definition~\ref{def:lowlink}.
    \end{Prop}

    \begin{Prop}\label{prop:low-link-root}
    Let $\J_\A$ be an $\alpha$jungle constructed when \textit{safe-$\alpha$SCC()} (Algo.~\ref{algo:STCC}) runs on the $\alpha$graph~$\A$.
    If $\C$ is a safe-$\alpha$SCC of $\A$, then $\C$ induces a subtree in the forest of~$\J_\A$.
    \end{Prop}

    Since any safe-$\alpha$SCC() $\C$ induces a subtree in~$\J_\A$, we can identify the roots of the subtrees.
    \begin{Prop}\label{prop:low-link-root}
    Let $\J_\A$ be an $\alpha$jungle constructed when \textit{safe-$\alpha$SCC()} (Algo.~\ref{algo:STCC}) runs on the $\alpha$graph~$\A$.
    Let $\textit{open}[]:V\rightarrow \N$ be the corresponding timestamp, and let $\textit{$\alpha$lowlink}_{\J_\A}:V\rightarrow \N$ be as in Definition~\ref{def:lowlink}.
    Any vertex $v\in V$ is the root of some safe-$\alpha$SCC of $\A$ if and only if $\textit{$\alpha$lowlink}_{\J_\A}(v)=\textit{open}[v]$.
    \end{Prop}

    As a consequence safe-$\alpha$SCC() (Algo.~\ref{algo:STCC}) is correct, the inductive proof is sketched in Appendix~A. Concerning time complexity, Theorems~\ref{thm:aDFS:complexityRAM}~and~\ref{thm:aDFS:complexityPTM} already imply that it is linear on RAMs and, at least, Ackermann-linear on pointer machines.

    \section{Related and Future Works}\label{ref:relatedfutureworks}
Firstly let us discuss about possible lines of investigation concerining the time complexity of \textit{dsf-$\alpha$DFS()} and \textit{safe-$\alpha$SCC()} on pointer machines.
As already mentioned in Section~\ref{subsect:lca-dsf}, our current upper-bound is Ackermann-linear-time which comes from applying the dsf union-find data structure.
We observe that our proposed usage of the dsf union-find data structure falls within a rather special case of the incremental-tree set-union problem studied in~\cite{GT85}, \ie the special case in which the union operations are always done in a post-ordering, simultaneously with the backtracking of the depth-first search.
This might be amenable \eg to the techniques developed in \cite{BuchsbaumGKRTW08}, where a linear-time algorithm for the off-line LCA problem was offered.
That algorithm does not seem to extend to the incremental-tree set-union problem in it's full generality (\ie where the union operations can arrive incrementally in any order), still, in this case one should investigate about the very special case in which all of the union operations arrive in a post-ordering.

In the neighbourhood of possibly related lines appearing in the literature, of course we find the $\alpha$SCCs and the MECs~\cite{ChaKriHen2014} decompositions which may offer interesting directions of investigation for future works.

  Let $\A$ be an $\alpha$graph on vertex set $V$ and arc set $A$. As already mentioned, if $U\subseteq V$ is a safe-$\alpha$SCC, then $U$ is an $\alpha$sc set,
  thus $U$ is included in some $\alpha$SCC (though it may not correspond to the whole $\alpha$SCC as it may lack of maximality).
  On the other hand, if $U\subseteq V$ is an $\alpha$SCC, then any two vertices $u,v\in U$ are strongly-connected in the original input directed graph $G_\A$,
  thus $U$ is included in some SCC of $G_\A$ (but it may lack of maximality as well). Notice that the converse inclusions do not hold generally.
  Also recall that the $\alpha$SCC decomposition can be found in time $O(|V||A|)$ by computing $\square$-attractors;
 one natural question at this point is whether our proposed theory can possibly help improving the latter time complexity upper bound.
  We leave open that question, and at the same time we observe what follows.
  As a simple variation of our safe-$\alpha$SCC() algorithm, suppose just to drop the r3) rule: \ie
  assume a circled vertex $u\in \textit{rSt}[v]$ joins $\J_\A$ as a child of $v$ anyway if the LCA $\gamma$ exists, as it was for $\alpha$DFS(),
   without checking whether \emph{all} of $u$ out-neighbours are still found on the component stack $\textit{cSt}$.
  The conjecture may be this could be fine to find the $\alpha$SCC decomposition,
  unfortunately it's not difficult to provide counterexamples that this is not enough on its own: it still seems necessary to run what is basically an attractor computation from a big fraction of the vertices, falling in $O(|V||A|)$ running time.

    \begin{figure}[h!]
    \centering
    \begin{tikzpicture}[node distance=1cm]
    \title{The relation between the four notions of strongly-connectedness}
    \node(A4) [rectangle,draw] [label={$\Theta(n+m)$\cite{Tar72}}] {SCC};
    \node(V4) [rectangle,draw] [below right = of A4] [label={[xshift=10ex, yshift=-.5ex] \stackanchor{$O(\min(m^{3/2}, n^2))$~\cite{ChaKriHen2014}}{$\tilde{O}(m)$ (expected time)~\cite{ChaHen2019}}}] {MEC};
    \node(C3) [rectangle,draw] [below left = of A4] [label={[xshift=-1.5ex, yshift=-.5ex]$O(mn)$}] {$\alpha$SCC};
    \node(1) [rectangle,draw] [yshift=-1cm, below = of A4] [label={[yshift=-10ex]
													\stackanchor{$\Theta(n+m)$ on RAMs}{$O(n+m\alpha(m,n))$ on PTMs} }] {safe-$\alpha$SCC};
    \draw[->](V4) -- (A4);
    \draw[->](1) -- (V4);
    \draw[->](1) -- (C3);
    \draw[->](C3) -- (A4);
    \end{tikzpicture}
    \caption{The four notions of strongly-connectedness ordered by set inclusion,
    each one showing the time complexity of the corresponding best currently known decomposition algorithms for $m=|A|,n=|V|$.}\label{fig:fanfour}
    \end{figure}

    Let us now consider the MECs~\cite{ChaKriHen2014} decomposition as a possibly related line of interest for future works.
    Given a directed graph $G = (V, A)$ with a finite set $V$ of vertices, directed arcs $A \subseteq V \times V$
    and a partition $(V_\square, V_P)$ of the vertex set $V$, an end-component $U \subseteq V$ is a set of vertices such that:
    (i) the graph $(U, A \cap (U \times U))$ is strongly-connected; (ii) for all $u \in U\cap V_P$ and all $(u, v) \in A$ we have $v \in U$;
    and (iii) either $|U| \geq 2$, or $U = \{v\}$ and there is a self-loop at $v$ (\ie $(v, v) \in A$).
    Observe that if $U\subseteq V$ is a safe-$\alpha$SCC and $|U|\geq 2$, then $U$ is an end-component according to the above definition.
    Of course the converse is not generally true, since, $U$ may well be an end-component and strongly-connected as a directed graph,
    but Player~$\ocircle$ may possibly prevent Player~$\square$ to visit one particular vertex from some moment in time onwards.
    On the other hand, every maximal end-component is included in some SCC of $G_\A$ (again, it may lack maximality and the converse doesn't hold generally).
    A moment's reflection reveals that $\alpha$SCCs and MECs are generally uncomparable (in the sense that no one implies the other).
    The relationship between the four notions is depicted in~\figref{fig:fanfour}, with their currently known time complexities.

    We leave open whether our proposed theory can possibly help speeding up the efficient algorithms for MECs as devised in~\cite{ChaKriHen2014},
    at least by offering a novel approach for some kernelization (pre-processing).

\section{Conclusion}\label{sect:conclusion}
We expect that the proposed theory and the corresponding linear-time decomposition algorithm
could possibly pave the way for speeding up computations in other problems concerning \eg formal verification and infinite games on graphs.

Future works will likely investigate further on this way.

\bibliographystyle{acm}
\bibliography{biblio}

\appendix
\section{Proof that Safe-$\alpha$SCC (Algo.~\ref{algo:STCC}) is correct}
For the sake of the argument let us recall that, during an execution of $\textit{safe-$\alpha$SCC($\A$)}$ (Algo.~\ref{algo:STCC}),
any vertex $v$ is said \emph{open} when $\textit{safe-$\alpha$SCC-visit}(v,\A)$ (Proc.~\ref{algo:STCCs-visit}) assigns $\textit{open}[v]$ at line~\ref{algo:STCCs-visit:l1},
remaining so until $v$ is \emph{closed} (\ie recursive calls included, until the closing timestamp $\textit{close}[v]$ is finally assigned at line~\ref{algo:STCCs-visit:l27bis}).

In order to prove Proposition~\ref{prop:STCC:output}, which basically asserts that safe-$\alpha$SCC (Algo.~\ref{algo:STCC}) is correct, let us dive into the following technical lemmata by reasoning inductively on the graph structures.

\begin{Def}
Let	$\A$ be an $\alpha$graph on vertex set $V$. Pick any two $u,v\in V$ where $u$ is a descendant of $v$ in the forest of $\J_\A$.
The closed interval $[u,v]_{\J_\A}$ denotes the vertex subset of all the ancestors of $u$ that are also descendants of $v$ in the forest of $\J_\A$
(\ie the extremes $u,v$ are included).
\end{Def}
\begin{Lem}\label{lem:precorrectness1}
Assume that $\textit{safe-$\alpha$SCC()}$ (Algo.~\ref{algo:STCC}) runs on a given $\alpha$graph $\A$ on vertex set $V$.
Let $u\in V$ be a proper descendant of $\gamma\in V$ in the forest of $\J_\A$ such that $u$ is still on the component stack $\textit{cSt}$ when \textit{safe-$\alpha$SCC-visit($\gamma,\A$)} closes $\gamma$ (\ie say at line~\ref{algo:STCCs-visit:l27bis} of $\textit{safe-$\alpha$SCC-visit}(\gamma,\A)$) (Proc.~\ref{algo:STCCs-visit}).

Then there exists a proper ancestor $\gamma'$ of $u$ (possibly, $\gamma'=\gamma$) such that all the vertices in $[u,\gamma']_{\J_\A}$ lie within the same safe-$\alpha$SCC of $\A$, \ie such that $[u,\gamma']_{\J_\A}\subseteq \C_u$.
\end{Lem}

\begin{figure}[h!]
	\centering
\begin{tikzpicture}[scale=.6]
		\node [squa] (0) at (0, 7) {$r$};
		\node [squa] (1) at (-9, 0) {$u$};
		\node [style=none] (2) at (9, 0) {};
		\node [style=none] (3) at (2, 0) {};
		\node [squa] (5) at (-1.5, 5.75) {$\gamma$};
		\node [squa] (6) at (-4.75, 3.25) {$\gamma'$};
		\node [squa] (7) at (-7, 1.5) {$x$};

		\draw [thick] (1.north east) to (7.south west);
		\draw [thick] (7.north east) to (6.south west);
		\draw [thick] (0.south west) to (5.north east);
		\draw [thick] (5.south west) to (6.north east);
		\draw [arrows=->, dashed, bend left=300, looseness=1.25] (6.north west) to node[above,xshift=-.5ex, yshift=0ex] {\tiny frond} (1.north);
		\draw (0.south east) to (3.center);
		\draw (3.center) to (1.east);
\end{tikzpicture}
\caption{An illustration of the \emph{Base Case} in the proof of Lemma~\ref{lem:precorrectness1} (part I)}\label{fig:correctness1}
\end{figure}

\begin{proof}
We argue by fixing $u$ and let $\gamma$ be any \emph{satisfying} ancestor
	(\ie one that satisfies the hypothesis).
The proof goes by strong induction on the order in which the vertices $u$ are closed during the execution of $\textit{safe-$\alpha$SCC}(\A)$,
	let it be $u=v_1, \ldots, v_i, \ldots, v_{|V|}$.

	\emph{Base Case: $u=v_1$.}
	The first closed vertex $u=v_1$ must be a leaf in the forest of $\J_{\A}$.
Since $u$ is still on the component stack $\textit{cSt}$ when
	\textit{safe-$\alpha$SCC-visit($\gamma,\A$)} closes $\gamma$, then,
		by lines~\ref{algo:STCCs-visit:l2},\ref{algo:STCCs-visit:l10},\ref{algo:STCCs-visit:l21},\ref{algo:STCCs-visit:l26},\ref{algo:STCCs-visit:l28}:
\begin{equation}\alpha\textit{lowlink}[u] < \textit{open}[u];\tag{$*$}\end{equation}

Since $\alpha\textit{lowlink}[u] < \textit{open}[u]$, by lines~\ref{algo:STCCs-visit:l2},
	\ref{algo:STCCs-visit:l21} of $\textit{safe-$\alpha$SCC-visit}(v_i, \A)$ (Proc.~\ref{algo:STCCs-visit}), it holds that:
	\begin{itemize}[label={--}]
		\item there is one vertex $\gamma'\neq u$ such that $\textit{open}[\gamma']=\alpha\textit{lowlink}[u]$;
		\item $\gamma'$ must be a proper ancestor of $u$;
		\item thus, $(\gamma', u)$ is a frond arc.
		\end{itemize}

Now, let $x$ be any ancestor of $u$ (possibly $x=u$ or $x=\gamma'$).
We claim that $x$ can't be a circled vertex, so it must be $x\in V_\square$. \figref{fig:correctness1} illustrates the situation.

Indeed, suppose $x\in V_\ocircle$ for the sake of contradiction, consider any out-neighbour
$y\in N^{\text{out}}_\A(x)$ which is not the parent of $x$ in $\J_{\A}$ (notice $y$
	exists because \wwlog $|N^{\text{out}}_\A(x)|\geq 2$ if $x\in V_\ocircle$),
	so $(x,y)$ is just a stalk-arc. By the \emph{$\ocircle$-attraction-rule}, $y$ must have been closed before $x$;
this is absurd as $u$ is the first closed vertex and $u\neq y$. So $x\in V_\square$.

Since any ancestor of $u$ that is also a descendant of $\gamma'$ lies in $V_\square$, and since $(\gamma', u)$ is a frond arc,
	then $u$ and $\gamma'$ together with all the ancestors of $u$ that are also descendants of $\gamma'$
	clearly they form a safe-$\alpha$sc set, so they all lie within the same safe-$\alpha$SCC of $\A$.

	\begin{figure}[h!]
		\centering
	\begin{tikzpicture}[scale=.8]
			\node [squa] (0) at (0, 9) {$r$};
			\node [style=none] (1) at (-11, 0) {};
			\node [style=none] (2) at (2, 0) {};
			\node [style=none] (3) at (-1, 7) {};
			\node [squa] (4) at (-1.2, 7) {$\gamma$};
			\node [squa] (5) at (-2, 5.75) {$\gamma'$};
			\node [squa] (6) at (-3, 4.5) {$u$};
			\node [squa] (7) at (-4, 3) {$\pi(c)$};
			\node [circ] (8) at (-6.5, 1) {$c$};
			\node [squa] (9) at (-6.25, 2.625) {$c'_1$};
			\node [circ] (10) at (-4.75, 1.75) {$c'_2$};
			\node [squa] (11) at (-3, 1.25) {$c'_k$};
			\node [squa] (13) at (-8.5, 0.5) {$x$};
			\node [squa] (15) at (-6.7, 4) {$u'$};

	 		\draw [bend right=20] (0.south west) to node[above,xshift=6ex, yshift=8ex] {$\mathcal{P}_x$} (1.center);
			\draw (0.south east) to (2.center);
			\draw (1.center) to (2.center);
			\draw [thick] (0.south) to (4.north);
			\draw [thick] (4.south) to (5.north east);
			\draw [thick] (5.south west) to node[below,xshift=8ex, yshift=1ex] {$[u,\gamma']_{\J_\A}\subseteq \C_{u}$} (6.north);
			\draw [thick] (6.south west) to (7.north);
			\draw [thick] (7.west) to (9.east);
			\draw [thick] (7.south) to (10.north);
			\draw [thick] (7.south east) to (11.north);
			\draw [thick, bend left, looseness=0.5] (8.north) to (7.south west);
			\draw [arrows=->,dotted,bend left, looseness=0.5] (8.north) to node[above,xshift=-1ex, yshift=-1ex] {\tiny stalk} (9.south);
			\draw [arrows=->,dotted] (8.north east) to node[above,xshift=0ex, yshift=-2ex] {\tiny stalk} (10.west);
			\draw [arrows=->,dotted, bend right, looseness=0.75] (8.east) to node[above,xshift=0ex, yshift=-2ex] {\tiny stalk} (11.west);
			\draw [thick] (5.west) to (15.east);
			\draw [thick] (8.south west) to (13.east);
			\draw [arrows=->, dashed, bend right] (15.south west) to node[above,xshift=1.5ex, yshift=0ex] {\tiny frond or cross} (13.north);
			\draw [arrows=->,thick, dashed,bend left=300, looseness=0.5] (7.north west) to (9.north);
			\draw [arrows=->,thick, dashed,bend left=45, looseness=0.75] (7.east) to node[above,xshift=8ex, yshift=0ex] {$[c',\pi(c)]_{\J_\A}\subseteq \C_{c'}$} (11.north east);
			\draw [arrows=->,thick, dashed,bend left, looseness=1.75] (7.south) to (10.north east);
			\draw [arrows=->, thick, dashed, bend right=300, in=120, out=-70, looseness=.75]
				(5.north west) to node[above,xshift=6.25ex, yshift=6.25ex] {$[u',\gamma']_{\J_\A}\subseteq\C_{u'}$} (15.south east);
	\end{tikzpicture}
	\caption{An illustration of the \emph{Inductive Step} in the proof of Lemma~\ref{lem:precorrectness1}}\label{fig:correctness3}
	\end{figure}

\emph{Inductive Step: $u=v_i$ for some $i>1$.}
We shall leverage on the strong induction hypothesis that the thesis is true for every $v_j$ with $j<i$ and any satisfying ancestor of $v_j$.

\figref{fig:correctness3} helps following the argument.
Firstly notice that $\textit{$\alpha$lowlink}[v_i]$ can be assigned either at line~\ref{algo:STCCs-visit:l2}, \ref{algo:STCCs-visit:l10},
	\ref{algo:STCCs-visit:l21}, \ref{algo:STCCs-visit:l26} of $\textit{safe-$\alpha$SCC-visit}(v_i, \A)$ (Proc.~\ref{algo:STCCs-visit}).
Since $u$ is still on the component stack $\textit{cSt}(\gamma)$ when \textit{safe-$\alpha$SCC-visit($\gamma, \A$)} closes $\gamma$, then,
	by lines~\ref{algo:STCCs-visit:l2},\ref{algo:STCCs-visit:l10},\ref{algo:STCCs-visit:l21},\ref{algo:STCCs-visit:l26},\ref{algo:STCCs-visit:l28}:
\begin{equation}\alpha\textit{lowlink}[u] < \textit{open}[u];\tag{$*$}\end{equation}

 Since $\alpha\textit{lowlink}[u] < \textit{open}[u]$, by lines~\ref{algo:STCCs-visit:l2}, \ref{algo:STCCs-visit:l10},
 	\ref{algo:STCCs-visit:l21}, \ref{algo:STCCs-visit:l26} of $\textit{safe-$\alpha$SCC-visit}(v_i, \A)$ (Proc.~\ref{algo:STCCs-visit}), there must be:
	\begin{itemize}[label={--}]
\item one vertex $u'\neq u$ such that $\textit{open}[u']=\alpha\textit{lowlink}[u]$, so $\textit{open}[u']<\textit{open}[u]$;
\item one descendant $x$ of $u$ (possibly, $x=u$)
such that $\alpha\textit{lowlink}[x] = \textit{open}[u']$ and $(u',x)$ is either a frond or a cross-arc in $\J_\A$.
		\end{itemize}
Thus, $u'$ was still on $\textit{cSt}$ when \textit{safe-$\alpha$SCC-visit($x,\A$)} closed $x$.
Then let $\gamma'$ be the LCA of ${u', u}$ in the forest of $\J_\A$ (possibly, $\gamma'=u'$, but $\gamma'\neq u$).
Also, since $u'$ was still on $\textit{cSt}$ when \textit{safe-$\alpha$SCC-visit($x,\A$)} closed $x$, and since $x$ is a descendant of $u$, then
 the fact that $\textit{open}[u']<\textit{open}[u]$ implies that $u'$ is still on $\textit{cSt}(\gamma')$ when \textit{safe-$\alpha$SCC-visit($\gamma',\A$)} closes $\gamma'$ too. Before ending the proof, let us show two more claims.

\begin{itemize}
	\item[Claim 1] $[u',\gamma']_{\J_\A}\subseteq \C_{u'}$.

If $\gamma'=u'$, the thesis is obvious. So, assume \wwlog $\gamma'\neq u'$.
Then, since $u'$ was still on $\textit{cSt}$ when \textit{safe-$\alpha$SCC-visit($x,\A$)} closed $x$,
and $x$ is a descendant of $u$, then $u'$ must have been closed before $u$.
So the induction hypothesis applies to $u'$ and its ancestor $\gamma'$, thus $[u',\gamma']_{\J_\A}\subseteq \C_{u'}$.

\item[Claim 2] If $c\in [x, \gamma']_{\J_\A}\cap V_\ocircle$, then for every out-neighbour $c'\in N^{\text{out}}_\A(c)$ such
 	that $c'\neq \pi(c)$ (\ie such that $(c,c')$ is just a stalk-arc and not a tree-arc), it holds $[c',\pi(c)]_{\J_\A}\subseteq \C_{c'}$.

	Take any $c'\in N^{\text{out}}_\A(c)$ such that $c'\neq \pi(c)$ and observe that, by the \emph{$\ocircle$-attraction-rule}, $c'$ must have been already closed when $c$ joined $\J_\A$.
	Thus, since $x$ is a child of $c$, then $c'$ must have been closed before $x$ was.
	Since $x$ is a child of $u$, then $c'$ must have been closed before $u$ was.
	Moreover, by the (r3) revision of the \emph{$\ocircle$-attraction-rule} (\ie by lines~\ref{algo:STCCs-visit:l22}-26 of Proc.~\ref{algo:STCCs-visit}, particularly, line~23), $c'$ was still on $\textit{cSt}$ when $c$ joined $\J_\A$,
	so $c'$ is still on $\textit{cSt}$ even when \textit{safe-$\alpha$SCC-visit($\pi(c),\A$)} closes $\pi(c)$.
	Therefore, the induction hypothesis applies to $c'$ with parent $\pi(c)$, so $[c',\pi(c)]_{\J_\A}\subseteq \C_{c'}$.
	\end{itemize}
	We are now in the position to show that $[u,\gamma']_{\J_\A}\subseteq \C_{u}$.

	Recall $(u',x)$ is either a frond or a cross-arc,
		so by Claim~1 $x$ is $\C_{u'}\cup\{x\}$-safe-$\alpha$reachable from $\gamma'$.
		Moreover $x$ is a descendant of $u$, which is a descendant of $\gamma'$.
	Proposition~\ref{prop:safe-reachability} says that every vertex in $[x, \gamma']_{\J_\A}$ is $\mathcal{P}_x$-safe-$\alpha$reachable from $x$, with a strategy that simply goes up along the $\alpha$palm-tree $\mathcal{P}_x$ in which $x$ resides, \ie a strategy that goes from any $c\in [x, \gamma']_{\J_\A}$ to its parent $\pi(c)$.
			Claim~2 guarantees that, even when $c\in V_\ocircle$, $\pi(c)$ lies in $\C_{c'}$ for every possible $c'\in N^{\text{out}}_\A(c)$.

All in, by composition of safe-$\alpha$sc sets (\ie Lemma~\ref{lemma:safe-asc-union}),
$\C_u = \C_{p}$ for every $p\in [x, \gamma']_{\J_\A} \cup [u',\gamma']_{\J_\A}$.

So $[u,\gamma']_{\J_\A}\subseteq \C_{u}$ as expected.
 \end{proof}

\begin{Lem}\label{lem:correctness1}
	Assume that $\textit{safe-$\alpha$SCC()}$ (Algo.~\ref{algo:STCC}) runs on a given $\alpha$graph $\A$ on vertex set $V$.
	Let $u\in V$ be any descendant of $\gamma\in V$ in the forest of $\J_\A$ such that $u$ is still on the component stack $\textit{cSt}$ when \textit{safe-$\alpha$SCC-visit($\gamma,\A$)} closes $\gamma$ (\ie say at line~\ref{algo:STCCs-visit:l27bis} of $\textit{safe-$\alpha$SCC-visit}(\gamma,\A)$) (Proc.~\ref{algo:STCCs-visit}).

	Then $\gamma$ lies within the same safe-$\alpha$SCC of $\A$,
	\ie  $\gamma\in \C_u$.
\end{Lem}
\begin{figure}[h!]
	\centering
\begin{tikzpicture}[scale=.7]
		\node [squa] (0) at (0, 6.5) {$r$};
		\node [style=none] (14) at (-9, -.5) {};
		\node [squa] (1) at (-6.5, .75) {$u$};
		\node [style=none] (2) at (9, -.5) {};
		\node [style=none] (3) at (2, -.5) {};
		\node [squa] (5) at (-.35, 5.65) {$\gamma$};
		\node [squa] (6) at (-4.5, 2.25) {$\gamma'$};
		\node [style=none] (12) at (-5.75, 1.325) {};
		\node [style=none] (13) at (-3.75, 2.85) {};
		\node [squa] (9) at (-2.625, 3.75) {$\gamma''$};
		\node [style=none] (10) at (-1.25, 5.75) {$\iddots$};
		\node [squa] (11) at (-1.5, 4.75) {$\gamma'''$};

\draw [bend right=20] (0.south west) to node[above,xshift=6ex, yshift=8ex] {} (14.north);
		\draw [thick] (0.south) to (5.north);
		\draw [thick] (5.south west) to (11.north east);
		\draw [thick] (11.south west) to (9.north east);
		\draw [thick] (9.south west) to (6.north east);
		\draw [thick] (6.south west) to (1.north east);
		\draw [arrows=->, thick, dashed, bend left=300, looseness=1.25] (6.west) to node[below,xshift=11ex, yshift=-1.5ex] {$[u,\gamma']_{\J_\A}\subseteq \C_u$} (1.north);
		\draw (0.south east) to (3.center);
		\draw (3.center) to (14.east);
		\draw [arrows=->, thick, dashed, bend left=315, looseness=1.5] (9.west) to node[above,xshift=-.5ex, yshift=0ex] {} (12.north);
		\draw [arrows=->, thick, dashed, bend right=45, looseness=1.50] (11.west) to node[above,xshift=-.5ex, yshift=0ex] {} (13.north);
	\end{tikzpicture}
\caption{An illustration of the \emph{Inductive Step} in the proof of Lemma~\ref{lem:correctness1}}\label{fig:correctness4}
\end{figure}
\begin{proof}
	The proof goes by strong induction on $\textit{open}[u]$, \ie induction on the order in which the vertices $u\in V$ are opened during the execution of $\textit{safe-$\alpha$SCC}(\A)$, let it be $u=v_1, \ldots, v_i, \ldots, v_{|V|}$. Moreover the inductive step relies on Lemma~\ref{lem:precorrectness1}.

	\emph{Base Case: $u=v_1$.} In this case $u=\gamma$ and the thesis is obvious.

	\emph{Inductive Step: $u=v_i$ for some $i>1$.}
	We shall leverage on the strong induction hypothesis that the thesis is true for every $v_j$ with $j<i$ and any satisfying ancestor of $v_j$. \wwlog let $\gamma$ be a satisfying \emph{proper} ancestor of $u$.
	By Lemma~\ref{lem:precorrectness1} there is a proper ancestor $\gamma'$ of $u$ (possibly, $\gamma'=\gamma$) such that all the vertices in $[u,\gamma']_{\J_\A}$ lie within the same safe-$\alpha$SCC of $\A$,
		\ie such that $[u,\gamma']_{\J_\A}\subseteq \C_u$.
	If $\gamma'$ is an ancestor of $\gamma$, this already implies the thesis as $\gamma\in [u,\gamma']_{\J_\A}$.
	Otherwise, $\gamma'$ is a proper descendant of $\gamma$.
	Of course $\textit{open}[\gamma']<\textit{open}[u]$, so the strong induction hypothesis applies on $\gamma'$ with
	satisfying ancestor $\gamma$, \ie $\gamma\in\C_{\gamma'}$.
	By composition of safe-$\alpha$sc sets (\ie Lemma~\ref{lemma:safe-asc-union}) it holds that $\C_u=\C_{\gamma'}=\C_\gamma$, so $\gamma\in\C_{u}$.
\end{proof}

\begin{Lem}\label{lem:correctness2}
	Assume that $\textit{safe-$\alpha$SCC()}$ (Algo.~\ref{algo:STCC}) runs on a given $\alpha$graph $\A$ on vertex set $V$.
	Let $u\in V$ be a descendant of $\gamma\in V$ in the forest of $\J_\A$
	 lying in the same safe-$\alpha$SCC of $u$, \ie such that $\gamma\in \C_u$.
Then, $u$ is still on the component stack $\textit{cSt}$ when $\textit{safe-$\alpha$SCC-visit($\gamma, \A$)}$ closes $\gamma$.
\end{Lem}
\begin{proof}
Firstly, we can assume \wwlog that $\textit{open}[\gamma]$ is the smallest possible index in $\C_u$, \ie
\[\gamma=\arg\min_{v\in \C_u} \textit{open}[v].\]
Indeed, thanks to the structural connectivity properties of $\J_\A$ any
	cross-arc $(u,v)$ between two distinct $\alpha$palm-trees $\mathcal{P}_u,\mathcal{P}_v$ always goes \emph{forward in time},
	\ie satisfying \[\textit{open}[u]<\textit{close}[u]<\textit{open}[v]<\textit{close}[v]\];
	stated otherwise, cross-arcs can't go \emph{backward in time}, from an $\alpha$palm-tree $\mathcal{P}_v$ back to another $\alpha$palm-tree $\mathcal{P}_u$ which had been discovered and closed earlier (\cfr Definitions~\ref{def:tr-pt}~and~\ref{def:tr-jungle}, particularly, properties \emph{$\alpha$pt-2, $\alpha$pt-4 and $\alpha$jn-3}).

	Thus, if $\gamma\in \C_u$, then $\gamma$ belongs to the \emph{same} $\alpha$palm-tree in which $u$ resides; indeed, \emph{any} arc that crosses two distinct $\alpha$palm-trees is a cross-arc, so,
		if $\gamma$ and $u$ were in two distinct $\alpha$palm-trees there would have been at least one cross-arc going backward in time, which is not possible.
Any other possible ancestor $\gamma'\neq \gamma$ of $u$ satisfying $\gamma'\in \C_u$ must
	be a proper descendant of $\gamma$ if $\textit{open}[\gamma]$ is minimum;
so, proving the thesis w.r.t. the smallest $\gamma$ subsumes proving it for any other satsfying $\gamma'$.

The proof proceeds by strong induction on $\textit{open}[u]$, for $\textit{open}[u]\geq \textit{open}[\gamma]$.

		In the \emph{Base Case, $u=\gamma$,} the thesis is trivial.

\begin{figure}[h!]
	\centering
\begin{tikzpicture}[scale=.75]
		\node [squa] (0) at (0, 9) {$r$};
		\node [style=none] (1) at (-10, 0) {};
		\node [style=none] (2) at (2, 0) {};
		\node [squa] (4) at (-1, 6.5) {$\gamma$};
		\node [circ] (5) at (-2, 3) {$u$};
		\node [squa] (7) at (-4.25, 4.25) {$u'$};
		\node [squa] (8) at (-4, 2) {$x$};

		\draw [bend right=15] (0.south west) to (1.center);
		\draw (0.south east) to (2.center);
		\draw (1.center) to (2.center);
		\draw [thick] (0.south) to (4.north);
		\draw [thick] (4.south) to (5.north);
		\draw [thick] (5.south west) to (8.east);
		\draw [thick] (7.north) to (4.west);
		\draw [arrows=->, dotted, bend right=45, looseness=1.25] (7.west) to node[above,xshift=.85ex, yshift=0ex] {\tiny frond or cross} (8.west);
\end{tikzpicture}
	\caption{An illustration of the \emph{Inductive Step} in the proof of Lemma~\ref{lem:correctness2}}\label{fig:correctnessLemma6}
\end{figure}

\emph{Inductive Step:} $\textit{open}[u]>\textit{open}[\gamma]$.
Because of the structural connectivity properties of $\J_\A$
(\cfr Definitions~\ref{def:tr-pt}~and~\ref{def:tr-jungle}, particularly, properties \emph{$\alpha$pt-2, $\alpha$pt-4} and \emph{$\alpha$jn-3}),
as already mentioned all the vertices in $\C_u$ belong to the \emph{same} $\alpha$palm-tree in which $u$ resides, moreover,
since the frond and cross-arcs are the only types of arcs $(a,b)$ where $\textit{open}[a]<\textit{open}[b]$, then
along any of those paths that start at $\gamma$ and reach $u$ without ever leaving $\C_u$
(\ie any of those paths thanks to which $u$ is $\C_u$-safe-$\alpha$reachable from $\gamma$),
at some point there must be a pair of vertices $u',x\in \C_u$ such that the following hold:
\begin{itemize}[label={--}]
	\item $(u', x)$ is either a frond or a cross-arc (where $\textit{open}[u']<\textit{open}[x]$);
	\item $x$ is a descendant of $u$ in the	forest of $\J_\A$ (possibly $x=u$, but $x\neq u'$);
	\item thus, $\textit{open}[u']<\textit{open}[u]$.
\end{itemize}

\figref{fig:correctnessLemma6} illustrates the situation.

Now, since $\textit{open}[\gamma]$ is minimum in $\C_u$, then $u'$ is still a descendant of $\gamma$ (possibly, $u'=\gamma$).
Also notice that, since $u'\in \C_u$, then $\C_{u'}=\C_u$, so $\gamma\in\C_{u'}$.

Since $\textit{open}[u']<\textit{open}[u]$ and $u'$ is a descendant of $\gamma$ such that $\gamma\in\C_{u'}$, the induction hypothesis applies to $u'$ with satisfying ancestor $\gamma$,
	then	$u'$ is still on the component stack $\textit{cSt}$ when $\textit{safe-$\alpha$SCC-visit($\gamma, \A$)}$ closes $\gamma$.
Thus, since $x$ is a descendant of $u$ and $\textit{open}[u']<\textit{open}[u]$,
$u'$ is already on the component stack $\textit{cSt}$ when $\textit{safe-$\alpha$SCC-visit($x, \A$)}$ closes $x$.
Therefore, by lines~19,\ref{algo:STCCs-visit:l21} of $\textit{safe-$\alpha$SCC-visit($x, \A$)}$, \[\alpha\textit{lowlink}[x]\leq \textit{open}[u'].\]
This means that $x$ stays on the component stack $\textit{cSt}$ as long as $u'$ stays there.
Since $x$ is a descendant of $u$, also $u$ stays on $\textit{cSt}$ as long as $u'$ stays there.
Then, since $u'$ is on $\textit{cSt}$ when $\textit{safe-$\alpha$SCC-visit($\gamma, \A$)}$ closes $\gamma$, also $u$ is still there on $\textit{cSt}$ at that time.
\end{proof}

We are finally in the position to close the proof of correcteness.
Since it is clear that every vertex $v\in V$ is eventually outputted
by \textit{safe-$\alpha$SCC()} (Algo.~\ref{algo:STCC}) at lines~\ref{algo:STCCs-visit:l28}-\ref{algo:STCCs-visit:l34},
as part of some $\C\subseteq V$, it is enough to show what follows.
\begin{Prop}\label{prop:STCC:output}
If \textit{safe-$\alpha$SCC()} (Algo.~\ref{algo:STCC}) runs on a given $\alpha$graph $\A$,
	and $\textit{safe-$\alpha$SCC-visit}()$ (Proc.~\ref{algo:STCCs-visit}) outputs some subset of vertices $C\subseteq V$,
		then $C$ is a safe-$\alpha$SCC of $\A$.
\end{Prop}
\begin{proof}
Assume that	$\textit{safe-$\alpha$SCC-visit}(\gamma, \A)$ (Proc.~\ref{algo:STCCs-visit})
outputs some subset of vertices $C\subseteq V$,	for some vertex $\gamma\in C$ such that
$\textit{$\alpha$lowlink}[\gamma]=\textit{open}[\gamma]$ holds at line~\ref{algo:STCCs-visit:l28}.
So, $C=\textit{cSt}(\gamma)$. By Lemma~\ref{lem:correctness1}, then $\gamma\in \C_u$ for every $u\in \textit{cSt}(\gamma)$.
So, $\textit{cSt}(\gamma)\subseteq \C_{\gamma}$.

Now we claim that $\C_{\gamma}\subseteq \textit{cSt}(\gamma)$. Pick $\gamma'\in \C_{\gamma}$.
It is not possible for $\gamma'$ to be an ancestor of $\gamma$, because by Lemma~\ref{lem:correctness2} it would be $\gamma\in \textit{cSt}(\gamma')$
when $\textit{safe-$\alpha$SCC-visit($\gamma', \A$)}$ closes $\gamma'$, against $\textit{$\alpha$lowlink}[\gamma]=\textit{open}[\gamma]$.
Thus, $\gamma'$ is either uncomparable (\ie neither an ancestor nor a descendant) or a proper descendant of $\gamma$.

However it is not possible for $\gamma'$ to be uncomparable with $\gamma$.
Indeed, since $\gamma'\in \C_{\gamma}$, along any of those paths that start at one of the two $\gamma,\gamma'$
	and reach the other one (\ie any of those paths thanks to which $\gamma,\gamma'$ are safe-$\alpha$reachable from one another),
at some point there must be at least one cross-arc $(a,b)$ going backward in time (\ie such that $\textit{open}[a]>\textit{open}[b]$),
but this would contradict the structural connectivity properties of $\J_\A$ (\cfr Definitions~\ref{def:tr-pt}~and~\ref{def:tr-jungle},
particularly, properties \emph{$\alpha$pt-2, $\alpha$pt-4} and \emph{$\alpha$jn-3}).

So, $\gamma'$ must be a proper descendant of $\gamma$ in the forest of $\J_\A$.
Then by Lemma~\ref{lem:correctness2} it holds $\gamma'\in \textit{cSt}(\gamma)$.

Therefore, $\textit{cSt}(\gamma)=\C_{\gamma}$.
\end{proof}

To end, for the sake of completeness, it is shown Proposition~\ref{prop:alowlink_correct}, \ie that \textit{safe-$\alpha$SCC()} (Algo.~\ref{algo:STCC}) computes the very same  $\alpha$lowlinks that are given in Definition~\ref{def:lowlink}.
This follows from Lemma~\ref{lem:correctness1}~and~\ref{lem:correctness2}.
\begin{proof}[Proof of Proposition~\ref{prop:alowlink_correct}]
	For the sake of the argument let us denote $\textit{$\alpha$lowlink}[]$ (\ie with squared brackets) the indices computed by
	\textit{safe-$\alpha$SCC()} (Algo.~\ref{algo:STCC}),
	whereas $\textit{$\alpha$lowlink}()$ denotes the indices given by Definition~\ref{def:lowlink}.
	Thus we aim at showing that for every $v\in V$, it holds  $\textit{$\alpha$lowlink}[v] = \textit{$\alpha$lowlink}(v)$.

	The proof goes by induction on the order in which the vertices are \emph{closed} during the execution of $\textit{safe-$\alpha$SCC}(\A)$, let it be $(v_1, \ldots, v_i, \ldots, v_{|V|})$.

For every $v\in V$, to ease the argument, let us define the following in-neighbourhood	by considering the state of the component stack $\textit{cSt}$ when line~\ref{algo:STCCs-visit:l22} of $\textit{safe-$\alpha$SCC-visit}(v, \A)$ is executed:
\[
N^{\text{in}}_{\A}[\textit{cSt}](v)\doteq
\Big\{u\in N^{\text{in}}_\A(v) \mid
	u\in \textit{cSt} \text{ when line~\ref{algo:STCCs-visit:l22} of }
	\textit{safe-$\alpha$SCC-visit}(v, \A)\text{ (Proc.~\ref{algo:STCCs-visit}) is executed} \Big\}.
\]

	It's also rational to define for every $v\in V$:
		 \[N^{\text{in}}_{\A}[\text{LCA}](v) \doteq \{u\in N^{\text{in}}_\A(v)\cap V_\square \mid
											\text{the LCA } \gamma \text{ of } \{u,v\} \text{ in } \J_\A \text{ exists and } \gamma\in \C_u \}.\]

	\emph{Base Case: $i=1$.} Notice that the first closed vertex $v_1$ must be a leaf in the forest of $\J_{\A}$.
	In this case, $\textit{$\alpha$lowlink}[v_1]$ can be assigned
		only at line~\ref{algo:STCCs-visit:l21} of $\textit{safe-$\alpha$SCC-visit}(v_1, \A)$. So, the following holds:
	\[
		\textit{$\alpha$lowlink}[v_1]=\min\{\textit{open}[v_1]\}\cup\{\textit{open}[u]\mid u\in N^{\text{in}}_{\A}[\textit{cSt}](v_1)\}. \tag{eq. $1$}
	\]
	Since $v_1$ is the \emph{first} closed leaf,
		\[
				N^{\text{in}}_{\A}[\textit{cSt}](v_1) = \big\{ u\in N^{\text{in}}_\A(v_1) \mid
					u \text{ is an ancestor of } v_1 \text{ in } \J_{\A}\big\}. \tag{eq. $2$}
			\]
On the other hand, since $v_1$ is a leaf in $\J_{\A}$ and by Definition~\ref{def:lowlink}, a moment's reflection reveals:
  \[
		\textit{$\alpha$lowlink}(v_1) = \min\big\{\textit{open}[v_1]\}\cup
			\{ \textit{open}[u]\mid u\in N^{\text{in}}_{\A}[\text{LCA}](v_1) \big\}. \tag{eq. $3$}
	\]

Since $v_1$ is the first closed leaf, by (\emph{eq. $2$}), (\emph{eq. $3$}) and Lemma~\ref{lem:correctness1},
	$N^{\text{in}}_{\A}[\text{LCA}](v_1) = N^{\text{in}}_{\A}[\textit{cSt}](v_1)$.

Therefore, by (\emph{eq. $1$}) and (\emph{eq. $3$}), $\textit{$\alpha$lowlink}[v_1]=\textit{$\alpha$lowlink}(v_1)$.

This concludes the proof of the base case.

	\emph{Inductive Step: $i>1$.} In this case, $\textit{$\alpha$lowlink}[v_i]$ can be assigned either at line~\ref{algo:STCCs-visit:l2}, \ref{algo:STCCs-visit:l10},
		\ref{algo:STCCs-visit:l21}, \ref{algo:STCCs-visit:l26} of $\textit{safe-$\alpha$SCC-visit}(v_i, \A)$ (Proc.~\ref{algo:STCCs-visit}).
			A moment's reflection reveals that the following holds:
	\[
		\textit{$\alpha$lowlink}[v_i] = \min \big\{\textit{open}[v_i]\big\} \cup
				\big\{\textit{open}[u]\mid u\in N^{\text{in}}_{\A}[\textit{cSt}](v_i)\big\}
					\cup \big\{\textit{$\alpha$lowlink}[u]\mid u \text{ is a child of } v_i \text{ in } \J_{\A}\big\}.
	\]
 On the other side, by Definition~\ref{def:lowlink}, Definitions~\ref{def:tr-pt}~and~\ref{def:tr-jungle},
 	one moment's reflection reveals that:
 \begin{align*}
	 \textit{$\alpha$lowlink}(v_i)=\min\big\{\textit{open}[v_i]\big\} & \cup \big\{\textit{open}[u]\mid u\in N^{\text{in}}_{\A}[\text{LCA}](v_i)\big\} \\
	 																									& \cup \big\{\textit{$\alpha$lowlink}(u)\mid u \text{ is a child of } v_i \text{ in } \J_{\A}\big\}.
 \end{align*}
 If $u$ is a child of $v_i$ in $\J_{\A}$, then $u$ is closed before $v_i$.
 By induction hypothesis, $\textit{$\alpha$lowlink}[u] = \textit{$\alpha$lowlink}(u)$
 for every child $u$ of $v_i$ in $\J_{\A}$ that is considered either at
 	line~\ref{algo:STCCs-visit:l10}~or~\ref{algo:STCCs-visit:l26} of $\textit{safe-$\alpha$SCC-visit}(v_i, \A)$.

 To finish the proof, it is sufficient to show that $N^{\text{in}}_{\A}[\textit{cSt}](v_i)= N^{\text{in}}_{\A}[\text{LCA}](v_i)$.
\begin{itemize}

\item $N^{\text{in}}_{\A}[\textit{cSt}](v_i)\subseteq N^{\text{in}}_{\A}[\text{LCA}](v_i)$. Indeed,

let $u\in N^{\text{in}}_{\A}[\textit{cSt}](v_i)$. Then, $u$ and $v_i$ lie within the same $\alpha$palm-tree in $\J_{\A}$:
infact $\textit{cSt}$ is completely emptied as soon as the root of an $\alpha$palm-tree is closed,
thus the stack $\textit{cSt}$ can't contain vertices from two distinct maximal $\alpha$palm-trees.
Thus, the LCA $\gamma$ of $\{u,v_i\}$ in $\J_{\A}$ exists. Since $u\in \textit{cSt}$ when $v_i$ is being visited,
 $u\in \textit{cSt}(\gamma)$ when \textit{safe-$\alpha$SCC-visit($\gamma, \A$)} closes $\gamma$.
By Lemma~\ref{lem:correctness1}, $\gamma\in\C_u$. So, $u\in N^{\text{in}}_{\A}[\text{LCA}](v_i)$.

\item $ N^{\text{in}}_{\A}[\text{LCA}](v_i)\subseteq N^{\text{in}}_{\A}[\textit{cSt}](v_i)$. Indeed,

let $u\in  N^{\text{in}}_{\A}[\text{LCA}](v_i)$, and let $\gamma\in \C_u$ be the LCA of $\{u,v_i\}$ in $\J_{\A}$.
By Lemma~\ref{lem:correctness2}, since $\gamma\in\C_u$, then $u$ is still on the component stack $\textit{cSt}(\gamma)$
when \textit{safe-$\alpha$SCC-visit($\gamma, \A$)} closes $\gamma$. So, $u\in N^{\text{in}}_{\A}[\textit{cSt}](v_i)$.
\end{itemize}
All in, $N^{\text{in}}_{\A}[\textit{cSt}](v_i)= N^{\text{in}}_{\A}[\text{LCA}](v_i)$.
This concludes the \emph{inductive step}.
\end{proof}

\end{document}